\documentclass{article}

\usepackage{amsmath,amsfonts,amssymb,amsthm}
\usepackage{xy}
\usepackage[breaklinks=true]{hyperref}
\usepackage{float}

\usepackage{tikz}
\usetikzlibrary{matrix,arrows}
\usetikzlibrary{backgrounds}
\usetikzlibrary{decorations.pathmorphing}
\usetikzlibrary{decorations.pathreplacing}
\usetikzlibrary{decorations.shapes}
\usetikzlibrary{decorations.markings}
\usetikzlibrary{calc}

\tikzset{l/.style={font=\fontsize{8}{8}\selectfont}}
\definecolor{salmon}{rgb}{1,0.47,0.425}

\input xy
\xyoption{all}
\CompileMatrices

\newcommand{\maps}{\colon}    
\newcommand{\R}{{\mathbb R}}  
\newcommand{\C}{{\mathbb C}}  
\renewcommand{\H}{{\mathbb H}}  
\renewcommand{\O}{{\mathbb O}}  
\newcommand{\K}{{\mathbb K}}  
\newcommand{\Z}{{\mathbb Z}}  

\newcommand{\Mor}{\mathrm{Mor}} 
\newcommand{\h}{\mathfrak{h}} 

\newcommand{\U}{{\rm U}}    
\newcommand{\SO}{{\rm SO}}    
\newcommand{\SU}{{\rm SU}}    
\newcommand{\Sp}{{\rm Sp}}    
\newcommand{\Spin}{{\rm Spin}}    
\newcommand{\ISO}{\mathrm{ISO}} 
\newcommand{\SISO}{\mathrm{SISO}} 
\newcommand{\Brane}{{\rm Brane}}    
\newcommand{\Twobrane}{2\text{-}\mathrm{Brane}} 
 
\newcommand{\G}{\mathcal{G}} 

\newcommand{\n}{{\mathfrak{n}}} 
\newcommand{\so}{{\mathfrak{so}}}  
\newcommand{\gl}{{\mathfrak{gl}}}  
\newcommand{\g}{\mathfrak{g}}  
\newcommand{\T}{\mathcal{T}} 
\newcommand{\siso}{\mathfrak{siso}} 
\newcommand{\sugra}{\mathfrak{sugra}} 
\newcommand{\brane}{\mathfrak{brane}} 
\newcommand{\twobrane}{2\text{-}\mathfrak{brane}} 
\newcommand{\superstring}{\mathfrak{superstring}} 

\newcommand{\Sym}{{\rm Sym}} 

\newcommand{\SuperVect}{\mathrm{SuperVect}} 
\newcommand{\GrAlg}{\mathrm{GrAlg}} 
\newcommand{\Man}{\mathrm{Man}} 
\newcommand{\SuperMan}{\mathrm{SuperMan}} 

\newcommand{\inv}{\mathrm{inv}} 
\newcommand{\id}{\mathrm{id}} 

\newcommand{\iso}{\cong} 
\newcommand{\To}{\Rightarrow}
\newcommand{\tensor}{\otimes} 

\newcommand{\half}{\frac{1}{2}} 

\newcommand{\V}{V} 
\renewcommand{\S}{S} 
\newcommand{\A}{\mathcal{A}} 
\newcommand{\B}{\mathcal{B}} 

\newcommand{\chibar}{\overline{\chibar}} 

\newcommand{\define}[1]{{\bf \boldmath{#1}}}
\newcommand{\arxiv}[1]{\href{http://arxiv.org/abs/#1}{arXiv:{#1}}}

\newtheorem{thm}{Theorem}

\newtheorem{prop}[thm]{Proposition}

\theoremstyle{definition}

\newtheorem{defn}[thm]{Definition}

        \newcommand{\be}{\begin{equation}}
        \newcommand{\ee}{\end{equation}}
        \newcommand{\ba}{\begin{eqnarray}}
        \newcommand{\ea}{\end{eqnarray}}
        \newcommand{\ban}{\begin{eqnarray*}}
        \newcommand{\ean}{\end{eqnarray*}}
        \newcommand{\barr}{\begin{array}}
        \newcommand{\earr}{\end{array}}

\title{Division Algebras and Supersymmetry IV}

\author{John Huerta \\
\\
CAMGSD, Instituto Superior T\'ecnico, U Lisboa \\
jhuerta@math.ist.utl.pt
}

\begin{document}
\maketitle

\begin{abstract}

  Recent work applying higher gauge theory to the superstring has indicated the
  presence of `higher symmetry', and the same methods work for the super-2-brane.  In
  the previous paper in this series, we used a geometric technique to construct a
  `Lie 2-supergroup' extending the Poincar\'e supergroup in precisely those spacetime
  dimensions where the classical Green--Schwarz superstring makes sense: 3, 4, 6 and
  10.  In this paper, we use the same technique to construct a `Lie 3-supergroup'
  extending the Poincar\'e supergroup in precisely those spacetime dimensions where
  the classical Green--Schwarz super-2-brane makes sense: 4, 5, 7 and 11.  Because
  the geometric tools are identical, our focus here is on the precise definition of a
  Lie 3-supergroup. 

\end{abstract}

\section{Introduction}

There is a deep connection between supersymmetry and the normed division
algebras: the real numbers, $\R$, the complex numbers, $\C$, the quaternions,
$\H$, and the octonions, $\O$. This is visible in super-Yang--Mills theory and
in the classical supersymmetric string and 2-brane. Most simply, it is seen in
the dimensions for which these theories make sense. The normed division
algebras have dimension $n = 1$, 2, 4 and 8, while the classical supersymmetric
string and super-Yang--Mills make sense in spacetimes of dimension two higher
than these: $n+2 = 3$, 4, 6 and 10.  Similarly, the classical supersymmetric
2-brane makes sense in dimensions three higher: $n+3 = 4$, 5, 7 and 11. At a
deeper level, we find this connection leads to `higher gauge theory', a kind of
gauge theory suitable for describing the parallel transport not merely of
particles but of extended objects, such as strings and membranes. 

This is the fourth in a series of papers exploring the relationship between
supersymmetry and division algebras
\cite{BaezHuerta:susy1,BaezHuerta:susy2,Huerta:susy3}, the first two of which were
coauthored with John Baez. In the first paper \cite{BaezHuerta:susy1}, we reviewed
the known story of how the division algebras give rise to the supersymmetry of
super-Yang--Mills theory. In the second \cite{BaezHuerta:susy2}, we made contact with
higher gauge theory: we showed how the division algebras can be used to construct
`Lie 2-superalgebras' $\superstring(n+1,1)$, which extend the Poincar\'e superalgebra
in the superstring dimensions $n+2 = 3$, 4, 6 and 10, and `Lie 3-superalgebras'
$\twobrane(n+2,1)$ in the super-2-brane dimensions $n + 3 = 4$, 5, 7 and 11.  In our
previous paper \cite{Huerta:susy3}, we described a geometric method to integrate
these Lie 2-superalgebras to `2-supergroups'.  Here we employ the same method to
integrate the Lie 3-superalgebras to `3-supergroups'. Indeed, since our method of
integration is the same, our focus here is on giving the precise definition of a
3-supergroup.

As a prelude, let us recall a bit about `2-groups'. Roughly, 2-group is a mathematical
gadget like a group \cite{BaezLauda}, but where the group axioms, such as the
associative law:
\[ (g_1 g_2) g_3 = g_1(g_2 g_3) \]
\emph{no longer hold}. Instead, they are replaced by isomorphisms, such as the
`associator':
\[ a(g_1,g_2,g_3) \maps (g_1 g_2) g_3 \To g_1(g_2 g_3) \]
which must satisfy axioms of their own. For instance, the associator satisfies
the \define{pentagon identity}, which says the following pentagon commutes:
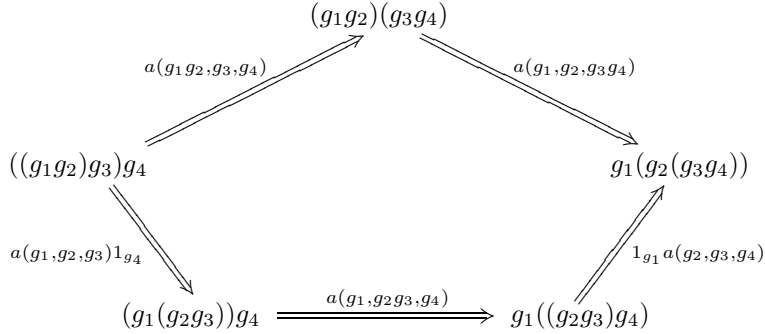
\[
\xy
 (0,20)*+{(g_1  g_2)  (g_3  g_4)}="1";
 (40,0)*+{g_1  (g_2  (g_3  g_4))}="2";
 (25,-20)*{ \quad g_1  ((g_2  g_3)  g_4)}="3";
 (-25,-20)*+{(g_1  (g_2  g_3))  g_4}="4";
 (-40,0)*+{((g_1  g_2)  g_3)  g_4}="5";
 (-23,13)*+{\scriptstyle a(g_1 g_2, g_3, g_4)};
 (26,13)*+{\scriptstyle a(g_1, g_2, g_3 g_4)};
 {\ar@{=>}     "1";"2"}
 {\ar@{=>}_{1_{g_1}  a(g_2,g_3,g_4)}  "3";"2"}
 {\ar@{=>}^{a(g_1,g_2  g_3,g_4)}    "4";"3"}
 {\ar@{=>}_{a(g_1,g_2,g_3)  1_{g_4}}  "5";"4"}
 {\ar@{=>}    "5";"1"}
\endxy
\]
Moving up in dimension, a `3-group' is a mathematical gadget like a 2-group,
but where the 2-group axioms \emph{no longer hold}. Instead, they are replaced
by isomorphisms, such as the `pentagonator':
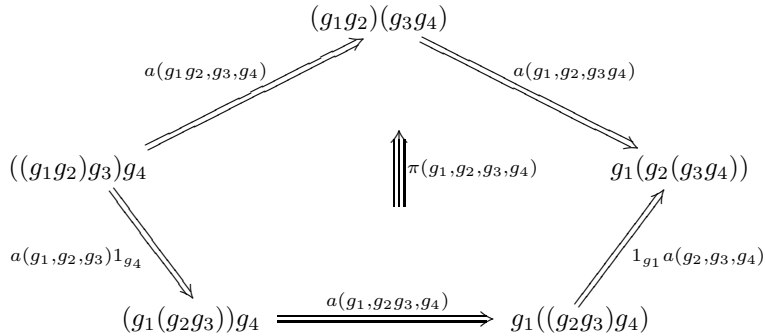
\[
\xy
 (0,20)*+{(g_1  g_2)  (g_3  g_4)}="1";
 (40,0)*+{g_1  (g_2  (g_3  g_4))}="2";
 (25,-20)*{ \quad g_1  ((g_2  g_3)  g_4)}="3";
 (-25,-20)*+{(g_1  (g_2  g_3))  g_4}="4";
 (-40,0)*+{((g_1  g_2)  g_3)  g_4}="5";
 (-23,13)*+{\scriptstyle a(g_1 g_2, g_3, g_4)};
 (26,13)*+{\scriptstyle a(g_1, g_2, g_3 g_4)};
 {\ar@{=>}     "1";"2"}
 {\ar@{=>}_{1_{g_1}  a(g_2,g_3,g_4)}  "3";"2"}
 {\ar@{=>}^{a(g_1,g_2  g_3,g_4)}    "4";"3"}
 {\ar@{=>}_{a(g_1,g_2,g_3)  1_{g_4}}  "5";"4"}
 {\ar@{=>}    "5";"1"}
 {\ar@3{<-}^{\pi(g_1,g_2,g_3,g_4)} (-2,5)*{}; (-2,-5)*{} }
\endxy
\]
which must satisfy axioms of their own. For us, the most significant of the 3-group
axioms will be the \define{pentagonator identity}, which says the following
polyhedron commutes:
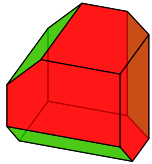
\begin{center}
  \begin{tikzpicture}[line join=round]
    \begin{scope}[scale=.2]
      \filldraw[fill=red,fill opacity=0.7](-4.306,-3.532)--(-2.391,-.901)--(-2.391,3.949)--(-5.127,.19)--(-5.127,-2.581)--cycle;
      \filldraw[fill=red,fill opacity=0.7](-4.306,-3.532)--(-2.391,-.901)--(2.872,-1.858)--(4.306,-3.396)--(3.212,-4.9)--cycle;
      \filldraw[fill=red,fill opacity=0.7](2.872,-1.858)--(2.872,5.07)--(-.135,5.617)--(-2.391,3.949)--(-2.391,-.901)--cycle;
      \filldraw[fill=green,fill opacity=0.7](4.306,-3.396)--(4.306,3.532)--(2.872,5.07)--(2.872,-1.858)--cycle;
      \filldraw[fill=green,fill opacity=0.7](-2.872,1.858)--(-.135,5.617)--(-2.391,3.949)--(-5.127,.19)--cycle;
      \filldraw[fill=red,fill opacity=0.7](3.212,-4.9)--(4.306,-3.396)--(4.306,3.532)--(2.391,.901)--(2.391,-3.949)--cycle;
      \filldraw[fill=green,fill opacity=0.7](-4.306,-3.532)--(3.212,-4.9)--(2.391,-3.949)--(-5.127,-2.581)--cycle;
      \filldraw[fill=red,fill opacity=0.7](-2.872,1.858)--(2.391,.901)--(4.306,3.532)--(2.872,5.07)--(-.135,5.617)--cycle;
      \filldraw[fill=red,fill opacity=0.7](-5.127,-2.581)--(-5.127,.19)--(-2.872,1.858)--(2.391,.901)--(2.391,-3.949)--cycle;
    \end{scope}
  \end{tikzpicture}
\end{center}
This shape is Stasheff's \define{associahedron}. This is the polyhedron where:
\begin{itemize}
	\item vertices are parenthesized lists of five things, e.g.\ $(((fg)h)k)p$;
	\item edges connect any two vertices related by an application of the associator, e.g.\ 
		\[ (((fg)h)k)p \Rightarrow ((fg)(hk))p . \]
\end{itemize}
The pentagon faces of the associahedron are filled with pentagonators, and the square
faces are filled with a naturality condition for the associator, which will be
trivial in the examples we consider.

What is the meaning of the pentagonator identity? For a certain class of 3-groups that we call
`slim 3-groups', it is a \emph{cocycle condition}. When $\G$ is a slim 3-group, it is
entirely determined by the following data:
\begin{itemize}
  \item a group $G$;
  \item an abelian group $H$ on which $G$ acts;
  \item a 4-cocycle
    \[ \pi \maps G^4 \to H \]
    in group cohomology.
\end{itemize}
The condition that the map $\pi$ be a 4-cocycle is an equation equivalent to the
pentagonator identity on $\G$. Thus, we can build slim 3-groups merely by giving a
4-cocycle in group cohomology. This will be our strategy for constructing 3-groups.

Just as we can view a Lie group as the global version of a Lie algebra, the 3-groups
we will construct will be the global versions of `Lie 3-superalgebras'. In a general,
a `Lie $n$-algebra' is a gadget like a Lie algebra, but defined on an $n$-term chain
complex, where the Lie algebra axioms hold only up to coherent chain homotopy. In
turn, this is a special case of Lada and Stasheff's $L_\infty$-algebras. All of these
algebras have `super' variants, with a $\Z_2$-grading throughout.

The Lie 3-superalgebra in question comes to us from physics. As we show in our
previous paper \cite{BaezHuerta:susy2}, the existence of the super-2-brane in
dimensions $n + 3 = 4$, 5, 7 and 11 secretly gives rise to a way to extend the
Poincar\'e superalgebra, $\siso(n+2,1)$, to a Lie 3-superalgebra we like to call
$\twobrane(n+2,1)$. Here, the Poincar\'e superalgebra is a Lie superalgebra whose
even part consists of the infinitesimal rotations $\so(n+2,1)$ and translations $\V$
on Minkowski spacetime, and whose odd part consists of `supertranslations' $\S$, or
spinors:
\[ \siso(n+2,1) = \so(n+2,1) \ltimes (\V \oplus \S) . \]
Naturally, we can identify the translations $\V$ with the vectors in
Minkowski spacetime, so $\V$ carries a Minkowski inner product $h$ of signature
$(n+2,1)$. We extend $\siso(n+2,1)$ to a Lie 3-superalgebra $\twobrane(n+2,1)$
defined on a 3-term chain complex
\[ \siso(n+2,1) \stackrel{d}{\longleftarrow} 0 \stackrel{d}{\longleftarrow} \R \]
equipped with some extra structure.

The most interesting Lie 3-algebra of this type, $\twobrane(10,1)$, plays an
important role in 11-dimensional supergravity.  This idea goes back to the work
of Castellani, D'Auria and Fr\'e \cite {TheCube, DAuriaFre}.  These authors
derived the field content of 11-dimensional supergravity starting from a
differential graded-commutative algebra, a mathematical gadget `dual' to a Lie
$n$-algebra.  Later, Sati, Schreiber and Stasheff \cite{SSS} explained that
these fields can be reinterpreted as a kind of generalized `connection' valued
in a Lie 3-superalgebra which they called `$\sugra(10,1)$'.  This is the Lie
3-superalgebra we are calling $\twobrane(10,1)$. If we follow these authors and
consider a 3-connection valued in 2-$\brane(10,1)$, we find it can be described
locally by these fields:
\vskip 1em
\begin{center}
\begin{tabular}{cc}
	\hline
	$\twobrane(n+2,1)$ & Connection component \\
	\hline
	$\R$            & $\R$-valued 3-form \\
	$\downarrow$    & \\
	$0$             & \\
	$\downarrow$    & \\
	$\siso(n+2,1)$  & $\siso(n+2,1)$-valued 1-form \\
	\hline
\end{tabular}
\end{center}
\vskip 1em Here, a $\siso(n+2,1)$-valued 1-form contains familiar fields: the
Levi-Civita connection, the vielbein, and the gravitino.  But now we also see a
3-form, called the $C$ field.  This is something we might expect on physical grounds,
at least in dimension 11, because for the quantum theory of a 2-brane to be
consistent, it must propagate in a background obeying the equations of 11-dimensional
supergravity, in which the $C$ field naturally shows up \cite{Tanii}. 

What sort of mathematical object is the $C$ field? In electromagnetism, the potential
is locally described by 1-form $A$ on spacetime, but globally it has a beautiful
geometric meaning: it is a connection on a $\U(1)$ bundle.  Similarly, in string
theory, the $B$ field is locally described by a 2-form, but globally is a connection
on a `$\U(1)$ gerbe'---this is a mathematical gadget like a bundle, but with fibers
that are categories rather than sets \cite{Bartels}. Likewise, the $C$ field is
locally described by a 3-form, but globally is a connection on a `$\U(1)$
2-gerbe'---a mathematical gadget like a gerbe, but with fibers that are 2-categories
rather than categories. This is shown in the work of Diaconescu, Freed, and Moore
\cite{DiaconescuFreedMoore} using the language of differential characters.  The work
of Aschieri and Jurco \cite{AschieriJurco} is also relevant here.

Meanwhile, the mathematical theory of 3-connections is in its infancy. Their holonomy
has been studied by Martins and Picken \cite{MartinsPicken}, and their gauge
transformations by Saemann and Wolf \cite{SaemannWolf} and in a different
formulation by Wang \cite{Wang}. Saemann and Wolf develop these ideas as an attack
on six-dimensional superconformal field theories, which physicists will recognize as
part of the program to develop M5-brane models in M-theory. In a related direction,
Sati, Schreiber and Fiorenza \cite{FSS} have discovered that $\twobrane(n+2,1)$ fits
into a much larger picture they are calling `the brane bouquet'. In particular, Sati,
Schreiber and Fiorenza find that the 3-group we construct in this paper may be the
correct background for the M5-brane, a much more mysterious object than the
2-brane that is our focus here.

So far, we have focused on the super-2-brane Lie 3-superalgebras and generalized
connections valued in them. This connection data is infinitesimal: it tells us
how to parallel transport 2-branes a little bit. Ultimately, we would like to
understand this parallel transport globally, as we do with particles in
ordinary gauge theory.

To achieve this global description, we will need `Lie $n$-groups' rather than Lie
$n$-algebras. Naively, one expects that there is a Lie 3-supergroup
$\Twobrane(n+2,1)$ for which the Lie 3-superalgebra $\twobrane(n+2,1)$ is an
infinitesimal approximation.  In fact, this is precisely what we will construct.

This paper is organized as follows. In Section \ref{sec:Lie-n-superalgebras}, we
recall the definition of a Lie $n$-superalgebra, and describe how particularly simple
examples of these arise from $(n+1)$-cocycles in Lie superalgebra cohomology. We then
describe our key example: the Lie 3-superalgebra we want to integrate,
$\twobrane(n+2,1)$, in Section \ref{sec:twobranealg}.

Section \ref{sec:lie3groups} is the heart of the paper. There we sketch how
an $n$-group may be defined using an $(n+1)$-cocyle in group cohomology, and
recall just enough tricategory theory to handle the case of interest to us:
3-groups coming from 4-cocycles. We give this definition in a way that makes it
easy to generalize to Lie 3-groups and 3-supergroups, using the supermanifold
theory we outline next.

In Section \ref{sec:supergeometry}, we provide a brief overview of supergeometry,
focusing on the supermanifolds of greatest interest to us: supermanifolds
diffeomorphic to super vector spaces, and smooth maps between them, though we outline
the general definition. In Section \ref{sec:3supergroups}, we use this supergeometry
to define a 3-supergroup, and show that 4-cocycles in supergroup cohomology give rise
to simple examples thereof.

Having defined Lie $n$-superalgebras from $(n+1)$-cocycles on Lie superalgebras
and $n$-supergroups from $(n+1)$-cocycles on supergroups, we have a plausible
strategy to integrate a Lie $n$-superalgebra to an $n$-supergroup. In Section
\ref{sec:integrating}, we recall how to do just that, for the special case
where the Lie superalgebra is nilpotent. Finally, in Section \ref{sec:finale},
we use this method to integrate the Lie 3-superalgebra $\twobrane(n+2,1)$ to a
3-supergroup.

\section{Lie \emph{n}-superalgebras from Lie superalgebra cohomology} \label{sec:Lie-n-superalgebras}

Unlike the Lie $n$-groups we shall meet in later sections, `Lie $n$-algebras' are
much simpler than their group-like counterparts. This is as one might expect from
experience with ordinary Lie groups and Lie algebras.  It is straightforward to
define a Lie $n$-algebra for all $n$. It also straightforward to incorporate the
`super' case immediately.

Better still, the `slim Lie $n$-superalgebras' we define in this section are a
prelude to the more difficult `slim Lie $n$-groups' and `slim $n$-supergroups' we
will describe later: they are constructed simply from an $(n+1)$-cocycle in Lie
superalgebra cohomology, much as the group versions will be constructed from cocycles
in the corresponding cohomology theories for Lie groups and supergroups.  Later, in
Section \ref{sec:twobranealg}, we introduce our principal example of a Lie
$n$-superalgebra: the Lie 3-superalgebra $\twobrane(n+2,1)$, which arises from the
theory of the supersymmetric 2-brane in spacetimes of dimension 4, 5, 7 and 11.

As we touched on in the Introduction, a Lie $n$-superalgebra is a certain kind of
$L_\infty$-superalgebra, which is the super version of an $L_\infty$-algebra.
This last is a chain complex, $V$:
\[ V_0 \stackrel{d}{\longleftarrow} V_1 \stackrel{d}{\longleftarrow} V_2 \stackrel{d}{\longleftarrow} \cdots \]
equipped with a structure like that of a Lie algebra, but where the Jacobi
identity only holds \emph{up to chain homotopy}, and this chain homotopy
satifies its own identity up to chain homotopy, and so on. For an
$L_\infty$-\emph{super}algebra, each term in the chain complex has a
$\Z_2$-grading, and we introduce extra signs. A Lie $n$-superalgebra is an
$L_\infty$-superalgebra in which only the first $n$ terms are nonzero, starting
with $V_0$.

Now, we shall describe how a Lie superalgebra $(n+1)$-cocycle:
\[ \omega \maps \Lambda^n \g \to \h \]
can be used to construct an especially simple Lie $n$-superalgebra, defined on
a chain complex with $\g$ in grade 0, $\h$ in grade $n-1$, and all terms in
between trivial. To make this precise, we had better start with some
definitions.

To begin at the beginning, a \define{super vector space} is a $\Z_2$-graded
vector space $V = V_0 \oplus V_1$ where $V_0$ is called the \define{even} part,
and $V_1$ is called the \define{odd} part.  There is a symmetric monoidal
category $\SuperVect$ which has:
\begin{itemize}
	\item $\Z_2$-graded vector spaces as objects;
	\item Grade-preserving linear maps as morphisms;
	\item A tensor product $\tensor$ that has the following grading: if $V
		= V_0 \oplus V_1$ and $W = W_0 \oplus W_1$, then $(V \tensor
		W)_0 = (V_0 \tensor W_0) \oplus (V_1 \tensor W_1)$ and 
              $(V \tensor W)_1 = (V_0 \tensor W_1) \oplus (V_1 \tensor W_0)$;
	\item A braiding
		\[ B_{V,W} \maps V \tensor W \to W \tensor V \]
		defined as follows: if $v \in V$ and $w \in W$ are of grade $|v|$
		and $|w|$, then
		\[ B_{V,W}(v \tensor w) = (-1)^{|v||w|} w \tensor v. \]
\end{itemize}
The braiding encodes the `the rule of signs': in any calculation, when two odd
elements are interchanged, we introduce a minus sign. We can see this in the
axioms of a Lie superalgebra, which resemble those of a Lie algebra with some
extra signs.

Briefly, a \define{Lie superalgebra} $\g$ is a Lie algebra in the category of
super vector spaces. More concretely, it is a super vector space $\g = \g_0
\oplus \g_1$, equipped with a graded-antisymmetric bracket:
\[ [-,-] \maps \Lambda^2 \g \to \g , \]
which satisfies the Jacobi identity up to signs:
\[ [X, [Y,Z]] = [ [X,Y], Z] + (-1)^{|X||Y|} [Y, [X, Z]]. \]
for all homogeneous $X, Y, Z \in \g$. Note how we have introduced an extra
minus sign upon interchanging odd $X$ and $Y$, exactly as the rule of signs says we
should. 

It is straightforward to generalize the cohomology of Lie algebras, as defined
by Chevalley--Eilenberg \cite{ChevalleyEilenberg,AzcarragaIzquierdo}, to Lie
superalgebras \cite{Leites}.  Suppose $\g$ is a Lie superalgebra and $\h$ is a
representation of $\g$. That is, $\h$ is a supervector space equipped with a Lie
superalgebra homomorphism $\rho \maps \g \to \gl(\h)$.  The \define{cohomology
of $\g$ with coefficients in $\h$} is computed using the \define{Lie
superalgebra cochain complex}, which consists of graded-antisymmetric
$p$-linear maps at level $p$:
\[ C^p(\g,\h) = \left\{ \omega \maps \Lambda^p \g \to \h \right\} . \]
We call elements of this set \define{$\h$-valued $p$-cochains on $\g$}. Note
that the $C^p(\g, \h)$ is a super vector space, in which grade-preserving
elements are even, while grade-reversing elements are odd.  When $\h = \R$, the
trivial representation, we typically omit it from the cochain complex and all
associated groups, such as the cohomology groups. Thus, we write
$C^\bullet(\g)$ for $C^\bullet(\g,\R)$.

Next, we define the coboundary operator $d \maps C^p(\g, \h) \to C^{p+1}(\g,
\h)$. Let $\omega$ be a homogeneous $p$-cochain and let $X_1, \dots, X_{p+1}$ be
homogeneous elements of $\g$. Now define:
\begin{eqnarray*}
& & d\omega(X_1, \dots, X_{p+1}) = \\ 
& & \sum^{p+1}_{i=1} (-1)^{i+1} (-1)^{|X_i||\omega|} \epsilon^{i-1}_1(i) \rho(X_i) \omega(X_1, \dots, \hat{X}_i, \dots, X_{p+1}) \\
& & + \sum_{i < j} (-1)^{i+j} (-1)^{|X_i||X_j|} \epsilon^{i-1}_1(i) \epsilon^{j-1}_1(j) \omega([X_i, X_j], X_1, \dots, \hat{X}_i, \dots, \hat{X}_j, \dots X_{p+1}) .
\end{eqnarray*}
Here, we denote the action of $\g$ on $\h$ by juxtaposition, and $\epsilon^j_i(k)$ is
shorthand for the sign one obtains by moving $X_k$ through $X_i, X_{i+1}, \dots,
X_j$. In other words,
\[ \epsilon^j_i(k) = (-1)^{|X_k|(|X_i| + |X_{i+1}| + \dots + |X_j|)}. \]

Following the usual argument for Lie algebras, one can check that:

\begin{prop}
	The Lie superalgebra coboundary operator $d$ satisfies $d^2 = 0$.
\end{prop}

\noindent
We thus say an $\h$-valued $p$-cochain $\omega$ on $\g$ is a
\define{$p$-cocycle} or \define{closed} when $d \omega = 0$, and a
\define{$p$-coboundary} or \define{exact} if there exists a $(p-1)$-cochain
$\theta$ such that $\omega = d \theta$.   Every $p$-coboundary is a
$p$-cocycle, and we say a $p$-cocycle is \define{trivial} if it is a
coboundary.  We denote the super vector spaces of $p$-cocycles and
$p$-coboundaries by $Z^p(\g,\h)$ and $B^p(\g,\h)$ respectively.  The $p$th \define{ 
Lie superalgebra cohomology of $\g$ with coefficients in $\h$}, denoted
$H^p(\g,\h)$, is defined by 
\[ H^p(\g,\h) = Z^p(\g,\h)/B^p(\g,\h). \]
This super vector space is nonzero if and only if there is a nontrivial
$p$-cocycle. In what follows, we shall be especially concerned with the even
part of this super vector space, which is nonzero if and only if there is a
nontrivial even $p$-cocycle. Our motivation for looking for even cocycles is
simple: these parity-preserving maps can regarded as morphisms in the category
of super vector spaces, which is crucial for the construction in Theorem
\ref{trivd} and everything following it.

Suppose $\g$ is a Lie superalgebra with a representation on a supervector space
$\h$.  Then we shall prove that an even $\h$-valued $(n+1)$-cocycle $\omega$ on
$\g$ lets us construct an Lie $n$-superalgebra, called $\brane_\omega(\g,\h)$,
of the following form:
\[  
\g \stackrel{d}{\longleftarrow} 0 
\stackrel{d}{\longleftarrow} \dots 
\stackrel{d}{\longleftarrow} 0 
\stackrel{d}{\longleftarrow} \h .
\]

Now let us make all of these ideas precise. In what follows, we shall use
\define{super chain complexes}, which are chain complexes in the category
SuperVect of $\Z_2$-graded vector spaces:
\[  V_0 \stackrel{d}{\longleftarrow}
    V_1 \stackrel{d}{\longleftarrow}
    V_2 \stackrel{d}{\longleftarrow} \cdots \]
Thus each $V_p$ is $\Z_2$-graded and $d$ preserves this grading.

There are thus two gradings in play: the $\Z$-grading by
\define{degree}, and the $\Z_2$-grading on each vector space, which we
call the \define{parity}. We shall require a sign convention to
establish how these gradings interact. If we consider an object of odd
parity and odd degree, is it in fact even overall?  By convention, we
assume that it is. That is, whenever we interchange something of
parity $p$ and degree $q$ with something of parity $p'$ and degree
$q'$, we introduce the sign $(-1)^{(p+q)(p'+q')}$. We shall call the
sum $p+q$ of parity and degree the \define{overall grade}, or when it
will not cause confusion, simply the grade. We denote the overall
grade of $X$ by $|X|$.

We require a compressed notation for signs. If $x_{1}, \ldots, x_{n}$ are
graded, $\sigma \in S_{n}$ a permutation, we define the \define{Koszul sign}
$\epsilon (\sigma) = \epsilon(\sigma; x_{1}, \dots, x_{n})$ by 
\[ x_{1} \cdots x_{n} = \epsilon(\sigma; x_{1}, \ldots, x_{n}) \cdot x_{\sigma(1)} \cdots x_{\sigma(n)}, \]
the sign we would introduce in the free graded-commutative algebra generated by
$x_{1}, \ldots, x_{n}$. Thus, $\epsilon(\sigma)$ encodes all the sign changes
that arise from permuting graded elements. Now define:
\[ \chi(\sigma) = \chi(\sigma; x_{1}, \dots, x_{n}) := \textrm{sgn} (\sigma) \cdot \epsilon(\sigma; x_{1}, \dots, x_{n}). \]
Thus, $\chi(\sigma)$ is the sign we would introduce in the free
graded-anticommutative algebra generated by $x_1, \dots, x_n$.

Yet we shall only be concerned with particular permutations. If $n$ is a
natural number and $1 \leq j \leq n-1$ we say that $\sigma \in S_{n}$ is an
\define{$(j,n-j)$-unshuffle} if
\[ \sigma(1) \leq\sigma(2) \leq \cdots \leq \sigma(j) \hspace{.2in} \textrm{and} \hspace{.2in} \sigma(j+1) \leq \sigma(j+2) \leq \cdots \leq \sigma(n). \] 
Readers familiar with shuffles will recognize unshuffles as their inverses. A
\emph{shuffle} of two ordered sets (such as a deck of cards) is a permutation
of the ordered union preserving the order of each of the given subsets. An
\emph{unshuffle} reverses this process. We denote the collection of all
$(j,n-j)$ unshuffles by $S_{(j,n-j)}$.

The following definition of an $L_{\infty}$-algebra was formulated by
Schlessinger and Stasheff in 1985 \cite{SS}:

\begin{defn} \label{L-alg} An
\define{$L_{\infty}$-superalgebra} is a graded vector space $V$
equipped with a system $\{l_{k}| 1 \leq k < \infty\}$ of linear maps
$l_{k} \maps V^{\otimes k} \rightarrow V$ with $\deg(l_{k}) = k-2$
which are totally antisymmetric in the sense that
\begin{eqnarray}
   l_{k}(x_{\sigma(1)}, \dots,x_{\sigma(k)}) =
   \chi(\sigma)l_{k}(x_{1}, \dots, x_{n})
\label{antisymmetry}
\end{eqnarray}
for all $\sigma \in S_{n}$ and $x_{1}, \dots, x_{n} \in V,$ and,
moreover, the following generalized form of the Jacobi identity
holds for $0 \le n < \infty :$
\begin{eqnarray}
   \displaystyle{\sum_{i+j = n+1}
   \sum_{\sigma \in S_{(i,n-i)}}
   \chi(\sigma)(-1)^{i(j-1)} l_{j}
   (l_{i}(x_{\sigma(1)}, \dots, x_{\sigma(i)}), x_{\sigma(i+1)},
   \ldots, x_{\sigma(n)}) =0,}
\label{megajacobi}
\end{eqnarray}
where the inner summation is taken over all $(i,n-i)$-unshuffles with $i
\geq 1.$
\end{defn}

A \define{Lie $n$-superalgebra} is an $L_\infty$-superalgebra where only the
first $n$ terms of the chain complex are nonzero. A \define{slim Lie
$n$-superalgebra} is a Lie $n$-superalgebra $V$ with only two nonzero
terms, $V_0$ and $V_{n-1}$, and $d=0$. Given an $\h$-valued $(n+1)$-cocycle
$\omega$ on a Lie superalgebra $\g$, we can construct a slim Lie
$n$-superalgebra $\brane_\omega(\g,\h)$ with:
\begin{itemize}
	\item $\g$ in degree 0, $\h$ in degree $n-1$, and trivial super vector
		spaces in between,
	\item $d = 0$, 
	\item $l_2 \maps (\g \oplus \h)^{\tensor 2} \to \g \oplus \h$ given by:
		\begin{itemize}
			\item the Lie bracket on $\g \tensor \g$,
			\item the action on $\g \tensor \h$,
			\item zero on $\h \tensor \h$, as required by degree.
		\end{itemize}
	\item $l_{n+1} \maps (\g \oplus \h)^{\tensor(n+1)} \to \g \oplus \h$
		given by the cocycle $\omega$ on $\g^{\tensor(n+1)}$, and
		zero otherwise, as required by degree,
	\item all other maps $l_k$ zero, as required by degree.
\end{itemize}

It remains to prove that this is, in fact, a Lie $n$-superalgebra. Indeed,
more is true: every slim Lie $n$-superalgebra is precisely of this form.

\begin{thm} \label{trivd}
	$\brane_\omega(\g,\h)$ is a Lie $n$-superalgebra. Conversely, every
	slim Lie $n$-superalgebra is of the form $\brane_\omega(\g,\h)$ for some
	Lie superalgebra $\g$, representation $\h$, and $\h$-valued
	$(n+1)$-cocycle $\omega$ on $\g$.
\end{thm}

\begin{proof}
  See \cite[Thm. 17]{BaezHuerta:susy2}, which is a straightforward generalization of
  the theorem found in Baez and Crans \cite{BaezCrans} to the super case.
\end{proof}

The key to above theorem is recognizing that $\omega$ is a Lie superalgebra
cocycle if and only if the generalized Jacobi identity, Equation
\ref{megajacobi}, holds. When $\h$ is the trivial representation $\R$, we omit
it, and write $\brane_\omega(\g)$ for $\brane_\omega(\g,\R)$.

\subsection{The super-2-brane Lie 3-superalgebra} \label{sec:twobranealg}

The supersymmetric 2-brane exists only in spacetimes of dimension 4, 5, 7 and
11. Secretly, this is because the infinitesimal symmetries of spacetimes in these
dimensions can be extended to a Lie 3-algebra, as we now describe. For full details,
see our previous paper with Baez \cite{BaezHuerta:susy2}.

One of the principal themes of theoretical physics over the last century has been the
search for the underlying symmetries of nature. This began with special relativity,
which could be summarized as the discovery that the laws of physics are invariant
under the action of the Poincar\'e group:
\[ \ISO(V) = \Spin(V) \ltimes V. \]
Here, $V$ is the set of vectors in Minkowski spacetime and acts on Minkowski
spacetime by translation, while $\Spin(V)$ is the \define{Lorentz group}: the double
cover of $\SO_0(V)$, the connected component of the group of symmetries of the
Minkowski norm. Much of the progress in physics since special relativity has been
associated with the discovery of additional symmetries, like the $\U(1) \times \SU(2)
\times \SU(3)$ symmetries of the Standard Model of particle physics
\cite{BaezHuerta:guts}.

Today, `supersymmetry' could be summarized as the hypothesis that the laws of physics
are invariant under the `Poincar\'e supergroup', which is larger than the Poincar\'e
group:
\[ \SISO(V) = \Spin(V) \ltimes T. \]
Here, $V$ is again the set of vectors in Minkowski spacetime and $\Spin(V)$ is the
Lorentz group, but $T$ is the supergroup of translations on Minkowski
`superspacetime'. Though we have not yet discussed enough supergeometry to talk about
$T$ precisely, at the moment we only need its infinitesimal approximation: the
\define{supertranslation algebra}, $\T$. This is the super vector space with
\[ \T_0 = V, \quad \T_1 = S , \]
where $V$, as before, is the set of vectors in Minkowski spacetime, and $S$ is a
spinor representation of $\Spin(V)$. We think of the spinor representation $S$ as
giving extra, supersymmetric translations, or `supersymmetries'.

For a suitable choice of spinor representation, there is a symmetric map, equivariant
with respect to the action of $\Spin(V)$:
\[ [-,-] \maps \Sym^2 S \to V . \]
This makes $\T$ into a Lie superalgebra where the bracket of two elements of $S$ is
given by the map above, and all other brackets are trivial. It is clear that this
satisfies the Jacobi identity, in a trivial way: bracketing with a bracket yields
zero.

Everything we have said so far makes sense in spacetimes of any dimension. Yet the
dimensions where we can write down the Green--Schwarz super 2-brane action
\cite{Duff}, $n + 3 = 4, 5, 7$ and 11 are special, thanks to the normed division
algebras \cite{BaezHuerta:susy2, FootJoshi}. Recall that a \define{normed division
algebra} is a nonassociative real algebra $\K$ equipped with a norm $| - |$ such that
$|xy| = |x||y|$ for any $x, y \in \K$. There are precisely four such algebras: the
real numbers $\R$, the complex numbers $\C$, the quaternions $\H$, and the octonions
$\O$, with dimensions $n = 1, 2, 4$ and $8$.

In the super 2-brane dimensions $n + 3 = 4, 5, 7$ and 11, we can use the normed
division algebras to construct the supertranslation algebra $\T$ and a nontrivial
4-cocycle. In our previous paper \cite{BaezHuerta:susy2}, we showed how the vectors
$\V$ are a certain subspace of the $4 \times 4$ matrices over $\K$, the spinors $\S =
\K^4$. This yields an obvious map
\[ \V \tensor \S \to \S \]
given by matrix multiplication. We can show this map is
$\Spin(\V)$-equivariant. Moreover, we used this description to construct a
$\Spin(\V)$-invariant pairing
\[ \langle -, - \rangle \maps \S \tensor \S \to \R \]
and a $\Spin(\V)$-equivariant bracket
\[ [ -,- ] \maps \Sym^2 \S \to \V . \]
The latter gives a supertranslation algebra $\T = \V \oplus \S$, while the former
allows us to construct a 4-cocycle on $\T$.  We can decompose the space
of $n$-cochains on $\T$ into summands by counting how many of the
arguments are vectors and how many are spinors:
\[     C^n(\T) \iso  
 \bigoplus_{p + q = n}   (\Lambda^p(\V) \otimes \Sym^q(\S))^* . \]
We call an element of $(\Lambda^p(\V) \otimes \Sym^q(\S))^*$ a 
\define{$(p,q)$-form}.

\begin{thm}
\label{thm:4-cocycle}
	In dimensions 4, 5, 7 and 11, the supertranslation algebra $\T$ has a
	nontrivial, Lorentz-invariant even 4-cocycle, namely the unique
	$(2,2)$-form with
	\[ \beta(\Psi, \Phi, \A, \B) = \langle \Psi, \A (\B \Phi) - \B (\A \Phi) \rangle \]
        for spinors $\Psi, \Phi \in \S$ and vectors $\A, \B \in \V$. Here the vectors
        $\A$ and $\B$ can act on the spinor $\Phi$ because they are $4 \times 4$
        matrices with entries in $\K$, while $\Phi$ is an element of $\K^4$.
\end{thm} 

\begin{proof}
	See \cite[Thm. 15]{BaezHuerta:susy2}.
\end{proof}

In spacetime dimensions 4, 5, 7 and 11, we proved in Theorem
\ref{thm:4-cocycle} that there is a 4-cocycle $\beta$, which is nonzero only
when its arguments consist of two spinors and two vectors:
\[ \begin{array}{cccc}
	\beta \maps & \Lambda^4(\T)                        & \to     & \R \\
	            & \A \wedge \B \wedge \Psi \wedge \Phi & \mapsto & \langle \Psi, (\A \wedge \B) \Phi \rangle . \\
   \end{array}
\]
There is thus a Lie 3-superalgebra, the \define{supertranslation Lie 3-superalgebra},
$\brane_{\beta}(\T)$. There is much more that one can do with the cocycle $\beta$,
however. We can use it to extend not just the supertranslations $\T$ to a Lie
3-superalgebra, but the full Poincar\'e superalgebra, $\so(V) \ltimes \T$. This is
because $\beta$ is invariant under the action of $\Spin(\V)$, by construction: it is
made of $\Spin(\V)$-equivariant maps.

\begin{prop} \label{prop:extendingcochains1}
       	Let $\g$ and $\h$ be Lie superalgebras such that $\g$ acts on
	$\h$, and let $R$ be a representation of $\g \ltimes \h$. Given any
	$R$-valued $n$-cochain $\omega$ on $\h$, we can uniquely extend it to
	an $n$-cochain $\tilde{\omega}$ on $\g \ltimes \h$ that takes the value
	of $\omega$ on $\h$ and vanishes on $\g$. When $\omega$ is even, we
	have:
	\begin{enumerate}
		\item $\tilde{\omega}$ is closed if and only if $\omega$ is
			closed and $\g$-equivariant.
		\item $\tilde{\omega}$ is exact if and only if $\omega =
			d\theta$, for $\theta$ a $\g$-equivariant
			$(n-1)$-cochain on $\h$.
	\end{enumerate}
\end{prop}
\begin{proof}
  See \cite[Prop. 20]{BaezHuerta:susy2}.
\end{proof}

\begin{thm} \label{thm:2brane}
	In dimensions 4, 5, 7 and 11, there exists a Lie 3-superalgebra formed
	by extending the Poincar\'e superalgebra $\siso(n+2,1)$ by the
	4-cocycle $\beta$, which we call the \define{2-brane Lie
	3-superalgebra, $\twobrane(n+2,1)$}.
\end{thm}

\section{Lie 3-groups from group cohomology} \label{sec:lie3groups}

Roughly speaking, an `$n$-group' is a weak $n$-groupoid with one object---an
$n$-category with one object in which all morphisms are weakly invertible, up
to higher-dimensional morphisms. This definition is a rough one because there
are many possible definitions to use for `weak $n$-category', but despite this
ambiguity, it can still serve to motivate us. 

The richness of weak $n$-categories, no matter what definition we apply, makes
$n$-groups a complicated subject. In the midst of this complexity, we seek to define
a class of $n$-groups that have a simple description, and which are straightforward
to internalize, so that we may easily construct Lie $n$-groups and Lie
$n$-supergroups, as we shall do later in this paper. The motivating example for this
is what Baez and Lauda \cite{BaezLauda} call a `special 2-group', which has a
concrete description using group cohomology. Since Baez and Lauda prove that all
2-groups are equivalent to special ones, group cohomology also serves to classify
2-groups. We seek a similar class of Lie 3-groups.

So, we will define `slim Lie 3-groups'. This is an Lie 3-group which is trivial in
the middle: all 2-morphisms are the identity. In more generality, `slim Lie
$n$-groups' should be a Lie $n$-group where all $k$-morphisms are the identity for $1
< k < n$, though we will make this precise only for $n = 3$. The concept is useful
because such Lie $n$-groups can be completely classified by Lie group
cohomology. They are also easy to `superize', and their super versions can be
completely classified using Lie supergroup cohomology, as we shall see in Section
\ref{sec:3supergroups}. Finally, we note that we could equally well define `slim
$n$-groups', working in the category of sets rather than the category of smooth
manifolds. Indeed, when $n=2$, this is what Baez and Lauda call a `special 2-group',
though we prefer the adjective `slim'.

We should stress that the definition of Lie 3-group we give here, while it is good
enough for our needs, is known to be too naive in some important respects. For
instance, it does not seem possible to integrate every Lie 3-algebra to a Lie 3-group
of this type, while Henriques's definition of Lie $n$-group does make this possible
\cite{Henriques}.

First we need to review the cohomology of Lie groups, as originally defined by
van Est \cite{vanEst}, who was working in parallel with the definition of group
cohomology given by Eilenberg and Mac Lane. Fix a Lie group $G$, an abelian Lie
group $H$, and a smooth action of $G$ on $H$ which respects addition in $H$.
That is, for any $g \in G$ and $h, h' \in H$, we have:
\[ g(h + h') = gh + gh'. \]
Then \define{the cohomology of $G$ with coefficients in $H$} is given by the
\define{Lie group cochain complex}, $C^{\bullet}(G,H)$. At level $p$, this
consists of the smooth functions from $G^p$ to $H$:
\[ C^p(G,H) = \left\{ f \maps G^p \to H \right\}. \]
We call elements of this set \define{$H$-valued $p$-cochains on $G$}. The
boundary operator is the same as the one defined by Eilenberg--Mac Lane. On a
$p$-cochain $f$, it is given by the formula:
\begin{eqnarray*}
	df(g_1, \dots, g_{p+1}) & = & g_1 f(g_2, \dots, g_{p+1}) \\
	                        &   & + \sum_{i=1}^p (-1)^i f(g_1, \dots, g_{i-1}, g_i g_{i+1}, g_{i+2}, \dots, g_{p+1}) \\
				&   & + (-1)^{p+1} f(g_1, \dots, g_p) .
\end{eqnarray*}
The proof that $d^2 = 0$ is routine. All the usual terminology applies: a
$p$-cochain $f$ for which $df = 0$ is called \define{closed}, or a
\define{cocycle}, a $p$-cochain $f = dg$ for some $(p-1)$-cochain $g$ is called
\define{exact}, or a \define{coboundary}. A $p$-cochain is said to be
\define{normalized} if it vanishes when any of its entries is 1. Every
cohomology class can be represented by a normalized cocycle. Finally, when $H =
\R$ with trivial $G$ action, we omit it when writing the complex
$C^\bullet(G)$, and we call real-valued cochains, cocycles, or coboundaries,
simply cochains, cocycles or coboundaries, respectively.

This last choice, that $\R$ will be our default coefficient group, may seem
innocuous, but there is another one-dimensional abelian Lie group we might have
chosen: $\U(1)$, the group of phases. This would have been an equally valid
choice, and perhaps better for some physical applications, but we have chosen
$\R$ because it simplifies our formulas slightly.

We now sketch how to build a slim Lie 3-group from an 4-cocycle. In essence, given a
normalized $H$-valued 4-cocycle $\pi$ on a Lie group $G$, we want to construct a Lie
3-group $\Brane_\pi(G,H)$, which is the smooth, weak 3-groupoid with:
\begin{itemize}
	\item One object. We can depict this with a dot, or `0-cell': $\bullet$

	\item For each element $g \in G$, a 1-automorphism of the one object, which
          we depict as an arrow, or `1-cell':
		\[ \xymatrix{ \bullet \ar[r]^g & \bullet }, \quad g \in G. \]
		Composition corresponds to multiplication in the group:
		\[ \xymatrix{ \bullet \ar[r]^g & \bullet \ar[r]^{g'} & \bullet } = \xymatrix{ \bullet \ar[r]^{gg'} & \bullet }. \]
              \item Trivial 2-morphisms. If we depict 2-morphisms with 2-cells, we
                are saying there is just one of these (the identity) for each
                1-morphism:
		
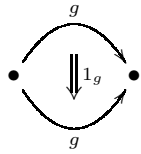
\[
		\xy
		(-8,0)*+{\bullet}="4";
		(8,0)*+{\bullet}="6";
		{\ar@/^1.65pc/^g "4";"6"};
		{\ar@/_1.65pc/_g "4";"6"};
		{\ar@{=>}^{1_g} (0,3)*{};(0,-3)*{}} ;
		\endxy
		\] 

              \item For each element $h \in H$, an 3-automorphism on the identity of
                the 1-morphism $g$, and no 3-morphisms which are not 3-automorphisms:
		
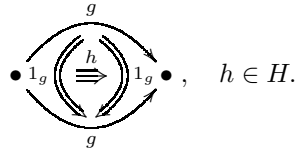
\[
		\xy 
		(-10,0)*+{\bullet}="1";
		(10,0)*+{\bullet}="2";
		{\ar@/^1.65pc/^g "1";"2"};
		{\ar@/_1.65pc/_g "1";"2"};
		(0,5)*+{}="A";
		(0,-5)*+{}="B";
		{\ar@{=>}@/_.75pc/ "A"+(-1.33,0) ; "B"+(-.66,-.55)};
		{\ar@{=}@/_.75pc/ "A"+(-1.33,0) ; "B"+(-1.33,0)};
		{\ar@{=>}@/^.75pc/ "A"+(1.33,0) ; "B"+(.66,-.55)};
		{\ar@{=}@/^.75pc/ "A"+(1.33,0) ; "B"+(1.33,0)};
		{\ar@3{->} (-2,0)*{}; (2,0)*{}};
		(0,2.5)*{\scriptstyle h};
		(-7,0)*{\scriptstyle 1_g};
		(7,0)*{\scriptstyle 1_g};
		\endxy
		, \quad h \in H.
		\]

              \item There are three ways of composing 3-morphisms, given by different
                ways of sticking 3-cells together---we can glue two 3-cells along a
                2-cell, which should just correspond to addition in $H$:
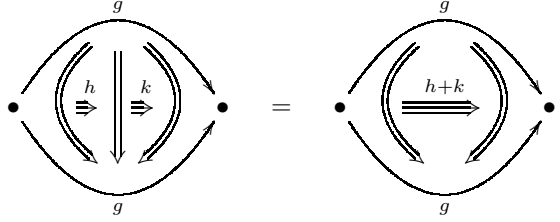
		\[
		\xy 0;/r.22pc/:
		(0,15)*{};
		(0,-15)*{};
		(0,8)*{}="A";
		(0,-8)*{}="B";
		{\ar@{=>} "A" ; "B"};
		{\ar@{=>}@/_1pc/ "A"+(-4,1) ; "B"+(-3,0)};
		{\ar@{=}@/_1pc/ "A"+(-4,1) ; "B"+(-4,1)};
		{\ar@{=>}@/^1pc/ "A"+(4,1) ; "B"+(3,0)};
		{\ar@{=}@/^1pc/ "A"+(4,1) ; "B"+(4,1)};
		{\ar@3{->} (-6,0)*{} ; (-2,0)*+{}};
		(-4,3)*{\scriptstyle h};
		{\ar@3{->} (2,0)*{} ; (6,0)*+{}};
		(4,3)*{\scriptstyle k};
		(-15,0)*+{\bullet}="1";
		(15,0)*+{\bullet}="2";
		{\ar@/^2.75pc/^g "1";"2"};
		{\ar@/_2.75pc/_g "1";"2"};
		\endxy 
		\quad = \quad
		\xy 0;/r.22pc/:
		(0,15)*{};
		(0,-15)*{};
		(0,8)*{}="A";
		(0,-8)*{}="B";
		{\ar@{=>}@/_1pc/ "A"+(-4,1) ; "B"+(-3,0)};
		{\ar@{=}@/_1pc/ "A"+(-4,1) ; "B"+(-4,1)};
		{\ar@{=>}@/^1pc/ "A"+(4,1) ; "B"+(3,0)};
		{\ar@{=}@/^1pc/ "A"+(4,1) ; "B"+(4,1)};
		{\ar@3{->} (-6,0)*{} ; (6,0)*+{}};
		(0,3)*{\scriptstyle h + k};
		(-15,0)*+{\bullet}="1";
		(15,0)*+{\bullet}="2";
		{\ar@/^2.75pc/^g "1";"2"};
		{\ar@/_2.75pc/_g "1";"2"};
		\endxy .
		\]
		We also can glue two 3-cells along a 1-cell, which should again
		just be addition in $H$:
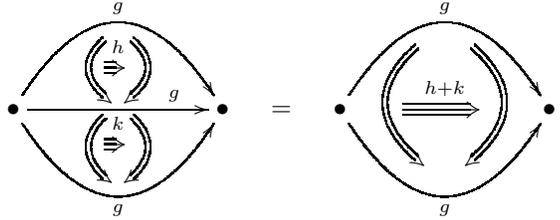
		\[	
		\xy 0;/r.22pc/:
		(0,15)*{};
		(0,-15)*{};
		(0,9)*{}="A";
		(0,1)*{}="B";
		{\ar@{=>}@/_.5pc/ "A"+(-2,1) ; "B"+(-1,0)};
		{\ar@{=}@/_.5pc/ "A"+(-2,1) ; "B"+(-2,1)};
		{\ar@{=>}@/^.5pc/ "A"+(2,1) ; "B"+(1,0)};
		{\ar@{=}@/^.5pc/ "A"+(2,1) ; "B"+(2,1)};
		{\ar@3{->} (-2,6)*{} ; (2,6)*+{}};
		(0,9)*{\scriptstyle h};
		(0,-1)*{}="A";
		(0,-9)*{}="B";
		{\ar@{=>}@/_.5pc/ "A"+(-2,-1) ; "B"+(-1,-1.5)};
		{\ar@{=}@/_.5pc/ "A"+(-2,0) ; "B"+(-2,-.7)};
		{\ar@{=>}@/^.5pc/ "A"+(2,-1) ; "B"+(1,-1.5)};
		{\ar@{=}@/^.5pc/ "A"+(2,0) ; "B"+(2,-.7)};
		{\ar@3{->} (-2,-5)*{} ; (2,-5)*+{}};
		(0,-2)*{\scriptstyle k};
		(-15,0)*+{\bullet}="1";
		(15,0)*+{\bullet}="2";
		{\ar@/^2.75pc/^g "1";"2"};
		{\ar@/_2.75pc/_g "1";"2"};
		{\ar "1";"2"};
		(8,2)*{\scriptstyle g};
		\endxy 
		\quad = \quad
		\xy 0;/r.22pc/:
		(0,15)*{};
		(0,-15)*{};
		(0,8)*{}="A";
		(0,-8)*{}="B";
		{\ar@{=>}@/_1pc/ "A"+(-4,1) ; "B"+(-3,0)};
		{\ar@{=}@/_1pc/ "A"+(-4,1) ; "B"+(-4,1)};
		{\ar@{=>}@/^1pc/ "A"+(4,1) ; "B"+(3,0)};
		{\ar@{=}@/^1pc/ "A"+(4,1) ; "B"+(4,1)};
		{\ar@3{->} (-6,0)*{} ; (6,0)*+{}};
		(0,3)*{\scriptstyle h + k};
		(-15,0)*+{\bullet}="1";
		(15,0)*+{\bullet}="2";
		{\ar@/^2.75pc/^g "1";"2"};
		{\ar@/_2.75pc/_g "1";"2"};
		\endxy  .
		\]
		And finally, we can glue two 3-cells at the 0-cell, the object
		$\bullet$.  This is the only composition of 3-morphisms where
		the attached 1-morphisms can be distinct, which distinguishes
		it from the first two cases. It should be addition
		\emph{twisted by the action of $G$}:
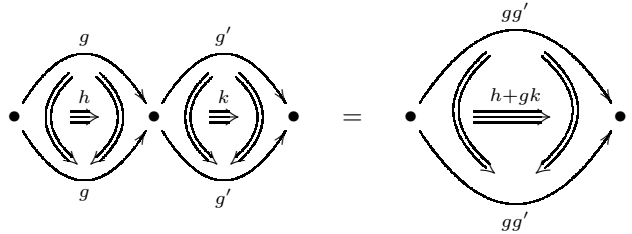
		\[
		\xy 0;/r.22pc/:
		(0,15)*{};
		(0,-15)*{};
		(-20,0)*+{\bullet}="1";
		(0,0)*+{\bullet}="2";
		{\ar@/^2pc/^g "1";"2"};
		{\ar@/_2pc/_g "1";"2"};
		(20,0)*+{\bullet}="3";
		{\ar@/^2pc/^{g'} "2";"3"};
		{\ar@/_2pc/_{g'} "2";"3"};
		(-10,6)*+{}="A";
		(-10,-6)*+{}="B";
		{\ar@{=>}@/_.7pc/ "A"+(-2,0) ; "B"+(-1,-.8)};
		{\ar@{=}@/_.7pc/ "A"+(-2,0) ; "B"+(-2,0)};
		{\ar@{=>}@/^.7pc/ "A"+(2,0) ; "B"+(1,-.8)};
		{\ar@{=}@/^.7pc/ "A"+(2,0) ; "B"+(2,0)};
		(10,6)*+{}="A";
		(10,-6)*+{}="B";
		{\ar@{=>}@/_.7pc/ "A"+(-2,0) ; "B"+(-1,-.8)};
		{\ar@{=}@/_.7pc/ "A"+(-2,0) ; "B"+(-2,0)};
		{\ar@{=>}@/^.7pc/ "A"+(2,0) ; "B"+(1,-.8)};
		{\ar@{=}@/^.7pc/ "A"+(2,0) ; "B"+(2,0)};
		{\ar@3{->} (-12,0)*{}; (-8,0)*{}};
		(-10,3)*{\scriptstyle h};
		{\ar@3{->} (8,0)*{}; (12,0)*{}};
		(10,3)*{\scriptstyle k};
		\endxy 
		\quad = \quad 
		\xy 0;/r.22pc/:
		(0,15)*{};
		(0,-15)*{};
		(0,8)*{}="A";
		(0,-8)*{}="B";
		{\ar@{=>}@/_1pc/ "A"+(-4,1) ; "B"+(-3,0)};
		{\ar@{=}@/_1pc/ "A"+(-4,1) ; "B"+(-4,1)};
		{\ar@{=>}@/^1pc/ "A"+(4,1) ; "B"+(3,0)};
		{\ar@{=}@/^1pc/ "A"+(4,1) ; "B"+(4,1)};
		{\ar@3{->} (-6,0)*{} ; (6,0)*+{}};
		(0,3)*{\scriptstyle h + g k};
		(-15,0)*+{\bullet}="1";
		(15,0)*+{\bullet}="2";
		{\ar@/^2.75pc/^{gg'} "1";"2"};
		{\ar@/_2.75pc/_{gg'} "1";"2"};
		\endxy .
		\]
		
              \item For any 4-tuple of 1-morphisms, a 3-automorphism $\pi(g_1, g_2,
                g_3, g_4)$ on the identity of the 1-morphism $g_1 g_2 g_3 g_4$. We
                call $\pi$ the \define{pentagonator}.

              \item $\pi$ satisfies an equation corresponding to the 3-dimensional
                associahedron, which is equivalent to the cocycle condition.
\end{itemize}

The rest of this section is concerned with making this definition precise. We proceed
as follows. In Section \ref{sec:smoothcats}, we recall the definition of smooth
categories and bicategories from our earlier paper. In Section
\ref{sec:smoothtricats}, we outline the definition of a smooth tricategory. We do not
need the full details, because our slim Lie 3-groups will only use some of the
structure. In Section \ref{sec:inverses}, we discuss inverses and weak inverses in a
tricategory. Finally, in Section \ref{sec:l3g} we give the definition of Lie
3-groups, and we show how 4-cocycles in Lie group cohomology give rise to slim Lie
3-groups.

\subsection{Smooth categories and bicategories} \label{sec:smoothcats}

In order to define Lie 3-groups, we now sketch the definition of a `smooth weak
$n$-category', where every set is a smooth manifold and every structure map is
smooth, at least for $n \leq 3$. We focus especially on the case of smooth
tricategories, as this part is new. Our previous paper \cite{Huerta:susy3} includes
the full definition of smooth categories and smooth bicategories, so here we only
recall the main ideas.

The key idea at play in this section is internalization. Due to Ehresmann
\cite{Ehresmann} in the 1960s, internalization has become part of the toolkit of the
working category theorist. The essence of the idea should be familiar to all
mathematicians: any mathematical structure that can be defined using sets, functions,
and equations between functions can be defined in categories other than Set. For
instance, a group in the category of smooth manifolds is a Lie group.  To perform
internalization, we apply this idea to the definition of category itself. We recall
the essentials here to define `smooth categories'. More generally, one can define a
`category in $K$' for many categories $K$, though here we will work exclusively with
the example where $K$ is the category of smooth manifolds. For a readable treatment
of internalization, see Borceux's handbook \cite{Borceux}. A general treatment of
internal bicategories appears in the work of Douglas and Henriques
\cite{DouglasHenriques}.

A word of caution is needed here before we proceed: \emph{in this section
only}, we are bucking standard mathematical practice by writing the result of
doing first $\alpha$ and then $\beta$ as $\alpha \circ \beta$ rather than
$\beta \circ \alpha$, as one would do in most contexts where $\circ$ denotes
composition of \emph{functions}. This has the effect of changing how we read
commutative diagrams. For instance, the commutative triangle:
\[ \xymatrix{ f \ar[r]^\alpha \ar[rd]_\gamma & g \ar[d]^\beta \\
		         & h \\
}
\]
reads $\gamma = \alpha \circ \beta$ rather than $\gamma = \beta \circ \alpha$.

As the reader knows, a category has a collection of objects
\[ x \, \bullet, \]
and a collection of morphisms
\[ \xymatrix{ x \, \bullet \ar[r]^f & \bullet \, y } \]
such that suitable pairs of morphisms can be composed:
\[ \xymatrix{ x \ar[r]^f & y \ar[r]^g & z } = \xymatrix{ x \ar[r]^{fg} & z } . \]
Composition is associative:
\[ (fg)h = f(gh) \]
and has left and right units:
\[ 1_x f = f = f 1_y \]
for $f \maps x \to y$, and units $1_x \maps x \to x$ and $1_y \maps y \to y$. 

A \define{smooth category} $C$ is just like a category, but now the collection of
objects $C_0$ and morphisms $C_1$ are smooth manifolds and composition, along with
some less obvious operations, are given by smooth maps. Specifically, 
\begin{itemize}
         \item the {\bf source} and {\bf target} maps $s,t \maps C_{1} \rightarrow C_{0}$,
  sending a morphism $f \maps x \to y$ to $x$ and $y$, respectively, are smooth.
  
	\item the {\bf identity-assigning} map $i \maps C_{0} \rightarrow C_{1}$, sending an object $x$ to its identity map $1_x \maps x \to x$, is smooth.
	\item and, as already mentioned, {\bf composition} $\circ \maps C_{1} \times _{C_{0}}
		C_{1} \rightarrow C_{1}$ is smooth, where $C_1 \times_{C_0} C_1$ is the
		pullback of the source and target maps:
		\[ C_1 \times_{C_0} C_1 = \left\{ (f,g) \in C_1 \times C_1 : t(f) = s(g) \right\}, \]
		and is assumed to be a smooth manifold.
\end{itemize}
These maps satisfy equations making $C$ into a category. Between smooth categories
$C$ and $D$, we have smooth functors $F \maps C \to D$, and between smooth functors
we have smooth natural transformations
\[
\xy 
(-10,0)*+{C}="1";
(10,0)*+{C'}="2";
{\ar@/^1.65pc/^{F} "1";"2"};
{\ar@/_1.65pc/_{G} "1";"2"};
(0,5)*+{}="A";
{\ar@{=>}^{\scriptstyle \theta} (0,3)*{};(0,-3)*{}} ;
(0,-5)*+{}="B";
\endxy .
\]
Both of these are defined in the obvious way.

With this bit of smooth category theory, we are now ready to move on to smooth
bicategory theory. Before we give this definition, let us review the idea of a
`bicategory', so that its basic simplicity is not obscured in technicalities. A
bicategory has objects:
\[ x \, \bullet, \]
morphisms going between objects,
\[ \xymatrix{ x \, \bullet \ar[r]^f & \bullet \, y}, \]
and 2-morphisms going between morphisms:
\[
\xy
(-10,0)*+{x};
(-8,0)*+{\bullet}="4";
(8,0)*+{\bullet}="6";
(10,0)*+{y};
{\ar@/^1.65pc/^f "4";"6"};
{\ar@/_1.65pc/_g "4";"6"};
{\ar@{=>}^{\scriptstyle \alpha} (0,3)*{};(0,-3)*{}} ;
\endxy .
\] 
Morphisms in a bicategory can be composed just as morphisms in a category:
\[ \xymatrix{ x \ar[r]^f & y \ar[r]^g & z } \quad = \quad \xymatrix{ x \ar[r]^{f \cdot g} & z } . \]
But there are two ways to compose 2-morphisms---vertically:
\[
\xy
(-8,0)*+{x}="4";
(8,0)*+{y}="6";
{\ar^g "4";"6"};
{\ar@/^1.75pc/^{f} "4";"6"};
{\ar@/_1.75pc/_{h} "4";"6"};
{\ar@{=>}^<<{\scriptstyle \alpha} (0,6)*{};(0,1)*{}} ;
{\ar@{=>}^<<{\scriptstyle \beta} (0,-1)*{};(0,-6)*{}} ;
\endxy
\quad = \quad \xy
(-8,0)*+{x}="4";
(8,0)*+{y}="6";
{\ar@/^1.65pc/^f "4";"6"};
{\ar@/_1.65pc/_h "4";"6"};
{\ar@{=>}^{\scriptstyle \alpha \circ \beta } (0,3)*{};(0,-3)*{}} ;
\endxy
\] 
and horizontally:
\[
\xy
(-16,0)*+{x}="4";
(0,0)*+{y}="6";
{\ar@/^1.65pc/^{f} "4";"6"};
{\ar@/_1.65pc/_{g} "4";"6"};
{\ar@{=>}^<<<{\scriptstyle \alpha} (-8,3)*{};(-8,-3)*{}} ;
(0,0)*+{y}="4";
(16,0)*+{z}="6";
{\ar@/^1.65pc/^{f'} "4";"6"};
{\ar@/_1.65pc/_{g'} "4";"6"};
{\ar@{=>}^<<<{\scriptstyle \beta} (8,3)*{};(8,-3)*{}} ;
\endxy
\quad = \quad \xy
(-10,0)*+{x}="4";
(10,0)*+{z}="6";
{\ar@/^1.65pc/^{f \cdot f'} "4";"6"};
{\ar@/_1.65pc/_{g \cdot g'} "4";"6"};
{\ar@{=>}^{\alpha \cdot \beta} (0,3)*{};(0,-3)*{}} ;
\endxy .
\] 
Unlike a category, composition of morphisms need not be associative or have
left and right units. The presence of 2-morphisms allows us to \emph{weaken the
axioms}.  Rather than demanding $(f \cdot g) \cdot h = f \cdot (g \cdot h)$,
for composable morphisms $f, g$ and $h$, the presence of 2-morphisms allow for
the weaker condition that these two expressions are merely isomorphic:
 \[ a(f,g,h) \maps (f \cdot g) \cdot h \Rightarrow f \cdot (g \cdot h), \]
where $a(f,g,h)$ is an 2-isomorphism called the \define{associator}. In the
same vein, rather than demanding that:
\[ 1_x \cdot f = f = f \cdot 1_y, \]
for $f \maps x \to y$, and identities $1_x \maps x \to x$ and $1_y \maps y \to
y$, the presence of 2-morphisms allows us to weaken these equations to
isomorphisms:
\[ l(f) \maps 1_x \cdot f \Rightarrow f, \quad r(f) \maps f \cdot 1_y \Rightarrow f. \]
Here, $l(f)$ and $r(f)$ are 2-isomorphisms called the \define{left and right
unitors}.

Of course, these 2-isomorphisms obey rules of their own. The associator
satisfies its own axiom, called the \define{pentagon identity}, which says that
this pentagon commutes:
\[
\xy
 (0,20)*+{(f g) (h k)}="1";
 (40,0)*+{f (g (h k))}="2";
 (25,-20)*{ \quad f ((g h) k)}="3";
 (-25,-20)*+{(f (g h)) k}="4";
 (-40,0)*+{((f g) h) k}="5";
 {\ar@{=>}^{a(f,g,h k)}     "1";"2"}
 {\ar@{=>}_{1_f \cdot a_(g,h,k)}  "3";"2"}
 {\ar@{=>}^{a(f,g h,k)}    "4";"3"}
 {\ar@{=>}_{a(f,g,h) \cdot 1_k}  "5";"4"}
 {\ar@{=>}^{a(fg,h,k)}    "5";"1"}
\endxy
\]

Finally, the associator and left and right unitors satisfy the \define{triangle
identity}, which says the following triangle commutes:
\[ 
\xy
(-20,10)*+{(f 1) g}="1";
(20,10)*+{f (1 g)}="2";
(0,-10)*+{f g}="3";
{\ar@{=>}^{a(f,1,g)}	"1";"2"}
{\ar@{=>}_{r(f) \cdot 1_g}	"1";"3"}
{\ar@{=>}^{1_f \cdot l(g)} "2";"3"}
\endxy
\]

For a \define{smooth bicategory} $B$, the collection of objects, $B_0$, collection of
morphisms $B_1$, and collection of 2-morphisms $B_2$, are all smooth manifolds, and
all structure is given by smooth maps. In particular, $B$ is equipped with:
\begin{itemize}
	\item a smooth category structure on $\underline{\Mor} B$, with
		\begin{itemize}
			\item $B_1$ as the smooth manifold of objects;
			\item $B_2$ as the smooth manifold of morphisms;
		\end{itemize}
		The composition in $\underline{\Mor} B$ is called
		\define{vertical composition} and denoted $\circ$.
	\item smooth \define{source} and \define{target maps}:
		\[ s, t \maps B_1 \to B_0. \]
	\item a smooth \define{identity-assigning map}:
		\[ i \maps B_0 \to B_1. \]
	\item a smooth \define{horizontal composition} functor:
		\[ \cdot \maps \underline{\Mor} B \times_{B_0} \underline{\Mor} B \to \underline{\Mor} B . \]
	\item a smooth natural isomorphism, the \define{associator}:
		\[ a(f,g,h) \maps (f \cdot g) \cdot h \To f \cdot (g \cdot h). \]
	\item smooth natural isomorphisms, the \define{left} and
		\define{right unitors}:
		\[ l(f) \maps 1 \cdot f \To f, \quad r(f) \maps f \cdot 1 \To f. \]
\end{itemize}

\subsection{Smooth tricategories} \label{sec:smoothtricats}

We now sketch the definition of tricategories, originally defined by Gordon, Power
and Street \cite{GPS}, but extensively studied by Gurski. We use the definition
from his thesis \cite{Gurski}.

First, we approach the definition informally. A `tricategory' has objects:
\[ x \, \bullet, \]
morphisms going between objects,
\[ \xymatrix{ x \, \bullet \ar[r]^f & \bullet \, y}, \]
2-morphisms going between morphisms,
\[
\xy
(-10,0)*+{x};
(-8,0)*+{\bullet}="4";
(8,0)*+{\bullet}="6";
(10,0)*+{y};
{\ar@/^1.65pc/^f "4";"6"};
{\ar@/_1.65pc/_g "4";"6"};
{\ar@{=>}^{\scriptstyle \alpha} (0,3)*{};(0,-3)*{}} ;
\endxy ,
\] 
and 3-morphisms going between 2-morphisms:
\[
\xy 
(-12,0)*+{x \, \bullet}="1";
(12,0)*+{\bullet \, y}="2";
{\ar@/^1.65pc/^f "1";"2"};
{\ar@/_1.65pc/_g "1";"2"};
(0,5)*+{}="A";
(0,-5)*+{}="B";
{\ar@{=>}@/_.75pc/ "A"+(-1.33,0) ; "B"+(-.66,-.55)};
(-6.5,0)*{\scriptstyle \alpha};
{\ar@{=}@/_.75pc/ "A"+(-1.33,0) ; "B"+(-1.33,0)};
{\ar@{=>}@/^.75pc/ "A"+(1.33,0) ; "B"+(.66,-.55)};
(6.5,0)*{\scriptstyle \beta};
{\ar@{=}@/^.75pc/ "A"+(1.33,0) ; "B"+(1.33,0)};
{\ar@3{->} (-2,0)*{}; (2,0)*{}};
(0,2.5)*{\scriptstyle \Gamma};
\endxy .
\]
Morphisms in a bicategory can be composed just as morphisms in a bicategory:
\[ \xymatrix{ x \ar[r]^f & y \ar[r]^g & z } \quad = \quad \xymatrix{ x \ar[r]^{f \cdot g} & z } . \]
And 2-morphisms can be composed in two different ways, again just as in a
bicategory. They can be composed at a morphism:
\[
\xy
(-8,0)*+{x}="4";
(8,0)*+{y}="6";
{\ar^g "4";"6"};
{\ar@/^1.75pc/^{f} "4";"6"};
{\ar@/_1.75pc/_{h} "4";"6"};
{\ar@{=>}^<<{\scriptstyle \alpha} (0,6)*{};(0,1)*{}} ;
{\ar@{=>}^<<{\scriptstyle \beta} (0,-1)*{};(0,-6)*{}} ;
\endxy
\quad = \quad \xy
(-8,0)*+{x}="4";
(8,0)*+{y}="6";
{\ar@/^1.65pc/^f "4";"6"};
{\ar@/_1.65pc/_h "4";"6"};
{\ar@{=>}^{\scriptstyle \alpha \circ \beta } (0,3)*{};(0,-3)*{}} ;
\endxy
\] 
or at an object:
\[
\xy
(-16,0)*+{x}="4";
(0,0)*+{y}="6";
{\ar@/^1.65pc/^{f} "4";"6"};
{\ar@/_1.65pc/_{g} "4";"6"};
{\ar@{=>}^<<<{\scriptstyle \alpha} (-8,3)*{};(-8,-3)*{}} ;
(0,0)*+{y}="4";
(16,0)*+{z}="6";
{\ar@/^1.65pc/^{f'} "4";"6"};
{\ar@/_1.65pc/_{g'} "4";"6"};
{\ar@{=>}^<<<{\scriptstyle \beta} (8,3)*{};(8,-3)*{}} ;
\endxy
\quad = \quad \xy
(-10,0)*+{x}="4";
(10,0)*+{z}="6";
{\ar@/^1.65pc/^{f \cdot f'} "4";"6"};
{\ar@/_1.65pc/_{g \cdot g'} "4";"6"};
{\ar@{=>}^{\alpha \cdot \beta} (0,3)*{};(0,-3)*{}} ;
\endxy .
\]
But now 3-morphisms can be composed in three different ways; we can compose two 3-morphisms at a 2-morphism, which we call \define{composition at a 2-cell}:
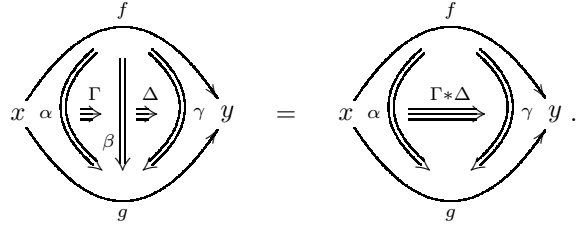
\[
\xy 0;/r.22pc/:
(0,15)*{};
(0,-15)*{};
(0,8)*{}="A";
(0,-8)*{}="B";
{\ar@{=>} "A" ; "B"};
{\ar@{=>}@/_1pc/ "A"+(-4,1) ; "B"+(-3,0)};
{\ar@{=}@/_1pc/ "A"+(-4,1) ; "B"+(-4,1)};
{\ar@{=>}@/^1pc/ "A"+(4,1) ; "B"+(3,0)};
{\ar@{=}@/^1pc/ "A"+(4,1) ; "B"+(4,1)};
{\ar@3{->} (-6,0)*{} ; (-2,0)*+{}};
(-4,3)*{\scriptstyle \Gamma};
{\ar@3{->} (2,0)*{} ; (6,0)*+{}};
(4,3)*{\scriptstyle \Delta};
(-15,0)*+{x}="1";
(15,0)*+{y}="2";
{\ar@/^2.75pc/^f "1";"2"};
{\ar@/_2.75pc/_g "1";"2"};
(-11,0)*{\scriptstyle \alpha};
(-2,-4)*{\scriptstyle \beta};
(11,0)*{\scriptstyle \gamma};
\endxy 
\quad = \quad
\xy 0;/r.22pc/:
(0,15)*{};
(0,-15)*{};
(0,8)*{}="A";
(0,-8)*{}="B";
{\ar@{=>}@/_1pc/ "A"+(-4,1) ; "B"+(-3,0)};
{\ar@{=}@/_1pc/ "A"+(-4,1) ; "B"+(-4,1)};
{\ar@{=>}@/^1pc/ "A"+(4,1) ; "B"+(3,0)};
{\ar@{=}@/^1pc/ "A"+(4,1) ; "B"+(4,1)};
{\ar@3{->} (-6,0)*{} ; (6,0)*+{}};
(0,3)*{\scriptstyle \Gamma * \Delta};
(-15,0)*+{x}="1";
(15,0)*+{y}="2";
{\ar@/^2.75pc/^f "1";"2"};
{\ar@/_2.75pc/_g "1";"2"};
(-11,0)*{\scriptstyle \alpha};
(11,0)*{\scriptstyle \gamma};
\endxy .
\]
We can also compose two 3-morphisms at a morphism, which we call \define{composition at a 1-cell}:
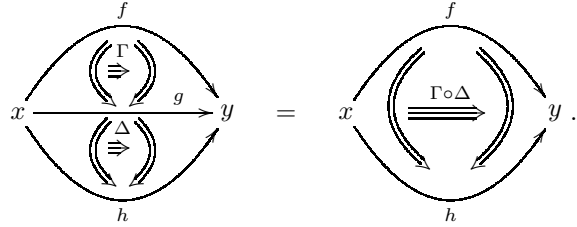
\[	
\xy 0;/r.22pc/:
(0,15)*{};
(0,-15)*{};
(0,9)*{}="A";
(0,1)*{}="B";
{\ar@{=>}@/_.5pc/ "A"+(-2,1) ; "B"+(-1,0)};
{\ar@{=}@/_.5pc/ "A"+(-2,1) ; "B"+(-2,1)};
{\ar@{=>}@/^.5pc/ "A"+(2,1) ; "B"+(1,0)};
{\ar@{=}@/^.5pc/ "A"+(2,1) ; "B"+(2,1)};
{\ar@3{->} (-2,6)*{} ; (2,6)*+{}};
(0,9)*{\scriptstyle \Gamma};
(0,-1)*{}="A";
(0,-9)*{}="B";
{\ar@{=>}@/_.5pc/ "A"+(-2,-1) ; "B"+(-1,-1.5)};
{\ar@{=}@/_.5pc/ "A"+(-2,0) ; "B"+(-2,-.7)};
{\ar@{=>}@/^.5pc/ "A"+(2,-1) ; "B"+(1,-1.5)};
{\ar@{=}@/^.5pc/ "A"+(2,0) ; "B"+(2,-.7)};
{\ar@3{->} (-2,-5)*{} ; (2,-5)*+{}};
(0,-2)*{\scriptstyle \Delta};
(-15,0)*+{x}="1";
(15,0)*+{y}="2";
{\ar@/^2.75pc/^f "1";"2"};
{\ar@/_2.75pc/_h "1";"2"};
{\ar "1";"2"};
(8,2)*{\scriptstyle g};
\endxy 
\quad = \quad
\xy 0;/r.22pc/:
(0,15)*{};
(0,-15)*{};
(0,8)*{}="A";
(0,-8)*{}="B";
{\ar@{=>}@/_1pc/ "A"+(-4,1) ; "B"+(-3,0)};
{\ar@{=}@/_1pc/ "A"+(-4,1) ; "B"+(-4,1)};
{\ar@{=>}@/^1pc/ "A"+(4,1) ; "B"+(3,0)};
{\ar@{=}@/^1pc/ "A"+(4,1) ; "B"+(4,1)};
{\ar@3{->} (-6,0)*{} ; (6,0)*+{}};
(0,3)*{\scriptstyle \Gamma \circ \Delta};
(-15,0)*+{x}="1";
(15,0)*+{y}="2";
{\ar@/^2.75pc/^f "1";"2"};
{\ar@/_2.75pc/_h "1";"2"};
\endxy .
\]
And finally, we can compose two 3-morphisms at an object, which we call \define{composition at a 0-cell}:
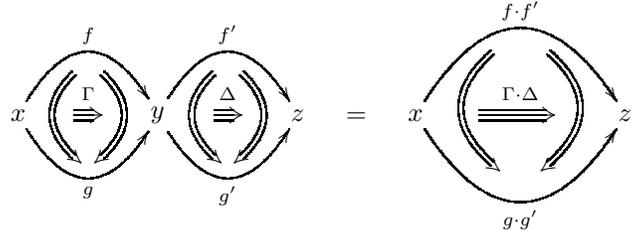
\[
\xy 0;/r.22pc/:
(0,15)*{};
(0,-15)*{};
(-20,0)*+{x}="1";
(0,0)*+{y}="2";
{\ar@/^2pc/^f "1";"2"};
{\ar@/_2pc/_g "1";"2"};
(20,0)*+{z}="3";
{\ar@/^2pc/^{f'} "2";"3"};
{\ar@/_2pc/_{g'} "2";"3"};
(-10,6)*+{}="A";
(-10,-6)*+{}="B";
{\ar@{=>}@/_.7pc/ "A"+(-2,0) ; "B"+(-1,-.8)};
{\ar@{=}@/_.7pc/ "A"+(-2,0) ; "B"+(-2,0)};
{\ar@{=>}@/^.7pc/ "A"+(2,0) ; "B"+(1,-.8)};
{\ar@{=}@/^.7pc/ "A"+(2,0) ; "B"+(2,0)};
(10,6)*+{}="A";
(10,-6)*+{}="B";
{\ar@{=>}@/_.7pc/ "A"+(-2,0) ; "B"+(-1,-.8)};
{\ar@{=}@/_.7pc/ "A"+(-2,0) ; "B"+(-2,0)};
{\ar@{=>}@/^.7pc/ "A"+(2,0) ; "B"+(1,-.8)};
{\ar@{=}@/^.7pc/ "A"+(2,0) ; "B"+(2,0)};
{\ar@3{->} (-12,0)*{}; (-8,0)*{}};
(-10,3)*{\scriptstyle \Gamma};
{\ar@3{->} (8,0)*{}; (12,0)*{}};
(10,3)*{\scriptstyle \Delta};
\endxy 
\quad = \quad 
\xy 0;/r.22pc/:
(0,15)*{};
(0,-15)*{};
(0,8)*{}="A";
(0,-8)*{}="B";
{\ar@{=>}@/_1pc/ "A"+(-4,1) ; "B"+(-3,0)};
{\ar@{=}@/_1pc/ "A"+(-4,1) ; "B"+(-4,1)};
{\ar@{=>}@/^1pc/ "A"+(4,1) ; "B"+(3,0)};
{\ar@{=}@/^1pc/ "A"+(4,1) ; "B"+(4,1)};
{\ar@3{->} (-6,0)*{} ; (6,0)*+{}};
(0,3)*{\scriptstyle \Gamma \cdot \Delta};
(-15,0)*+{x}="1";
(15,0)*+{z}="2";
{\ar@/^2.75pc/^{f \cdot f'} "1";"2"};
{\ar@/_2.75pc/_{g \cdot g'} "1";"2"};
\endxy .
\]
As in a bicategory, a tricategory has an associator and left and right unitors:
\[ a(f,g,h) \maps (f \cdot g) \cdot h \Rightarrow f \cdot (g \cdot h), \quad  l(f) \maps 1_x \cdot f \Rightarrow f, \quad r(f) \maps f \cdot 1_y \Rightarrow f. \]
But the presence of 3-morphisms allows us to \emph{weaken the axioms}; the associator
only satisfies the pentagon identity up to a 3-isomorphism we call the
\define{pentagonator}:
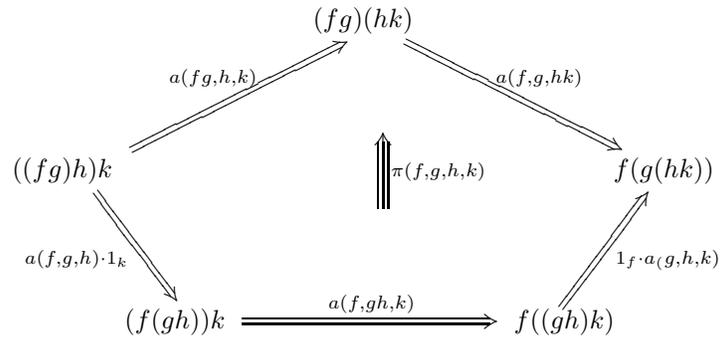
\[
\xy
(0,20)*+{(f g) (h k)}="1";
(40,0)*+{f (g (h k))}="2";
(25,-20)*{ \quad f ((g h) k)}="3";
(-25,-20)*+{(f (g h)) k}="4";
(-40,0)*+{((f g) h) k}="5";
{\ar@{=>}^{a(f,g,h k)}     "1";"2"}
{\ar@{=>}_{1_f \cdot a_(g,h,k)}  "3";"2"}
{\ar@{=>}^{a(f,g h,k)}    "4";"3"}
{\ar@{=>}_{a(f,g,h) \cdot 1_k}  "5";"4"}
{\ar@{=>}^{a(f g,h,k)}    "5";"1"}
{\ar@3{<-}^{\pi(f,g,h,k)} (-2,5)*{}; (-2,-5)*{} }
\endxy
\]
and the triangle axiom is replaced by three distinct 3-isomorphisms, called the
`triangulators'. These 3-isomorphisms must in turn satisfy axioms of their own.

Of course, we actually aim to define a `smooth tricategory', $T$. This has a smooth
manifold of objects, $T_0$, a smooth manifold of morphisms, $T_1$, a smooth manifold
of 2-morphisms, $T_2$, and a smooth manifold of 3-morphisms, $T_3$, such that:
\begin{itemize}
	\item $T_1$, $T_2$ and $T_3$ fit together to form a smooth bicategory;
	\item composition at a 0-cell is a smooth `pseudofunctor';
	\item satisfying associativity and left and right unit laws up to
          smooth `adjoint equivalences', the associator and left and right unitors;
	\item satisfying the pentagon and triangle identities up to invertible smooth `modifications';
	\item satisfying some identities of their own.
\end{itemize}
Each of the above quoted terms---pseudofunctor, adjoint equivalences,
modification---would usually need to be defined completely in order to understand
tricategories. For us, the key new concept is that of modification, because our
functors will not be pseudo, and our adjoint equivalences will be
identities. Nonetheless, let us describe each of these concepts in turn.

\begin{itemize}
	\item `Pseudofunctor' is to `bicategory' as `functor' is to `category':
		it is a map $F \maps B \to B'$ between bicategories $B$ and
		$B'$, preserving all structure in sight \emph{except}
		horizontal composition and identities, which are only preserved
		up to specified 2-isomorphisms:
		\[ F(f \cdot g) \To F(f) \cdot F(g), \quad F(1_x) \To 1_{F(x)} . \] 
		For the tricategories we construct, all pseudofunctors will be
		strict: the above 2-isomorphisms are identities.

	\item `Adjoint equivalences' involve both `pseudonatural transformations' and `modifications':

          \begin{itemize}
            \item  `Pseudonatural transformation' is to `pseudofunctor' as `natural
		transformation' is to 'functor': given two pseudofunctors 
		
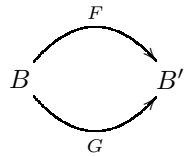
\[
		\xy 
		(-10,0)*+{B}="1";
		(10,0)*+{B'}="2";
		{\ar@/^1.65pc/^F "1";"2"};
		{\ar@/_1.65pc/_G "1";"2"};
		(0,5)*+{}="A";
		(0,-5)*+{}="B";
		\endxy
		\]
		a pseudonatural transformation is a map:
		
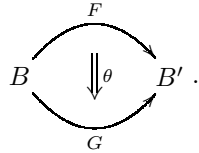
\[
		\xy 
		(-10,0)*+{B}="1";
		(10,0)*+{B'}="2";
		{\ar@/^1.65pc/^{F} "1";"2"};
		{\ar@/_1.65pc/_{G} "1";"2"};
		(0,5)*+{}="A";
		{\ar@{=>}^{\scriptstyle \theta} (0,3)*{};(0,-3)*{}} ;
		(0,-5)*+{}="B";
		\endxy .
		\]
		Like a natural transformation, this consists of a morphism for
		each object $x$ in $B$:
		\[ \theta(x) \maps F(x) \to G(x). \]
		Unlike a natural transformation, it is only natural up to a
		specified 2-isomorphism. That is, the naturality square:
		\[ \xymatrix{ F(x) \ar[r]^{F(f)} \ar[d]_{\theta(x)} & F(y) \ar[d]^{\theta(y)} \\
			G(x) \ar[r]^{G(f)} & G(y) \\
		}
		\]
		does \emph{not} commute. It is replaced with a 2-isomorphism:
		\[ \xymatrix{ F(x) \ar[r]^{F(f)} \ar[d]_{\theta(x)} & F(y) \ar[d]^{\theta(y)} \\
		\ar@{=>}[ur]_{\theta(f)} G(x) \ar[r]_{G(f)} & G(y) \\
			}
		\]
		that satisfies some equations of its own. For the tricategories
		we construct, all pseudonatural transformations will be strict:
		the 2-isomorphism above is the identity.
	\item A `modification' is something new: it is a map between
		pseudonatural transformations. Given two pseudonatural
		transformations:
		\[
		\xy 
		(-12,0)*+{B}="1";
		(12,0)*+{B'}="2";
		{\ar@/^1.65pc/^F "1";"2"};
		{\ar@/_1.65pc/_G "1";"2"};
		(0,5)*+{}="A";
		(0,-5)*+{}="B";
		{\ar@{=>}@/_.75pc/ "A"+(-1.33,0) ; "B"+(-.66,-.55)};
		(-6.5,0)*{\scriptstyle \theta};
		{\ar@{=}@/_.75pc/ "A"+(-1.33,0) ; "B"+(-1.33,0)};
		{\ar@{=>}@/^.75pc/ "A"+(1.33,0) ; "B"+(.66,-.55)};
		(6.5,0)*{\scriptstyle \eta};
		{\ar@{=}@/^.75pc/ "A"+(1.33,0) ; "B"+(1.33,0)};
		\endxy
		\]
		a modification $\Gamma$ is a map:
		\[
		\xy 
		(-12,0)*+{B}="1";
		(12,0)*+{B'}="2";
		{\ar@/^1.65pc/^F "1";"2"};
		{\ar@/_1.65pc/_G "1";"2"};
		(0,5)*+{}="A";
		(0,-5)*+{}="B";
		{\ar@{=>}@/_.75pc/ "A"+(-1.33,0) ; "B"+(-.66,-.55)};
		(-6.5,0)*{\scriptstyle \theta};
		{\ar@{=}@/_.75pc/ "A"+(-1.33,0) ; "B"+(-1.33,0)};
		{\ar@{=>}@/^.75pc/ "A"+(1.33,0) ; "B"+(.66,-.55)};
		(6.5,0)*{\scriptstyle \eta};
		{\ar@{=}@/^.75pc/ "A"+(1.33,0) ; "B"+(1.33,0)};
		{\ar@3{->} (-2,0)*{}; (2,0)*{}};
		(0,2.5)*{\scriptstyle \Gamma};
		\endxy .
		\]
		Just as a pseudonatural transformation consists of a morphism for
		each object $x$ in $B$, a modification consists of a 2-morphism
		for each object $x$ in $B$:
		\[
		\xy 
		(-12,0)*+{F(x)}="1";
		(12,0)*+{G(x)}="2";
		{\ar@/^1.65pc/^{\theta(x)} "1";"2"};
		{\ar@/_1.65pc/_{\eta(x)} "1";"2"};
		(0,5)*+{}="A";
		{\ar@{=>}^{\scriptstyle \Gamma(x)} (0,3)*{};(0,-3)*{}} ;
		(0,-5)*+{}="B";
		\endxy .
		\]
		These 2-morphisms satisfy an equation that will hold trivially
		in the tricategories we consider, so we omit it. See Leinster
		\cite{Leinster:bicat} for more details.

              \end{itemize}
\end{itemize}
Finally, an \define{adjoint equivalence} is a quadruple $(\theta, \theta^{-1}, e, u)$
consisting of two pseudonatural transformations, $\theta$ and $\theta^{-1}$, going in
opposite directions:
\[
	\xy 
	(-10,0)*+{B}="1";
	(10,0)*+{B'}="2";
	{\ar@/^1.65pc/^{F} "1";"2"};
	{\ar@/_1.65pc/_{G} "1";"2"};
	(0,5)*+{}="A";
	{\ar@{=>}^{\scriptstyle \theta} (0,3)*{};(0,-3)*{}} ;
	(0,-5)*+{}="B";
	\endxy
        \quad \quad \quad
	\xy 
	(-10,0)*+{B}="1";
	(10,0)*+{B'}="2";
	{\ar@/^1.65pc/^{F} "1";"2"};
	{\ar@/_1.65pc/_{G} "1";"2"};
	(0,5)*+{}="A";
	{\ar@{<=}^{\scriptstyle \theta^{-1}} (0,3)*{};(0,-3)*{}} ;
	(0,-5)*+{}="B";
	\endxy 
\]
along with invertible modifications, $e$ and $u$, that exhibit $\theta$ and $\theta^{-1}$ as weak inverses to each other:
\[ e \maps \theta^{-1} \theta \Rrightarrow 1_F, \quad u \maps 1_G \Rrightarrow \theta \theta^{-1} . \]
We further demand that $e$ and $u$ also satisfy the \define{zig-zag identities}:
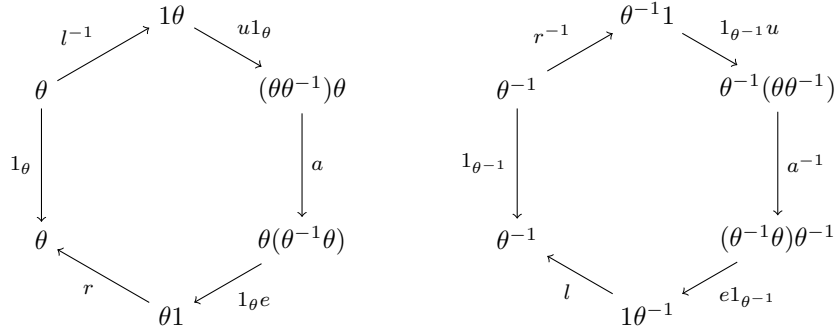
\begin{center}
  \begin{tikzpicture}
    \node (A) at (30:2cm) {$(\theta\theta^{-1})\theta$};
    \node (B) at (90:2cm) {$1 \theta$}
      edge [->] node [above right, l] {$u 1_\theta$} (A);
    \node (C) at (150:2cm) {$\theta$}
      edge [->] node [above left, l] {$l^{-1}$} (B);
    \node (D) at (210:2cm) {$\theta$}
      edge [<-] node [left, l] {$1_\theta$} (C);
    \node (E) at (270:2cm) {$\theta 1$}
      edge [->] node [below left, l] {$r$} (D);
    \node (F) at (330:2cm) {$\theta(\theta^{-1}\theta)$}
      edge [->] node [below right, l] {$1_\theta e$} (E)
      edge [<-] node [right, l] {$a$} (A);
  \end{tikzpicture}
  $\quad\quad\quad$
  \begin{tikzpicture}
    \node (A) at (30:2cm) {$\theta^{-1} (\theta\theta^{-1})$};
    \node (B) at (90:2cm) {$\theta^{-1} 1$}
      edge [->] node [above right, l] {$1_{\theta^{-1}} u$} (A);
    \node (C) at (150:2cm) {$\theta^{-1}$}
      edge [->] node [above left, l] {$r^{-1}$} (B);
    \node (D) at (210:2cm) {$\theta^{-1}$}
      edge [<-] node [left, l] {$1_{\theta^{-1}}$} (C);
    \node (E) at (270:2cm) {$1 \theta^{-1}$}
      edge [->] node [below left, l] {$l$} (D);
    \node (F) at (330:2cm) {$(\theta^{-1}\theta) \theta^{-1}$}
      edge [->] node [below right, l] {$e 1_{\theta^{-1}}$} (E)
      edge [<-] node [right, l] {$a^{-1}$} (A);
  \end{tikzpicture}
\end{center}
In these diagrams, we have drawn the triple arrows for the modifications as single
arrows to avoid clutter. Juxtaposition of modifications comes from horizontal
composition of 2-morphisms in the target bicategory, along with the associator and
unitors in that bicategory. The word ``adjoint'' appears because these identities are
analogous to those satisfied by the unit and counit of an adjoint pair of
functors. While an adjoint equivalence refers to the entire quadruple $(\theta,
\theta^{-1}, e, u)$, we will often abuse terminology and mention only $\theta$.

With these preliminaries in mind, we can now sketch the definition of a smooth
tricategory. 
\begin{defn} \label{def:smoothtricat}
	A \define{smooth tricategory} $T$ consists of: 
\begin{itemize}
	\item a \define{manifold of objects}, $T_0$;
	\item a \define{manifold of morphisms}, $T_1$;
	\item a \define{manifold of 2-morphisms}, $T_2$;
	\item a \define{manifold of 3-morphisms}, $T_3$;
\end{itemize}
equipped with:
\begin{itemize}
\item a smooth bicategory structure on $\underline{\Mor} \, T$, with
  \begin{itemize}
  \item $T_1$ as the smooth manifold of objects;
  \item $T_2$ as the smooth manifold of morphisms;
  \item $T_3$ as the smooth manifold of 2-morphisms;
  \end{itemize}
  We call the vertical composition in $\underline{\Mor} \, T$ \define{composition at
  a 2-cell}, and the horizontal composition in $\underline{\Mor} \, T$
  \define{composition at a 1-cell}.
\item smooth \define{source} and \define{target maps}:
  \[ s, t \maps T_1 \to T_0. \]
\item a smooth \define{identity-assigning map}:
  \[ i \maps T_0 \to T_1. \]
\item a smooth \define{composition} pseudofunctor, called \define{composition at a
  0-cell}, which is strict in the tricategories we consider:
  \[ \cdot \maps \underline{\Mor} \, T \times_{T_0} \underline{\Mor} \, T \to \underline{\Mor} \, T. \]
  That is, three smooth maps:
  \[ \cdot \maps T_1 \times_{T_0} T_1 \to T_1 \]
  \[ \cdot \maps T_2 \times_{T_0} T_2 \to T_2 \]
  \[ \cdot \maps T_3 \times_{T_0} T_3 \to T_3 \]
  satisfying the axioms of a strict functor. Here, the pullbacks give us the spaces
  of 1-, 2- and 3-morphisms which are composable at a 0-cell, and we assume they are
  smooth manifolds.
\item a smooth adjoint equivalence, the \define{associator}, trivial in the
  tricategories we consider:
  \[ a(f,g,h) \maps (f \cdot g) \cdot h \To f \cdot (g \cdot h). \]
\item smooth adjoint equivalences, the \define{left} and \define{right unitors}, both
  trivial in the tricategories we consider:
  \[ l(f) \maps 1 \cdot f \To f, \quad r(f) \maps f \cdot 1 \To f. \]
\item a smooth invertible modification called the \define{pentagonator}:
  
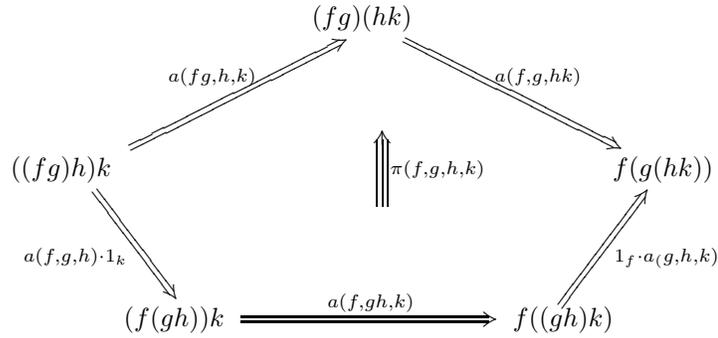
\[
	\xy
	 (0,20)*+{(f g) (h k)}="1";
	 (40,0)*+{f (g (h k))}="2";
	 (25,-20)*{ \quad f ((g h) k)}="3";
	 (-25,-20)*+{(f (g h)) k}="4";
	 (-40,0)*+{((f g) h) k}="5";
	 {\ar@{=>}^{a(f,g,h k)}     "1";"2"}
	 {\ar@{=>}_{1_f \cdot a_(g,h,k)}  "3";"2"}
	 {\ar@{=>}^{a(f,g h,k)}    "4";"3"}
	 {\ar@{=>}_{a(f,g,h) \cdot 1_k}  "5";"4"}
	 {\ar@{=>}^{a(f g,h,k)}    "5";"1"}
	 {\ar@3{<-}^{\pi(f,g,h,k)} (-2,5)*{}; (-2,-5)*{} }
	\endxy
  \]
\item smooth invertible modifications called the \define{middle, left} and
  \define{right triangulators}, all trivial in the tricategories we consider:
  
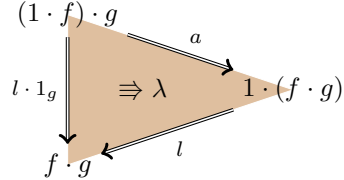
\begin{center}
    \begin{tikzpicture}
          \filldraw[white,fill=brown,fill opacity=0.5](0,2)--(3,1)--(0
,0)--cycle;
          \node (1fg1) at (0,2) {$(1\cdot f)\cdot g$};
          \node (1fg2) at (3,1) {$1\cdot (f \cdot g)$}
            edge [<-, double] node [l, above right] {$a$} (1fg1);
          \node (fg) at (0,0) {$f \cdot g$}
            edge [<-, double] node [l, below right] {$l$} (1fg2)
            edge [<-, double] node [l, left] {$l \cdot 1_g$} (1fg1);
          \node at (1,1) {$\Rrightarrow \lambda$};
        \end{tikzpicture}
      \end{center}
      
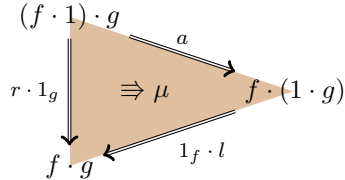
\begin{center}
        \begin{tikzpicture}
          \filldraw[white,fill=brown,fill opacity=0.5](0,2)--(3,1)--(0
,0)--cycle;
          \node (fg) at (0,0) {$f \cdot g$};
          \node (f1g2) at (3,1) {$f\cdot (1\cdot g)$}
            edge [->, double] node [l, below right] {$1_f \cdot l$} (fg);
          \node (f1g1) at (0,2) {$(f\cdot 1) \cdot g$}
            edge [->, double] node [l, above] {$a$} (f1g2)
            edge [->, double] node [l, left] {$r \cdot 1_g$} (fg);
          \node at (1,1) {$\Rrightarrow \mu$};
        \end{tikzpicture}
      \end{center}
      
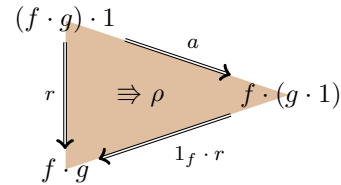
\begin{center}
        \begin{tikzpicture}
          \filldraw[white,fill=brown,fill opacity=0.5](0,2)--(3,1)--(0
,0)--cycle;
          \node (fg11) at (0,2) {$(f \cdot g) \cdot 1$};
          \node (fg12) at (3,1) {$f \cdot (g \cdot 1)$}
            edge [<-, double] node [l, above right] {$a$} (fg11);
          \node (fg) at (0,0) {$f \cdot g$}
            edge [<-, double] node [l, below right] {$1_f \cdot r$} (fg12)
            edge [<-, double] node [l, left] {$r$} (fg11);
          \node at (1,1) {$\Rrightarrow \rho$};
        \end{tikzpicture}
      \end{center}
\end{itemize}

The pentagonator and triangulators satisfy some axioms. When $\lambda$, $\rho$ and
$\mu$ are trivial, as they are for us, the axioms involving the triangulators just
say $\pi$ is trivial whenever one of its arguments is 1, so we omit these. This
leaves one axiom, involving only the pentagonator. This key axiom is the
\define{pentagonator identity}:

    
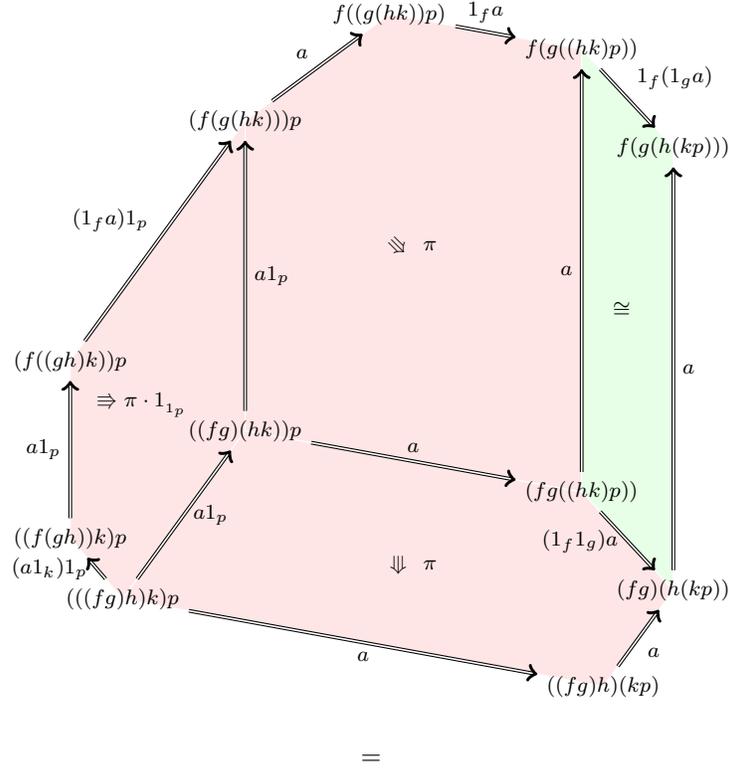
\begin{center}
      \begin{tikzpicture}[scale=.85]
        [line join=round]
        \filldraw[white,fill=red,fill opacity=0.1](-4.306,-3.532)--(-2.391,-.901)--(-2.391,3.949)--(-5.127,.19)--(-5.127,-2.581)--cycle;
        \filldraw[white,fill=red,fill opacity=0.1](-4.306,-3.532)--(-2.391,-.901)--(2.872,-1.858)--(4.306,-3.396)--(3.212,-4.9)--cycle;
        \filldraw[white,fill=red,fill opacity=0.1](2.872,-1.858)--(2.872,5.07)--(-.135,5.617)--(-2.391,3.949)--(-2.391,-.901)--cycle;
        \filldraw[white,fill=green,fill opacity=0.1](4.306,-3.396)--(4.306,3.532)--(2.872,5.07)--(2.872,-1.858)--cycle;
        \begin{scope}[font=\fontsize{8}{8}\selectfont]
          \node (A) at (-2.391,3.949) {$(f(g(hk)))p$};
          \node (B) at (-5.127,.19) {$(f((gh)k))p$}
	  edge [->, double] node [l, above left] {$(1_{f}  a)1_{p}$} (A);
          \node (C) at (-5.127,-2.581) {$((f(gh))k)p$}
	  edge [->, double] node [l, left] {$a  1_{p}$} (B);
          \node (D) at (-4.306,-3.532) {$(((fg)h)k)p$}
	  edge [->, double] node [l, left] {$(a  1_{k})  1_{p}$} (C);
          \node (E) at (3.212,-4.9) {$((fg)h)(kp)$}
            edge [<-, double] node [l, below] {$a$} (D);
          \node (F) at (4.306,-3.396) {$(fg)(h(kp))$}
            edge [<-, double] node [l, below right] {$a$} (E);
          \node (G) at (4.306,3.532) {$f(g(h(kp)))$}
            edge [<-, double] node [l, right] {$a$} (F);
          \node (H) at (2.872,5.07) {$f(g((hk)p))$}
            edge [->, double] node [l, above right] {$1_{f}(1_{g}a)$} (G);
          \node (I) at (-.135,5.617) {$f((g(hk))p)$}
            edge [->, double] node [l, above] {$1_{f}a$} (H)
            edge [<-, double] node [l, above left] {$a$} (A);
          \node (M) at (-2.391,-.901) {$((fg)(hk))p$}
            edge [<-, double] node [l, right] {$a1_{p}$} (D)
            edge [->, double] node [l, right] {$a1_{p}$} (A);
          \node (N) at (2.872,-1.858) {$(fg((hk)p))$}
            edge [<-, double] node [l, above] {$a$} (M)
            edge [->, double] node [l, left] {$a$} (H)
            edge [->, double] node [l, left] {$(1_{f}1_{g})a$} (F);
	    \node at (-4,-.5) {$\Rrightarrow \pi \cdot 1_{_{1_{p}}}$};
          \node at (0,-3) {\tikz\node [rotate=-90] {$\Rrightarrow$};};
          \node at (0.5,-3) {$\pi$};
          \node at (0,2) {\tikz\node [rotate=-45] {$\Rrightarrow$};};
          \node at (0.5,2) {$\pi$};
          \node at (3.5,1) {$\cong$};
        \end{scope}
      \end{tikzpicture}
      \[ = \]
      
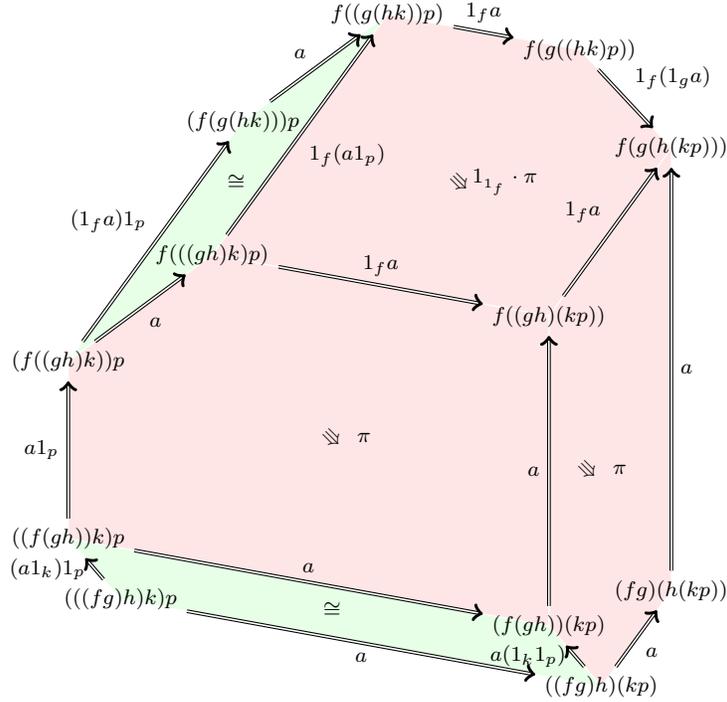
\begin{tikzpicture}[scale=.85]
        [line join=round]
        \filldraw[white,fill=green,fill opacity=0.1](-2.872,1.858)--(-.135,5.617)--(-2.391,3.949)--(-5.127,.19)--cycle;
        \filldraw[white,fill=red,fill opacity=0.1](3.212,-4.9)--(4.306,-3.396)--(4.306,3.532)--(2.391,.901)--(2.391,-3.949)--cycle;
        \filldraw[white,fill=green,fill opacity=0.1](-4.306,-3.532)--(3.212,-4.9)--(2.391,-3.949)--(-5.127,-2.581)--cycle;
        \filldraw[white,fill=red,fill opacity=0.1](-2.872,1.858)--(2.391,.901)--(4.306,3.532)--(2.872,5.07)--(-.135,5.617)--cycle;
        \filldraw[white,fill=red,fill opacity=0.1](-5.127,-2.581)--(-5.127,.19)--(-2.872,1.858)--(2.391,.901)--(2.391,-3.949)--cycle;
        \begin{scope}[font=\fontsize{8}{8}\selectfont]
          \node (A) at (-2.391,3.949) {$(f(g(h k)))p$};
          \node (B) at (-5.127,.19) {$(f((gh)k))p$}
            edge [->, double] node [l, above left] {$(1_{f}a)1_{p}$} (A);
          \node (C) at (-5.127,-2.581) {$((f(gh))k)p$}
            edge [->, double] node [l, left] {$a1_{p}$} (B);
          \node (D) at (-4.306,-3.532) {$(((fg)h)k)p$}
            edge [->, double] node [l, left] {$(a1_{k})1_{p}$} (C);
          \node (E) at (3.212,-4.9) {$((fg)h)(kp)$}
            edge [<-, double] node [l, below] {$a$} (D);
          \node (F) at (4.306,-3.396) {$(fg)(h(kp))$}
            edge [<-, double] node [l, below right] {$a$} (E);
          \node (G) at (4.306,3.532) {$f(g(h(kp)))$}
            edge [<-, double] node [l, right] {$a$} (F);
          \node (H) at (2.872,5.07) {$f(g((hk)p))$}
            edge [->, double] node [l, above right] {$1_{f}(1_{g}a)$} (G);
          \node (I) at (-.135,5.617) {$f((g(hk))p)$}
            edge [->, double] node [l, above] {$1_{f}a$} (H)
            edge [<-, double] node [l, above left] {$a$} (A);
          \node (J) at (-2.872,1.858) {$f(((gh)k)p)$}
            edge [->, double] node [l, below right] {$1_{f}(a1_{p})$} (I)
            edge [<-, double] node [l, below right] {$a$} (B);
          \node (K) at (2.391,-3.949) {$(f(gh))(kp)$}
            edge [<-, double] node [l, left] {$a(1_{k}1_{p})$} (E)
            edge [<-, double] node [l, above] {$a$} (C);
          \node (L) at (2.391,.901) {$f((gh)(kp))$}
            edge [<-, double] node [l, left] {$a$} (K)
            edge [<-, double] node [l, above] {$1_{f}a$} (J)
            edge [->, double] node [l, above left] {$1_{f}a$} (G);
          \node at (-1,-1) {\tikz\node [rotate=-45] {$\Rrightarrow$};};
	  \node at (-.5,-1) {$\pi$};
          \node at (1,3) {\tikz\node [rotate=-45] {$\Rrightarrow$};};
	  \node at (1.7,3) {$1_{_{1_{f}}} \cdot \pi$};
          \node at (3,-1.5) {\tikz\node [rotate=-45] {$\Rrightarrow$};};
          \node at (3.5,-1.5) {$\pi$};
          \node at (-1,-3.7) {$\cong$};
          \node at (-2.5,3) {$\cong$};
        \end{scope}
      \end{tikzpicture}
    \end{center}

\end{defn}

\noindent
In the pentagonator identity, we have omitted $\cdot$ everywhere except the faces to
save space, and the unlabeled squares are naturality squares for the associator,
which will be trivial in the examples we consider. This identity comes from a
3-dimensional solid called the \define{associahedron}. This is the polyhedron where:
\begin{itemize}
	\item vertices are parenthesized lists of five morphisms, e.g.\ $(((fg)h)k)p$;
	\item edges connect any two vertices related by an application of the associator, e.g.\ 
		\[ (((fg)h)k)p \Rightarrow ((fg)(hk))p . \]
\end{itemize}

In fact, the pentagonator identity gives us a picture of the associahedron.
Regarding the left-hand side of the equation as the back and the right-hand
side as the front, we assemble the following polyhedron:
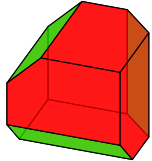
\begin{center}
  \begin{tikzpicture}[line join=round]
    \begin{scope}[scale=.2]
      \filldraw[fill=red,fill opacity=0.7](-4.306,-3.532)--(-2.391,-.901)--(-2.391,3.949)--(-5.127,.19)--(-5.127,-2.581)--cycle;
      \filldraw[fill=red,fill opacity=0.7](-4.306,-3.532)--(-2.391,-.901)--(2.872,-1.858)--(4.306,-3.396)--(3.212,-4.9)--cycle;
      \filldraw[fill=red,fill opacity=0.7](2.872,-1.858)--(2.872,5.07)--(-.135,5.617)--(-2.391,3.949)--(-2.391,-.901)--cycle;
      \filldraw[fill=green,fill opacity=0.7](4.306,-3.396)--(4.306,3.532)--(2.872,5.07)--(2.872,-1.858)--cycle;
      \filldraw[fill=green,fill opacity=0.7](-2.872,1.858)--(-.135,5.617)--(-2.391,3.949)--(-5.127,.19)--cycle;
      \filldraw[fill=red,fill opacity=0.7](3.212,-4.9)--(4.306,-3.396)--(4.306,3.532)--(2.391,.901)--(2.391,-3.949)--cycle;
      \filldraw[fill=green,fill opacity=0.7](-4.306,-3.532)--(3.212,-4.9)--(2.391,-3.949)--(-5.127,-2.581)--cycle;
      \filldraw[fill=red,fill opacity=0.7](-2.872,1.858)--(2.391,.901)--(4.306,3.532)--(2.872,5.07)--(-.135,5.617)--cycle;
      \filldraw[fill=red,fill opacity=0.7](-5.127,-2.581)--(-5.127,.19)--(-2.872,1.858)--(2.391,.901)--(2.391,-3.949)--cycle;
    \end{scope}
  \end{tikzpicture}
\end{center}
Identifying the vertices, edges and faces of this polyhedron with the
corresponding morphisms, 2-morphisms and 3-morphisms from the pentagonator
identity, we see the identity just says that the associahedron commutes.

\subsection{Inverses} \label{sec:inverses}

To talk about Lie 3-groups, we need to talk about inverses of morphisms, 2-morphisms
and 3-morphisms. This is complicated by the fact that inverses in a 3-group should be
weak inverses, satisfying the left and right inverse laws
\[ f^{-1} f = 1, \quad f f^{-1} = 1 \]
only up to higher morphisms. Fortunately, for the examples we care about, the
inverses are strict, and these laws hold, but nonetheless we will describe the
general situation. 

For 3-morphisms, there are no higher morphisms, so the inverse of a 3-morphism
$\Gamma$ is the usual, strict notion: a 3-morphism $\Gamma^{-1}$ satisfying
the left and right inverse laws:
\[ \Gamma^{-1} \Gamma = 1, \quad \Gamma \Gamma^{-1} = 1 .\]
Here, juxtaposition denotes composition at a 2-cell, and 1 is the identity of the
appropriate 2-cell. Provided such an inverse exists, it is automatically unique, so
we are justified in saying \emph{the} inverse.

For 2-morphisms, we can weaken the notion of inverse. A 2-morphism $\alpha$ has a
\define{weak inverse} $\alpha^{-1}$ if there are invertible 3-morphisms $e$ and $u$
which weaken the left and right inverse laws:
\[ e \maps \alpha^{-1} \alpha \Rrightarrow 1, \quad u \maps 1 \Rrightarrow \alpha \alpha^{-1} . \]
Now, juxtaposition denotes composition at a 1-cell, and the 1 denotes the identity of
the appropriate 1-cell. Unlike strict inverses, weak inverses are no longer unique,
but any two of them must be isomorphic. A 2-morphism $\alpha$ with a weak inverse is
called an \define{equivalence}. Yet, to maximize the utility of a weak inverse, it is
best if we impose some axioms on $u$ and $e$. We have already seen the desired axioms
in the last section, when we introduced the concept of an adjoint equivalence. They
are the \define{zig-zag identities}:
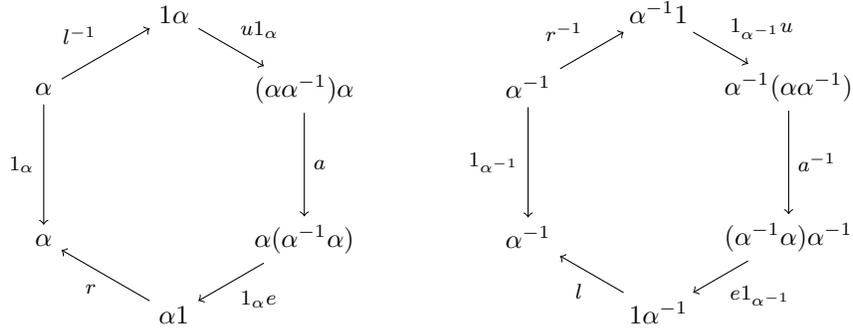
\begin{center}
  \begin{tikzpicture}
    \node (A) at (30:2cm) {$(\alpha\alpha^{-1})\alpha$};
    \node (B) at (90:2cm) {$1 \alpha$}
      edge [->] node [above right, l] {$u 1_\alpha$} (A);
    \node (C) at (150:2cm) {$\alpha$}
      edge [->] node [above left, l] {$l^{-1}$} (B);
    \node (D) at (210:2cm) {$\alpha$}
      edge [<-] node [left, l] {$1_\alpha$} (C);
    \node (E) at (270:2cm) {$\alpha 1$}
      edge [->] node [below left, l] {$r$} (D);
    \node (F) at (330:2cm) {$\alpha(\alpha^{-1}\alpha)$}
      edge [->] node [below right, l] {$1_\alpha e$} (E)
      edge [<-] node [right, l] {$a$} (A);
  \end{tikzpicture}
  $\quad\quad\quad$
  \begin{tikzpicture}
    \node (A) at (30:2cm) {$\alpha^{-1} (\alpha\alpha^{-1})$};
    \node (B) at (90:2cm) {$\alpha^{-1} 1$}
      edge [->] node [above right, l] {$1_{\alpha^{-1}} u$} (A);
    \node (C) at (150:2cm) {$\alpha^{-1}$}
      edge [->] node [above left, l] {$r^{-1}$} (B);
    \node (D) at (210:2cm) {$\alpha^{-1}$}
      edge [<-] node [left, l] {$1_{\alpha^{-1}}$} (C);
    \node (E) at (270:2cm) {$1 \alpha^{-1}$}
      edge [->] node [below left, l] {$l$} (D);
    \node (F) at (330:2cm) {$(\alpha^{-1}\alpha) \alpha^{-1}$}
      edge [->] node [below right, l] {$e 1_{\alpha^{-1}}$} (E)
      edge [<-] node [right, l] {$a^{-1}$} (A);
  \end{tikzpicture}
\end{center}
When $e$ and $u$ satisfy these identities, the quadruple $(\alpha, \alpha^{-1}, e,
u)$ is called an \define{adjoint equivalence}, though we generally abuse terminology
and say $\alpha$ itself is an adjoint equivalence. By adjusting $e$ and $u$, every
equivalence can be made into an adjoint equivalence. Unlike the choice of
$\alpha^{-1}$, which is only unique up to isomorphism, the choice of $(\alpha^{-1},
e, u)$ making $\alpha$ an adjoint equivalence is unique up to canonical isomorphism.

Moving down to morphisms, we can weaken the notion of invertibility even
further. Whereas a 2-morphism had a weak inverse up to isomorphism, a morphism has a
weak inverse up to equivalence. That is, a morphism $f$ has a weak inverse $f^{-1}$
if there are equivalences $\epsilon$ and $\eta$ weakening the left and right inverse
laws:
\[ \epsilon \maps f^{-1} f \To 1, \quad \eta \maps 1 \To f f^{-1} \]
Of course, the weak inverse of a morphism is only unique up to equivalence. A
morphism that has a weak inverse is called a \define{biequivalence}. Just as every
equivalence can be improved by making it part of an adjoint equivalence, every
biequivalence can be improved by making it a part of a `biadjoint biequivalence'
\cite{Gurski:biequivalences}. A \define{biadjoint biequivalence} is a biequivalence where
$\epsilon$ and $\eta$ are both adjoint equivalences, satisfying the zig-zag identities up
to invertible 3-morphisms:
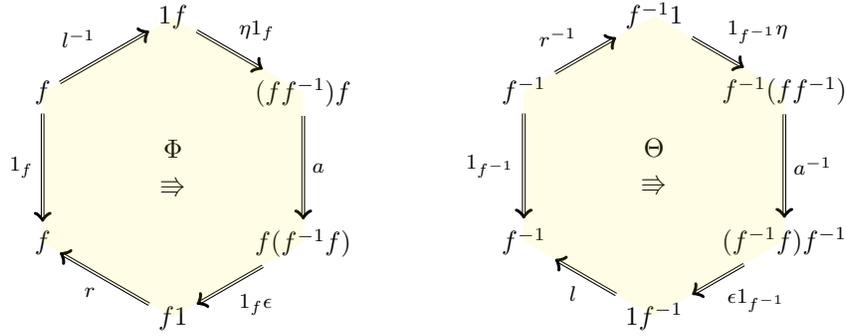
\begin{center}
  \begin{tikzpicture}
    \node (A) at (30:2cm) {$(ff^{-1})f$};
    \node (B) at (90:2cm) {$1 f$}
      edge [->, double] node [above right, l] {$\eta 1_f$} (A);
    \node (C) at (150:2cm) {$f$}
      edge [->, double] node [above left, l] {$l^{-1}$} (B);
    \node (D) at (210:2cm) {$f$}
      edge [<-, double] node [left, l] {$1_f$} (C);
    \node (E) at (270:2cm) {$f 1$}
      edge [->, double] node [below left, l] {$r$} (D);
    \node (F) at (330:2cm) {$f(f^{-1}f)$}
      edge [->, double] node [below right, l] {$1_f \epsilon$} (E)
      edge [<-, double] node [right, l] {$a$} (A);
    \node at (90:.25cm) {$\Phi$};
    \node at (-90:.25cm) {$\Rrightarrow$};
    \begin{scope}
      \fill [white,fill=yellow,fill opacity=0.1] (A.center) to
(B.center) to (C.center) to (D.center) to (E.center) to (F.center) to
(A.center);
    \end{scope}
  \end{tikzpicture}
  $\quad\quad\quad$
  \begin{tikzpicture}
    \node (A) at (30:2cm) {$f^{-1} (ff^{-1})$};
    \node (B) at (90:2cm) {$f^{-1} 1$}
      edge [->, double] node [above right, l] {$1_{f^{-1}} \eta$} (A);
    \node (C) at (150:2cm) {$f^{-1}$}
      edge [->, double] node [above left, l] {$r^{-1}$} (B);
    \node (D) at (210:2cm) {$f^{-1}$}
      edge [<-, double] node [left, l] {$1_{f^{-1}}$} (C);
    \node (E) at (270:2cm) {$1 f^{-1}$}
      edge [->, double] node [below left, l] {$l$} (D);
    \node (F) at (330:2cm) {$(f^{-1}f) f^{-1}$}
      edge [->, double] node [below right, l] {$\epsilon 1_{f^{-1}}$} (E)
      edge [<-, double] node [right, l] {$a^{-1}$} (A);
      \node at (90:.25cm) {$\Theta$};
    \node at (-90:.25cm) {$\Rrightarrow$};
    \begin{scope}
      \fill [white,fill=yellow,fill opacity=0.1] (A.center) to
(B.center) to (C.center) to (D.center) to (E.center) to (F.center) to
(A.center);
    \end{scope}
  \end{tikzpicture}
\end{center}
We require these 3-morphisms to satisfy a new equation, called the `swallowtail
identity'. The name comes from the `swallowtail castastrophe' in castastrophe theory,
and is motivated by a conjectured relationship between higher categories and tangles
known as the tangle hypothesis \cite{BaezDolan}. To fully appreciate this conjecture,
it is necessary to see the axioms in this section drawn as string diagrams, which can
be found in the paper of Stay \cite{Stay}. Here is the \define{swallowtail identity}:
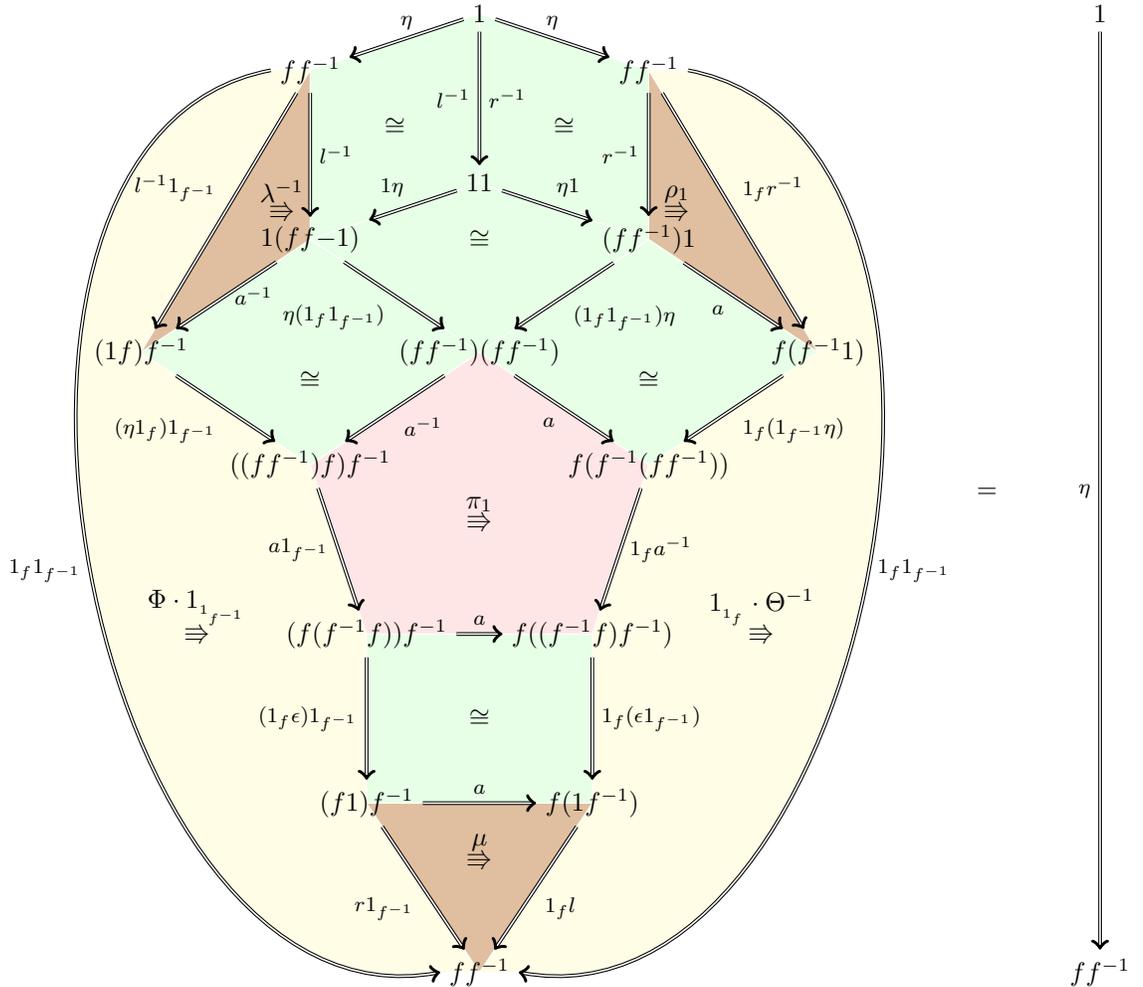
\begin{center}
  \begin{tikzpicture}[scale=.75]
    \filldraw [white,fill=brown,fill opacity=0.5] (6,0)--(4,3)--(8
,3)--cycle;
    \filldraw [white,fill=green,fill opacity=0.1] (4,6)--(8,6)--(8
,3)--(4,3)--cycle;
    \filldraw [white,fill=red,fill opacity=0.1] (4,6)--(8,6)--(9,9)--(6
,11)--(3,9)--cycle;
    \filldraw [white,fill=green,fill opacity=0.1] (3,9)--(0,11)--(3
,13)--(6,11)--cycle;
    \filldraw [white,fill=green,fill opacity=0.1] (6,11)--(9,9)--(12
,11)--(9,13)--cycle;
    \filldraw [white,fill=green,fill opacity=0.1] (6,11)--(3,13)--(6
,14)--(9,13)--cycle;
    \filldraw [white,fill=green,fill opacity=0.1] (3,13)--(3,16)--(6
,17)--(6,14)--cycle;
    \filldraw [white,fill=green,fill opacity=0.1] (6,17)--(6,14)--(9
,13)--(9,16)--cycle;
    \filldraw [white,fill=brown,fill opacity=0.5] (0,11)--(3,13)--(3
,16)--cycle;
    \filldraw [white,fill=brown,fill opacity=0.5] (12,11)--(9,13)--(9
,16)--cycle;
    
    \node (A) at (6,0) {$ff^{-1}$};
    \node (B) at (4,3) {$(f1)f^{-1}$}
      edge [->, double] node [below left, l] {$r 1_{f^{-1}}$} (A);
    \node (C) at (8,3) {$f(1f^{-1})$}
      edge [<-, double] node [above, l] {$a$} (B)
      edge [->, double] node [below right, l] {$1_f l$} (A);
    \node at (6,2) [above] {$\mu$};
    \node at (6,2) {$\Rrightarrow$};
    \node (D) at (4,6) {$(f(f^{-1}f))f^{-1}$}
    edge [->, double] node [left,l] {$(1_f \epsilon)1_{f^{-1}}$} (B);
    \node (E) at (8,6) {$f((f^{-1}f)f^{-1})$}
      edge [<-, double] node [above, l] {$a$} (D)
      edge [->, double] node [right, l] {$1_f (\epsilon 1_{f^{-1}})$} (C);
    \node at (6,4.5) {$\cong$};
    \node (F) at (3,9) {$((ff^{-1})f)f^{-1}$}
      edge [->, double] node [left,l] {$a 1_{f^{-1}}$} (D);
    \node (G) at (9,9) {$f(f^{-1}(ff^{-1}))$}
      edge [->, double] node [right, l] {$1_f a^{-1}$} (E);
    \node (H) at (0,11) {$(1f)f^{-1}$}
      edge [->, double] node [below left, l] {$(\eta 1_f)1_{f^{-1}}$} (F);
    \node (I) at (6,11) {$(ff^{-1})(ff^{-1})$}
      edge [->, double] node [below right, l] {$a^{-1}$} (F)
      edge [->, double] node [below left, l] {$a$} (G);
    \node at (6,8) [above] {$\pi_1$};
    \node at (6,8) {$\Rrightarrow$};
    \node (J) at (12,11) {$f(f^{-1}1)$}
      edge [->, double] node [below right, l] {$1_f(1_{f^{-1}} \eta)$} (G);
    \node (K) at (3,13) {$1(ff{-1})$}
      edge [->, double] (H)
      edge [->, double] node [below left, l] {$\eta(1_f 1_{f^{-1}})$} (I);
      \node[l] (Z) at (2,12) {$a^{-1}$}; 
    \node at (3,10.5) {$\cong$};
    \node (L) at (9,13) {$(ff^{-1})1$}
      edge [->, double] node [below right, l] {$(1_f 1_{f^{-1}})\eta$} (I)
      edge [->, double] node [below left, l] {$a$} (J);
    \node at (9,10.5) {$\cong$};
    \node (M) at (3,16) {$ff^{-1}$}
      edge [->, double] node [above left, l] {$l^{-1} 1_{f^{-1}}$} (H)
      edge [->, double] node [right, l] {$l^{-1}$} (K);
    \begin{scope}[on background layer]
      \fill[white,fill=yellow,fill opacity=0.1] (A.center) to (B.center)
to (D.center) to (F.center) to (H.center) to (M.center) to (M.west) to
[bend right=90] (A.west) to (A.center);
    \end{scope}
    \draw [->, double, black] (M.west) to [bend right=90] node [left, l] 
{$1_f 1_{f^{-1}}$} (A.west);
    \node at (2.5,13.5) [above] {$\lambda^{-1}$};
    \node at (2.5,13.5) {$\Rrightarrow$};
    \node (N) at (9,16) {$ff^{-1}$}
      edge [->, double] node [left, l] {$r^{-1}$} (L)
      edge [->, double] node [above right, l] {$1_f r^{-1}$} (J);
    \begin{scope}[on background layer]
      \fill[white,fill=yellow,fill opacity=0.1] (A.center) to (C.center)
to (E.center) to (G.center) to (J.center) to (N.center) to (N.east) to
[bend left=90] (A.east) to (A.center);
    \end{scope}
    \draw [->, double, black] (N.east) to [bend left=90] node [right, l] 
{$1_f 1_{f^{-1}}$} (A.east);
    \node at (9.5,13.5) [above] {$\rho_1$};
    \node at (9.5,13.5) {$\Rrightarrow$};
    \node (O) at (6,14) {$11$}
      edge [->, double] node [above left, l] {$1\eta$} (K)
      edge [->, double] node [above right, l] {$\eta1$} (L);
    \node at (6,13) {$\cong$};
    \node (P) at (6,17) {$1$}
      edge [->, double] node [above, l] {$\eta$} (M)
      edge [->, double] node [above, l] {$\eta$} (N)
      edge [->, double] node [left, l] {$l^{-1}$} node [right, l] 
{$r^{-1}$} (O);
    \node at (4.5,15) {$\cong$};
    \node at (7.5,15) {$\cong$};
    \node at (1,6) [above] {$\Phi \cdot 1_{_{1_{f^{-1}}}}$};
    \node at (1,6) {$\Rrightarrow$};
    \node at (11,6) [above] {$1_{_{1_f}} \cdot \Theta^{-1}$};
    \node at (11,6) {$\Rrightarrow$};
    \node at (15, 8.5) {$=$};

    \node (A') at (17,0) {$ff^{-1}$};
    \node (P') at (17,17) {$1$}
      edge [->, double] node [left, l] {$\eta$} (A');
  \end{tikzpicture}
\end{center}
In this identity, we have omitted $\cdot$ everywhere except the faces to save space,
and the unlabeled squares are naturality squares, which will be trivial in the
examples we consider. The 3-morphisms $\pi_1$ and $\rho_1$ are obtained from $\pi$
and $\rho$ by using the fact that the associator and unitors are adjoint equivalences
to reverse some arrows. For us, these adjoint equivalences will be trivial.

Ultimately, we are interested in having inverses in a smooth tricategory. For this,
we need smooth choices of all of these data. In particular, a smooth tricategory $T$ is said to have \define{invertible 3-morphisms} if there is a smooth map
\[ \inv_3 \maps T_3 \to T_3 \]
assigning to any 3-morphism $\Gamma$ its inverse $\Gamma^{-1} =\inv_3(\Gamma)$. We
say $T$ has \define{weakly invertible 2-morphisms} if there are smooth maps:
\[ \inv_2 \maps T_2 \to T_2, \quad e \maps T_2 \to T_3, \quad u \maps T_2 \to T_3 \]
assigning to any 2-morphism $\alpha$ a choice of weak inverse $\alpha^{-1} =
\inv_2(\alpha)$ and invertible 3-morphisms $e = e(\alpha)$ and $u = u(\alpha)$ such
that $(\alpha, \alpha^{-1}, e, u)$ is an adjoint equivalence. Finally, we say $T$ has
\define{weakly invertible morphisms} if there are smooth maps:
\[ \inv_1 \maps T_1 \to T_1, \quad \epsilon \maps T_1 \to T_2^2 \times T_3^2, \quad
\eta \maps T_1 \to T_2^2 \times T_3^2, \quad \Phi \maps T_1 \to T_3, \quad \Theta
\maps T_1 \to T_3 \] assigning to any morphism $f$ a choice of weak inverse $f^{-1} =
\inv_1(f)$ such that $(f, f^{-1}, \epsilon, \eta, \Phi, \Theta)$ is a biadjoint
biequivalence. Note that targets of $\epsilon$ and $\eta$ are what we need to specify
adjoint equivalences. Here, we have suppressed the dependence of the last
four entries on $f$ for brevity.

\subsection{Lie 3-groups} \label{sec:l3g}

A \define{Lie 3-group} is a smooth tricategory with one object where morphisms,
2-morphisms and 3-morphisms are weakly invertible. Though it looks quite complex, we can
construct an example using only a 4-cocycle in group cohomology, because the
pentagonator identity is secretly a cocycle condition. Given a normalized $H$-valued
4-cocycle $\pi$ on a Lie group $G$, we can construct a Lie 3-group $\Brane_\pi(G,H)$
with:
\begin{itemize}
	\item One object, $\bullet$, regarded as a manifold in the trivial way.
	\item For each element $g \in G$, an automorphism of the one object:
		\[ \bullet \stackrel{g}{\longrightarrow} \bullet \]
		Composition at a 0-cell given by multiplication in the group:
		\[ \cdot \maps G \times G \to G. \]
		The source and target maps are trivial, and identity-assigning
		map takes the one object to $1 \in G$.
	\item Only the identity 2-morphism on any 1-morphism, and no
		2-morphisms between distinct 1-morphisms:
		\[
		\xy
		(-8,0)*+{\bullet}="4";
		(8,0)*+{\bullet}="6";
		{\ar@/^1.65pc/^g "4";"6"};
		{\ar@/_1.65pc/_g "4";"6"};
		{\ar@{=>}^{1_g} (0,3)*{};(0,-3)*{}} ;
		\endxy
		, \quad g \in G.
		\]
		So the space of 2-morphisms is also $G$. The source, target and
		identity-assigning maps are all the identity on $G$.
		Composition at a 1-cell is trivial, while composition at a
		0-cell is again multiplication in $G$.
	\item For each $h \in H$, a 3-automorphism of the 2-morphism $1_g$,
		and no 3-morphisms between distinct 2-morphisms:
		
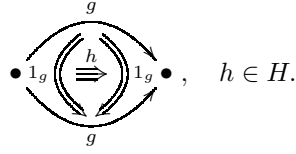
\[
		\xy 
		(-10,0)*+{\bullet}="1";
		(10,0)*+{\bullet}="2";
		{\ar@/^1.65pc/^g "1";"2"};
		{\ar@/_1.65pc/_g "1";"2"};
		(0,5)*+{}="A";
		(0,-5)*+{}="B";
		{\ar@{=>}@/_.75pc/ "A"+(-1.33,0) ; "B"+(-.66,-.55)};
		{\ar@{=}@/_.75pc/ "A"+(-1.33,0) ; "B"+(-1.33,0)};
		{\ar@{=>}@/^.75pc/ "A"+(1.33,0) ; "B"+(.66,-.55)};
		{\ar@{=}@/^.75pc/ "A"+(1.33,0) ; "B"+(1.33,0)};
		{\ar@3{->} (-2,0)*{}; (2,0)*{}};
		(0,2.5)*{\scriptstyle h};
		(-7,0)*{\scriptstyle 1_g};
		(7,0)*{\scriptstyle 1_g};
		\endxy
		, \quad h \in H.
		\]
		Thus the space of 3-morphisms is $G \times H$. The source and
		target maps are projection onto $G$, and the identity assigning
		map takes $1_g$ to $0 \in H$, for all $g \in G$.
	\item Three kinds of composition of 3-morphisms: given a pair of
		3-morphisms on the same 2-morphism, we can compose them at at a
		2-cell, which we take to be addition in $H$:
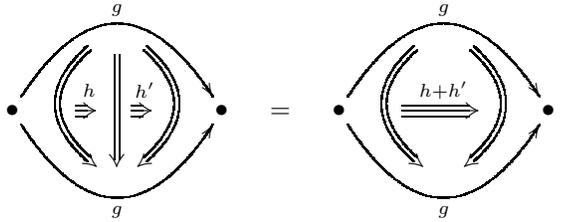
		\[
		\xy 0;/r.22pc/:
		(0,15)*{};
		(0,-15)*{};
		(0,8)*{}="A";
		(0,-8)*{}="B";
		{\ar@{=>} "A" ; "B"};
		{\ar@{=>}@/_1pc/ "A"+(-4,1) ; "B"+(-3,0)};
		{\ar@{=}@/_1pc/ "A"+(-4,1) ; "B"+(-4,1)};
		{\ar@{=>}@/^1pc/ "A"+(4,1) ; "B"+(3,0)};
		{\ar@{=}@/^1pc/ "A"+(4,1) ; "B"+(4,1)};
		{\ar@3{->} (-6,0)*{} ; (-2,0)*+{}};
		(-4,3)*{\scriptstyle h};
		{\ar@3{->} (2,0)*{} ; (6,0)*+{}};
		(4,3)*{\scriptstyle h'};
		(-15,0)*+{\bullet}="1";
		(15,0)*+{\bullet}="2";
		{\ar@/^2.75pc/^g "1";"2"};
		{\ar@/_2.75pc/_g "1";"2"};
		\endxy 
		\quad = \quad
		\xy 0;/r.22pc/:
		(0,15)*{};
		(0,-15)*{};
		(0,8)*{}="A";
		(0,-8)*{}="B";
		{\ar@{=>}@/_1pc/ "A"+(-4,1) ; "B"+(-3,0)};
		{\ar@{=}@/_1pc/ "A"+(-4,1) ; "B"+(-4,1)};
		{\ar@{=>}@/^1pc/ "A"+(4,1) ; "B"+(3,0)};
		{\ar@{=}@/^1pc/ "A"+(4,1) ; "B"+(4,1)};
		{\ar@3{->} (-6,0)*{} ; (6,0)*+{}};
		(0,3)*{\scriptstyle h + h'};
		(-15,0)*+{\bullet}="1";
		(15,0)*+{\bullet}="2";
		{\ar@/^2.75pc/^g "1";"2"};
		{\ar@/_2.75pc/_g "1";"2"};
		\endxy .
		\]
		We can also compose two 3-morphisms at a 1-cell, which we again
		take to be addition in $H$:
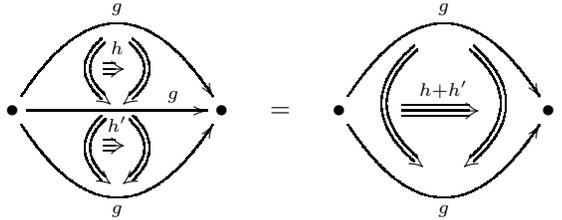
		\[	
		\xy 0;/r.22pc/:
		(0,15)*{};
		(0,-15)*{};
		(0,9)*{}="A";
		(0,1)*{}="B";
		{\ar@{=>}@/_.5pc/ "A"+(-2,1) ; "B"+(-1,0)};
		{\ar@{=}@/_.5pc/ "A"+(-2,1) ; "B"+(-2,1)};
		{\ar@{=>}@/^.5pc/ "A"+(2,1) ; "B"+(1,0)};
		{\ar@{=}@/^.5pc/ "A"+(2,1) ; "B"+(2,1)};
		{\ar@3{->} (-2,6)*{} ; (2,6)*+{}};
		(0,9)*{\scriptstyle h};
		(0,-1)*{}="A";
		(0,-9)*{}="B";
		{\ar@{=>}@/_.5pc/ "A"+(-2,-1) ; "B"+(-1,-1.5)};
		{\ar@{=}@/_.5pc/ "A"+(-2,0) ; "B"+(-2,-.7)};
		{\ar@{=>}@/^.5pc/ "A"+(2,-1) ; "B"+(1,-1.5)};
		{\ar@{=}@/^.5pc/ "A"+(2,0) ; "B"+(2,-.7)};
		{\ar@3{->} (-2,-5)*{} ; (2,-5)*+{}};
		(0,-2)*{\scriptstyle h'};
		(-15,0)*+{\bullet}="1";
		(15,0)*+{\bullet}="2";
		{\ar@/^2.75pc/^g "1";"2"};
		{\ar@/_2.75pc/_g "1";"2"};
		{\ar "1";"2"};
		(8,2)*{\scriptstyle g};
		\endxy 
		\quad = \quad
		\xy 0;/r.22pc/:
		(0,15)*{};
		(0,-15)*{};
		(0,8)*{}="A";
		(0,-8)*{}="B";
		{\ar@{=>}@/_1pc/ "A"+(-4,1) ; "B"+(-3,0)};
		{\ar@{=}@/_1pc/ "A"+(-4,1) ; "B"+(-4,1)};
		{\ar@{=>}@/^1pc/ "A"+(4,1) ; "B"+(3,0)};
		{\ar@{=}@/^1pc/ "A"+(4,1) ; "B"+(4,1)};
		{\ar@3{->} (-6,0)*{} ; (6,0)*+{}};
		(0,3)*{\scriptstyle h + h'};
		(-15,0)*+{\bullet}="1";
		(15,0)*+{\bullet}="2";
		{\ar@/^2.75pc/^g "1";"2"};
		{\ar@/_2.75pc/_g "1";"2"};
		\endxy .
		\]
		In terms of maps, both of these compositions are just:
		\[ 1 \times + \maps G \times H \times H \to G \times H. \]
		And finally, we can glue two 3-cells at the 0-cell, the object.
		We call this \define{composition at a 0-cell}, and define it to
		be addition \emph{twisted by the action of $G$}:
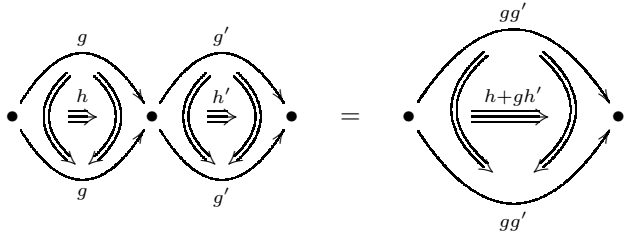
		\[
		\xy 0;/r.22pc/:
		(0,15)*{};
		(0,-15)*{};
		(-20,0)*+{\bullet}="1";
		(0,0)*+{\bullet}="2";
		{\ar@/^2pc/^g "1";"2"};
		{\ar@/_2pc/_g "1";"2"};
		(20,0)*+{\bullet}="3";
		{\ar@/^2pc/^{g'} "2";"3"};
		{\ar@/_2pc/_{g'} "2";"3"};
		(-10,6)*+{}="A";
		(-10,-6)*+{}="B";
		{\ar@{=>}@/_.7pc/ "A"+(-2,0) ; "B"+(-1,-.8)};
		{\ar@{=}@/_.7pc/ "A"+(-2,0) ; "B"+(-2,0)};
		{\ar@{=>}@/^.7pc/ "A"+(2,0) ; "B"+(1,-.8)};
		{\ar@{=}@/^.7pc/ "A"+(2,0) ; "B"+(2,0)};
		(10,6)*+{}="A";
		(10,-6)*+{}="B";
		{\ar@{=>}@/_.7pc/ "A"+(-2,0) ; "B"+(-1,-.8)};
		{\ar@{=}@/_.7pc/ "A"+(-2,0) ; "B"+(-2,0)};
		{\ar@{=>}@/^.7pc/ "A"+(2,0) ; "B"+(1,-.8)};
		{\ar@{=}@/^.7pc/ "A"+(2,0) ; "B"+(2,0)};
		{\ar@3{->} (-12,0)*{}; (-8,0)*{}};
		(-10,3)*{\scriptstyle h};
		{\ar@3{->} (8,0)*{}; (12,0)*{}};
		(10,3)*{\scriptstyle h'};
		\endxy 
		\quad = \quad 
		\xy 0;/r.22pc/:
		(0,15)*{};
		(0,-15)*{};
		(0,8)*{}="A";
		(0,-8)*{}="B";
		{\ar@{=>}@/_1pc/ "A"+(-4,1) ; "B"+(-3,0)};
		{\ar@{=}@/_1pc/ "A"+(-4,1) ; "B"+(-4,1)};
		{\ar@{=>}@/^1pc/ "A"+(4,1) ; "B"+(3,0)};
		{\ar@{=}@/^1pc/ "A"+(4,1) ; "B"+(4,1)};
		{\ar@3{->} (-6,0)*{} ; (6,0)*+{}};
		(0,3)*{\scriptstyle h + gh'};
		(-15,0)*+{\bullet}="1";
		(15,0)*+{\bullet}="2";
		{\ar@/^2.75pc/^{gg'} "1";"2"};
		{\ar@/_2.75pc/_{gg'} "1";"2"};
		\endxy .
		\]
		In terms of a map, $\cdot$ is just given by multiplication on
		the semidirect product:
		\[ \cdot \maps (G \ltimes H) \times (G \ltimes H) \to G \ltimes H. \]

        \item The associator and left and right unitors are trivial.

	\item For each quadruple of 1-morphisms, a specified 3-isomorphism,
		the \define{2-associator} or \define{pentagonator}:
		\[ \pi(g_1, g_2, g_3, g_4) \maps 1_{g_1 g_2 g_3 g_4} \to 1_{g_1 g_2 g_3 g_4}. \]
		given by the 4-cocycle $\pi \maps G^4 \to H$, which we think of
		as an element of $H$ because the source (and target) are
		understood to be $1_{g_1 g_2 g_3 g_4}$.

	\item The three other specified 3-isomorphisms, $\lambda$, $\rho$ and $\mu$, are trivial.

        \item The inverse of 3-morphisms is just given by negation in $H$:
		\[
                \inv_3 \left(
		\xy 
		(-10,0)*+{\bullet}="1";
		(10,0)*+{\bullet}="2";
		{\ar@/^1.65pc/^g "1";"2"};
		{\ar@/_1.65pc/_g "1";"2"};
		(0,5)*+{}="A";
		(0,-5)*+{}="B";
		{\ar@{=>}@/_.75pc/ "A"+(-1.33,0) ; "B"+(-.66,-.55)};
		{\ar@{=}@/_.75pc/ "A"+(-1.33,0) ; "B"+(-1.33,0)};
		{\ar@{=>}@/^.75pc/ "A"+(1.33,0) ; "B"+(.66,-.55)};
		{\ar@{=}@/^.75pc/ "A"+(1.33,0) ; "B"+(1.33,0)};
		{\ar@3{->} (-2,0)*{}; (2,0)*{}};
		(0,2.5)*{\scriptstyle h};
		\endxy
                \right)
                =
		\xy 
		(-10,0)*+{\bullet}="1";
		(10,0)*+{\bullet}="2";
		{\ar@/^1.65pc/^g "1";"2"};
		{\ar@/_1.65pc/_g "1";"2"};
		(0,5)*+{}="A";
		(0,-5)*+{}="B";
		{\ar@{=>}@/_.75pc/ "A"+(-1.33,0) ; "B"+(-.66,-.55)};
		{\ar@{=}@/_.75pc/ "A"+(-1.33,0) ; "B"+(-1.33,0)};
		{\ar@{=>}@/^.75pc/ "A"+(1.33,0) ; "B"+(.66,-.55)};
		{\ar@{=}@/^.75pc/ "A"+(1.33,0) ; "B"+(1.33,0)};
		{\ar@3{->} (-2,0)*{}; (2,0)*{}};
		(0,2.5)*{\scriptstyle -h};
		\endxy
		\]
          Or, as a map:
          \[ \inv_3 \maps G \times H \to G \times H \]
          sending $(g,h)$ to $(g,-h)$.
          
          The inverses for 2-morphisms are trivial, because 2-morphisms are
          trivial.
          
          Inverses for 1-morphisms are just inverses in $G$:
          \[ \inv_1 \left( \bullet \stackrel{g}{\longrightarrow} \bullet \right) =
          \bullet \stackrel{g^{-1}}{\longrightarrow} \bullet \]
          This is made into a biadjoint biequivalence with $\epsilon$ and $\eta$
          trivial, and $\Phi$ and $\Theta$ chosen to satisfy the swallowtail
          identity. There are many possible choices; here is a convenient one:
          \[ \Phi(g) = -\pi(g,g^{-1}, g, g^{-1}), \quad \Theta(g) = 0 . \]
          Again, we regard these 3-morphisms as elements of $H$, because their
          sources and targets are understood:
          \[ \Phi(g) \maps 1_g \to 1_g, \quad \Theta(g) \maps 1_{g^{-1}} \to 1_{g^{-1}} \]
            
\end{itemize}
A \define{slim Lie 3-group} is one of this form. It remains to check
that it is, in fact, a Lie 3-group. We claim:
\begin{prop} \label{prop:Lie3group}
	$\Brane_\pi(G,H)$ is a Lie 3-group: a smooth tricategory with one
	object and all morphisms, 2-morphisms and 3-morphisms weakly invertible.
\end{prop}

\begin{proof}
	As noted above, the triangulators $\lambda$, $\rho$ and
	$\mu$ are trivial. The axioms they satisfy are automatic because $\pi$
	is normalized.

	So to check that $\Brane_\pi(G,H)$ is a tricategory, it remains to
	check that $\pi$ satisfies the pentagonator identity. Since the 3-cells
	of the pentagonator identity commute (they represent elements of $H$),
	and all the faces not involving a $\pi$ are trivial, the first half reads:
	\[ \pi(g_1,g_2,g_3,g_4) \cdot 0_{g_5} + \pi(g_1 g_2, g_3, g_4, g_5) + \pi(g_1,g_2,g_3 g_4,g_5) . \]
	Here, we write $0_{g_5}$ to denote the 3-morphism which is the identity on
	the identity of the 1-morphism $g_5$. We do not need to be worried
	about order of terms, since composition of 3-morphisms is addition in
	$H$. The second half of the pentagonator identity reads:
	\[ 0_{g_1} \cdot \pi(g_2,g_3,g_4,g_5) + \pi(g_1,g_2 g_3, g_4, g_5) + \pi(g_1,g_2,g_3,g_4 g_5). \]
	Here, we write $0_{g_1}$ for to denote the 3-morphism which is the
	identity on the identity of the 1-morphism $g_1$. Applying the
	definition of $\cdot$ as the semidirect product, we see the equality of
	the first half with the second half is just the cocycle condition on
	$\pi$:
	\begin{eqnarray*} 
		\pi(g_1,g_2,g_3,g_4) + \pi(g_1 g_2, g_3, g_4, g_5) + \pi(g_1,g_2,g_3 g_4,g_5)   \\
		= g_1 \pi(g_2,g_3,g_4,g_5) + \pi(g_1,g_2 g_3, g_4, g_5) + \pi(g_1,g_2,g_3,g_4 g_5) .
	\end{eqnarray*} 

	So, $\Brane_\pi(G,H)$ is a tricategory. It is smooth because everything in
        sight is smooth: $G$, $H$, and the map $\pi \maps G^4 \to H$. And it is a Lie
        3-group: the 1-morphisms in $G$, the trivial 2-morphisms, and the 3-morphisms
        in $H$ are all strictly invertible, and thus weakly invertible. The swallowtail
        identity reduces to:
        \[ \Phi(g) \cdot 0_{g^{-1}} + \pi(g,g^{-1},g,g^{-1}) - 0_g \cdot \Theta(g) = 0 \]
        Upon substituting the values of $\Phi$ and $\Theta$ we chose above, this becomes:
        \[ -\pi(g,g^{-1},g,g^{-1}) + \pi(g,g^{-1},g,g^{-1}) + 0 = 0 \]
        which is indeed true.
\end{proof}

We can say something a bit stronger about $\Brane_\pi(G,H)$, if we let $\pi$ be any
normalized $H$-valued 4-cochain, rather than requiring it to be a cocycle. In this case,
$\Brane_{\pi}(G,H)$ is a Lie 3-group if and only if $\pi$ is a 4-cocycle, because
$\pi$ satisfies the pentagon identity if and only if it is a cocycle.

\section{Supergeometry} \label{sec:supergeometry}

We would now like to generalize our work from Lie algebras and Lie groups to Lie
superalgebras and supergroups. Of course, this means that we need a way to talk about
Lie supergroups, their underlying supermanifolds, and the maps between
supermanifolds. This task is made easier because we do not need the full machinery of
supermanifold theory. Our key examples of supergroups will be exponential, meaning
that the exponential map
\[ \exp \maps \g \to G \]
is a diffeomorphism, when $G$ is a supergroup and $\g$ is its Lie superalgebra. So,
we only need to work with supermanifolds that are diffeomorphic to super vector
spaces.

Nonetheless, we will give a lightning review of supermanifold theory from the
perspective that suits us best, which could loosely be called the `functor of points'
approach. A more leisurely and carefully motivated account is in our previous paper
\cite{Huerta:susy3}. That account and this one are based on the work of Sachse
\cite{Sachse} and Balduzzi, Carmeli and Fioresi \cite{BCF} developing the functor of
points for supermanifolds. Their work goes back to ideas of Schwarz \cite{Schwarz}
and Voronov \cite{Voronov}.

To begin, a \define{Grassmann algebra} is a finite-dimensional exterior algebra
\[ A = \Lambda \R^n, \]
equipped with the grading:
\[ A_0 = \Lambda^0 \R^n \oplus \Lambda^2 \R^n \oplus \cdots, \quad A_1 = \Lambda^1
\R^n \oplus \Lambda^3 \R^n \oplus \cdots . \]
Let us write $\GrAlg$ for the category with Grassmann algebras as objects and
grade-preserving homomorphisms as morphisms. We will define a supermanifold to be a
functor from the category of Grassmann algebras to smooth manifolds
\[ M \maps \GrAlg \to \Man \]
equipped with some extra structure. Let us write $M_A$ for the value of $M$ at the
Grassmann algebra $A$. We call this manifold the \define{$A$-points} of $M$. For $f \maps
A \to B$ a homomorphism of Grassmann algebras $A$ and $B$, we write the induced
smooth map as $M_f \maps M_A \to M_B$.

As we already remarked, our most important example of a supermanifold is a super
vector space. Indeed, given a finite-dimensional super vector space $V$, define the
\define{ supermanifold associated to $V$}, or just the \define{supermanifold $V$} to
be the functor:
\[ V \maps \GrAlg \to \Man \]
which takes:
\begin{itemize}
	\item each Grassmann algebra $A$ to the vector space:
		\[ V_A = (A \tensor V)_0 = A_0 \tensor V_0 \, \oplus \, A_1 \tensor V_1 \]
		regarded as a manifold in the usual way;
	\item each homomorphism $f \maps A \to B$ of Grassmann algebras to the
		linear map $V_f \maps V_A \to V_B$ that is the identity on
		$V$ and $f$ on $A$:
		\[ V_f = (f \tensor 1)_0 \maps (A \tensor V)_0 \to (B \tensor V)_0 . \]
		This map, being linear, is also smooth.
\end{itemize}

We also need to know how to smoothly map one super vector space to another. First
note that $V_A$ is more than a mere vector space; it is an $A_0$-module. Given two
finite-dimensional super vector spaces $V$ and $W$, we define a \define{smooth map
between super vector spaces}:
\[ \varphi \maps V \to W , \]
to be a natural transformation between the supermanifolds $V$ and $W$ such that
the derivative
\[ (\varphi_A)_* \maps T_x V_A \to T_{\varphi_A(x)} W_A \]
is $A_0$-linear at each $A$-point $x \in V_A$, where the $A_0$-module structure
on each tangent space comes from the canonical identification of a vector space
with its tangent space:
\[ T_x V_A \iso V_A, \quad T_{\varphi(x)} W_A \iso W_A . \]
Note that each component $\varphi_A \maps V_A \to W_A$ is smooth in the
ordinary sense, by virtue of living in the category of smooth
manifolds. We say that a smooth map $\varphi_A \maps V_A \to W_A$ whose
derivative is $A_0$-linear at each point is \define{$A_0$-smooth} for short.

These definitions, of the supermanifold associated to a super vector space and of
smooth maps between super vector spaces, are the most important for us.  Nonetheless,
we now sketch how to define a general supermanifold, $M$. Since $M$ will be locally
isomorphic to a super vector space $V$, it helps to have local pieces of $V$ to play
the same role that open subsets of $\R^n$ play for ordinary manifolds.  So, fix a
super vector space $V$, and let $U \subseteq V_0$ be open. The \define{superdomain}
over $U$ is the functor:
\[ \mathcal{U} \maps \GrAlg \to \Man \]
that takes each Grassmann algebra $A$ to the following open subset of $V_A$:
\[ \mathcal{U}_A = V_{\epsilon_A}^{-1}(U) \]
where $\epsilon_A \maps A \to \R$ is the projection of the Grassmann algebra
$A$ that kills all nilpotent elements. We say that $\mathcal{U}$ is a
\define{superdomain in $V$}, and write $\mathcal{U} \subseteq V$.

If $\mathcal{U} \subseteq V$ and $\mathcal{U}' \subseteq W$ are two
superdomains in super vector spaces $V$ and $W$, a \define{smooth map of
superdomains} is a natural transformation:
\[ \varphi \maps \mathcal{U} \to \mathcal{U}' \]
such that for each Grassmann algebra $A$, the component on
$A$-points is smooth:
\[ \varphi_{A} \maps \mathcal{U}_A \to {\mathcal{U}'}_A . \]
and the derivative:
\[ (\varphi_A)_* \maps T_x \mathcal{U}_A \to T_{\varphi_A(x)} {\mathcal{U}'}_A \]
is $A_0$-linear at each $A$-point $x \in \mathcal{U}_A$, where the  
$A_0$-module structure on each tangent space comes from the canonical
identification with the ambient vector spaces:
\[ T_x \mathcal{U}_A \iso V_A, \quad T_{\varphi(x)} \mathcal{U}'_A \iso W_A . \]
Again, we say that a smooth map $\varphi_{A} \maps \mathcal{U}_A \to
{\mathcal{U}'}_A$ whose derivative is $A_0$-linear at each point is
\define{$A_0$-smooth} for short.

At long last, a \define{supermanifold} is a functor
\[ M \maps \GrAlg \to \Man \]
equipped with an atlas 
\[ (\mathcal{U}_\alpha, \varphi_\alpha \maps \mathcal{U} \to M) , \]
where each $\mathcal{U}_\alpha$ is a superdomain, each $\varphi_\alpha$ is a
natural transformation, and one can define transition functions that are smooth
maps of superdomains. 

Finally, a \define{smooth map of supermanifolds} is a natural
transformation:
\[ \psi \maps M \to N \]
which induces smooth maps between the superdomains in the atlases.
Equivalently, each component 
\[ \varphi_A \maps M_A \to N_A \] 
is \define{$A_0$-smooth}: it is smooth and its derivative
\[ (\varphi_A)_* \maps T_x M_A \to T_{\varphi_A(x)} N_A \]
is $A_0$-linear at each $A$-point $x \in M_A$, where the $A_0$-module structure
on each tangent space comes from the superdomains in the atlases. Thus, there
is a category $\SuperMan$ of supermanifolds. See Sachse \cite{Sachse} for more
details. 

Finally, note that there is a supermanifold:
\[ 1 \maps \GrAlg \to \Man , \]
which takes each Grassmann algebra to the one-point manifold. We call this the
\define{one-point supermanifold}, and note that it is the supermanifold
associated to the super vector space $\R^{0|0}$. The one-point supermanifold is
the terminal object in the category of supermanifolds.

\subsection{Supergroups from nilpotent Lie superalgebras} \label{sec:supergroups}

We now describe a procedure to integrate a nilpotent Lie superalgebra to a Lie
supergroup. This is a partial generalization of Lie's Third Theorem, which
describes how any Lie algebra can be integrated to a Lie group. In fact, the
full theorem generalizes to Lie supergroups \cite{Tuynman}, but we do not need
it here.

Recall from Section \ref{sec:Lie-n-superalgebras} that a \define{Lie superalgebra} $\g$
is a Lie algebra in the category of super vector spaces. More concretely, it is
a super vector space $\g = \g_0 \oplus \g_1$, equipped with a
graded-antisymmetric bracket:
\[ [-,-] \maps \Lambda^2 \g \to \g , \]
which satisfies the Jacobi identity up to signs:
\[ [X, [Y,Z]] = [ [X,Y], Z] + (-1)^{|X||Y|} [Y, [X, Z]]. \]
for all homogeneous $X, Y, Z \in \g$.  A Lie superalgebra $\n$ is called
\define{k-step nilpotent} if any $k$ nested brackets vanish, and it is called
\define{nilpotent} if it is $k$-step nilpotent for some $k$. Nilpotent Lie
superalgebras can be integrated to a unique supergroup $N$ defined on the same
underlying super vector space $\n$.

A \define{Lie supergroup}, or \define{supergroup}, is a group object in the
category of supermanifolds. That is, it is a supermanifold $G$ equipped with
the following maps of supermanifolds:
\begin{itemize}
	\item \define{multiplication}, $m \maps G \times G \to G$;
	\item \define{inverse}, $\inv \maps G \to G$;
	\item \define{identity}, $\id \maps 1 \to G$, where $1$ is the one-point
		supermanifold;
\end{itemize}
such that the following diagrams commute, encoding the usual group axioms:
\begin{itemize}
	\item the associative law:
	\[ \vcenter{
	\xymatrix{ &   G \times G \times G \ar[dr]^{1 \times m}
	   \ar[dl]_{m \times 1} \\
	 G \times G \ar[dr]_{m}
	&&  G \times G \ar[dl]^{m}  \\
	&  G }}
	\]
	\item the right and left unit laws:
	\[ \vcenter{
	\xymatrix{
	 I \times G \ar[r]^{\id \times 1} \ar[dr]
	& G \times G \ar[d]_{m}
	& G \times I \ar[l]_{1 \times \id} \ar[dl] \\
	& G }}
	\]
	\item the right and left inverse laws:
	\[
	\xy (-12,10)*+{G \times G}="TL"; (12,10)*+{G \times G}="TR";
	(-18,0)*+{G}="ML"; (18,0)*+{G}="MR"; (0,-10)*+{1}="B";
	     {\ar_{} "ML";"B"};
	     {\ar^{\Delta} "ML";"TL"};
	     {\ar_{\id} "B";"MR"};
	     {\ar^{m} "TR";"MR"};
	     {\ar^{1 \times \inv } "TL";"TR"};
	\endxy
	\qquad \qquad \xy (-12,10)*+{G \times G}="TL"; (12,10)*+{G \times
	G}="TR"; (-18,0)*+{G}="ML"; (18,0)*+{G}="MR"; (0,-10)*+{1}="B";
	     {\ar_{} "ML";"B"};
	     {\ar^{\Delta} "ML";"TL"};
	     {\ar_{\id} "B";"MR"};
	     {\ar^{m} "TR";"MR"};
	     {\ar^{\inv \times 1} "TL";"TR"};
	\endxy
	\]
\end{itemize}
where $\Delta \maps G \to G \times G$ is the diagonal map. In addition, a
supergroup is \define{abelian} if the following diagram commutes:
\[ \xymatrix{ 
G \times G \ar[r]^\tau \ar[dr]_m & G \times G \ar[d]^m \\
                                 & G
}
\]
where $\tau \maps G \times G \to G \times G$ is the \define{twist map}. Using
$A$-points, it is defined to be:
\[ \tau_A(x,y) = (y,x), \]
for $(x,y) \in G_A \times G_A$.

Examples of supergroups arise easily from Lie groups. We can regard any
ordinary manifold as a supermanifold, and so any Lie group $G$ is also a
supergroup. In this way, any classical Lie group, such as $\SO(n)$, $\SU(n)$
and $\Sp(n)$, becomes a supergroup.

To obtain more interesting examples, we will integrate a nilpotent Lie
superalgebra, $\n$ to a supergroup $N$. For any Grassmann algebra $A$, the bracket 
\[ [-,-] \maps \Lambda^2 \n \to \n \]
induces an $A_0$-linear map between the $A$-points:
\[ [-,-]_A \maps \Lambda^2 \n_A \to \n_A, \]
where $\Lambda^2 \n_A$ denotes the exterior square of the $A_0$-module $\n_A$.
Thus $[-,-]_A$ is antisymmetric, and it easy to check that it makes $\n_A$ into
a Lie algebra which is also nilpotent.

On each such $A_0$-module $\n_A$, we can thus define a Lie group $N_A$ where
the multplication is given by the Baker--Campbell--Hausdorff formula, inversion
by negation, and the identity is $0$. Because we want to write the group $N_A$
multiplicatively, we write $\exp_A \maps n_A \to N_A$ for the identity map, and
then define the multiplication, inverse and identity:
\[ m_A \maps N_A \times N_A \to N_A, \quad \inv_A \maps N_A \to N_A, \quad \id_A \maps 1_A \to N_A, \]
as follows:
\[ m_A(\exp_A(X), \exp_A(Y)) = \exp_A(X) \exp_A(Y) = \exp_A(X + Y + \half[X,Y]_A + \cdots ) \]
\[ \inv_A(\exp_A(X)) = \exp_A(X)^{-1} = \exp_A(-X) , \]
\[ \id_A(1) = 1 = \exp_A(0),  \]
for any $A$-points $X, Y \in \n_A$, where the first 1 in the last equation
refers to the single element of $1_A$. But it is clear that all of these maps
are natural in $A$. Furthermore, they are all $A_0$-smooth, because as
polynomials with coefficients in $A_0$, they are smooth with derivatives that
are $A_0$-linear.  They thus define smooth maps of supermanifolds:
\[ m \maps N \times N \to N, \quad \inv \maps N \to N, \quad \id \maps 1 \to N, \]
where $N$ is the supermanifold $\n$. And because each of the $N_A$ is a group,
$N$ is a supergroup. We have thus proved:

\begin{prop} \label{prop:nilpotentsupergroup}
	Let $\n$ be a nilpotent Lie superalgebra. Then there is a supergroup
	$N$ defined on the supermanifold $\n$, obtained by integrating the
	nilpotent Lie algebra $\n_A$ with the Baker--Campbell--Hausdorff
	formula for all Grassmann algebras $A$. More precisely, we define the maps:
	\[ m \maps N \times N \to N, \quad \inv \maps N \to N, \quad \id \maps 1 \to N, \]
	by defining them on $A$-points as follows:
	\[ m_A(\exp_A(X), \exp_A(Y)) = \exp_A(Z(X,Y)), \]
	\[ \inv_A(\exp_A(X)) = \exp_A(-X), \]
	\[ \id_A(1) = \exp_A(0), \]
	where 
	\[ \exp \maps \n \to N \]
	is the identity map of supermanifolds, and:
	\[ Z(X,Y) = X + Y + \half[X,Y]_A + \cdots \]
	denotes the Baker--Campbell--Hausdorff series on $\n_A$, which
	terminates because $\n_A$ is nilpotent.
\end{prop}
\noindent
Experience with ordinary Lie theory suggests that, in general, there will be
more than one supergroup which has Lie superalgebra $\n$. To distinguish the
one above, we call $N$ the \define{exponential supergroup} of $\n$. 

\section{3-supergroups from supergroup cohomology} \label{sec:3supergroups}

We saw in Section \ref{sec:lie3groups} that Lie 3-groups can be defined from a
4-cocycle in Lie group cohomology.  We now generalize this to supergroups. The most
significant barrier is that we now work internally to the category of supermanifolds
instead of the much more familiar category of smooth manifolds.  Our task is to show
that this change of categories does not present a problem.  The main obstacle is that
the category of supermanifolds is not a concrete category: morphisms are determined
not by their value on the underlying set of a supermanifold, but by their value on
$A$-points for all Grassmann algebras $A$.

The most common approach is to define morphisms without reference to elements, and to
define equations between morphisms using commutative diagrams. As an alternative to
commutative diagrams, for supermanifolds, one can use $A$-points to define morphisms
and specify equations between them. This tends to make equations look friendlier,
because they look like equations between functions.  We shall use this approach.

First, let us define the cohomology of a supergroup $G$ with coefficients in an
abelian supergroup $H$, on which $G$ \define{acts by automorphism}. This means
that we have a morphism of supermanifolds:
\[ \alpha \maps G \times H \to H, \]
which, for any Grassmann algebra $A$, induces an action of the group
$G_A$ on the abelian group $H_A$:
\[ \alpha_A \maps G_A \times H_A \to H_A. \]
For this action to be by automorphism, we require:
\[ \alpha_A(g)(h + h') = \alpha_A(g)(h) + \alpha_A(g)(h'), \]
for all $A$-points $g \in G_A$ and $h, h' \in H_A$.

We define supergroup cohomology using \define{the supergroup cochain complex},
$C^\bullet(G,H)$, which at level $p$ just consists of the set of maps
from $G^p$ to $H$ as supermanifolds:
\[ C^p(G,H) = \left\{ f \maps G^p \to H \right\} . \]
Addition on $H$ makes $C^p(G,H)$ into an abelian group for all $p$.  The
differential is given by the usual formula, but using $A$-points:
\begin{eqnarray*} 
	df_A(g_1, \dots, g_{p+1}) & = & g_1 f_A(g_2, \dots, g_{p+1}) \\
	                          &   & + \sum_{i=1}^p (-1)^i f_A(g_1, \dots, g_i g_{i+1}, \dots, g_{p+1}) \\
				  &   & + (-1)^{p+1} f_A(g_1, \dots, g_p) , \\
\end{eqnarray*}
where $g_1, \dots, g_{p+1} \in G_A$ and the action of $g_1$ is given by
$\alpha_A$. Noting that $f_A$, $\alpha_A$, multiplication and $+$ are all:
\begin{itemize}
	\item natural in $A$;
	\item $A_0$-smooth: smooth with derivatives which are $A_0$-linear;
\end{itemize}
we see that $df_A$ is:
\begin{itemize}
	\item natural in $A$; 
	\item $A_0$-smooth: smooth with a derivative which is $A_0$-linear;
\end{itemize}
so it indeed defines a map of supermanifolds:
\[ df \maps G^{p+1} \to H. \]
Furthermore, it is immediate that:
\[ d^2 f_A = 0 \]
for all $A$, and thus
\[ d^2 f = 0. \]
So $C^\bullet(G,H)$ is truly a cochain complex. Its cohomology $H^\bullet(G,H)$
is the \define{supergroup cohomology of $G$ with coefficients in $H$}. Of
course, if $df = 0$, $f$ is called a \define{cocycle}, and $f$ is \define{normalized} if
\[ f_A(g_1, \dots, g_p) = 0 \]
for any Grassmann algebra $A$, whenever one of the $A$-points $g_1, \dots, g_p$
is 1. When $H = \R$, we omit reference to it, and write $C^\bullet(G,\R)$ as
$C^\bullet(G)$.

We can generalize our construction of Lie 3-groups to `3-supergroups'.  A
\define{super tricategory} $T$ has
\begin{itemize}
	\item a \define{supermanifold of objects} $T_0$;
	\item a \define{supermanifold of morphisms} $T_1$;
	\item a \define{supermanifold of 2-morphisms} $T_2$;
	\item a \define{supermanifold of 3-morphisms} $T_3$;
\end{itemize}
equipped with maps of supermanifolds as described in Definition
\ref{def:smoothtricat}: source, target, identity-assigning, composition at 0-cells,
1-cells and 2-cells, associator and left and right unitors, pentagonator and
triangulators all maps of supermanifolds, and satisfying the same axioms as a smooth
tricategory. We express the pentagonator identity in terms of $A$-points: the
following equation holds:

\newpage
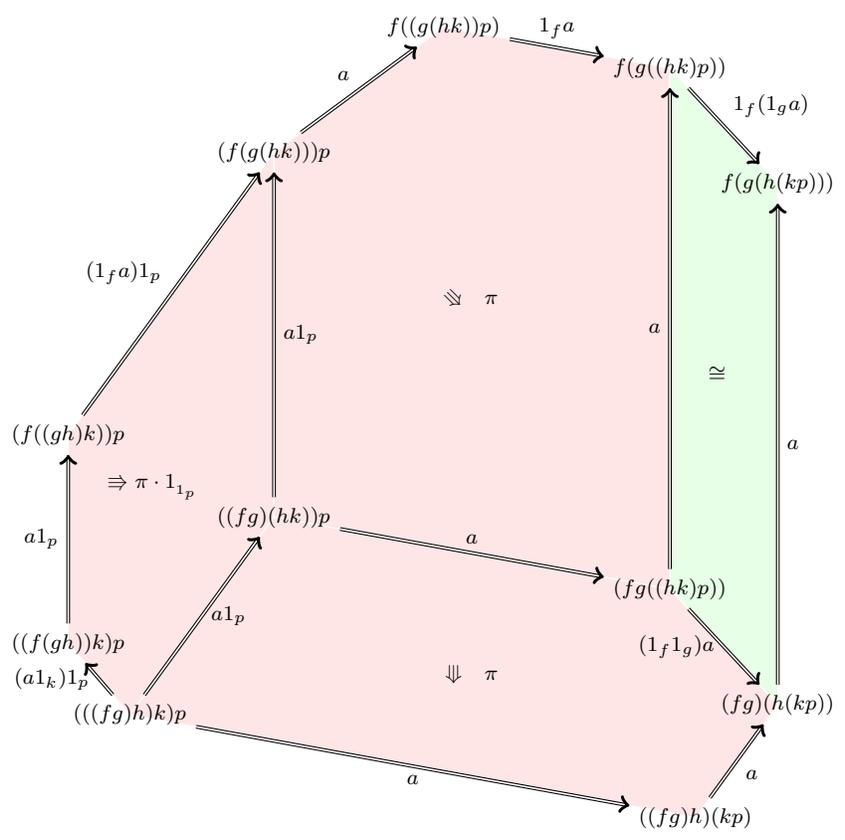
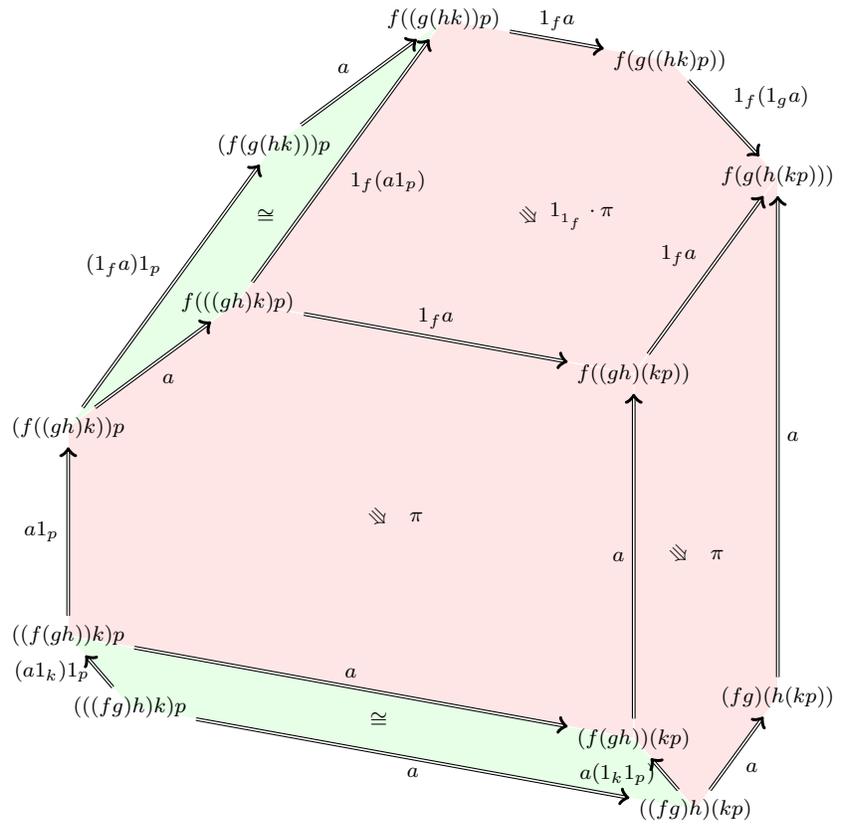
\thispagestyle{empty}

\begin{figure}[H]
    \begin{center}
      \begin{tikzpicture}[line join=round]
        \filldraw[white,fill=red,fill opacity=0.1](-4.306,-3.532)--(-2.391,-.901)--(-2.391,3.949)--(-5.127,.19)--(-5.127,-2.581)--cycle;
        \filldraw[white,fill=red,fill opacity=0.1](-4.306,-3.532)--(-2.391,-.901)--(2.872,-1.858)--(4.306,-3.396)--(3.212,-4.9)--cycle;
        \filldraw[white,fill=red,fill opacity=0.1](2.872,-1.858)--(2.872,5.07)--(-.135,5.617)--(-2.391,3.949)--(-2.391,-.901)--cycle;
        \filldraw[white,fill=green,fill opacity=0.1](4.306,-3.396)--(4.306,3.532)--(2.872,5.07)--(2.872,-1.858)--cycle;
        \begin{scope}[font=\fontsize{8}{8}\selectfont]
          \node (A) at (-2.391,3.949) {$(f(g(hk)))p$};
          \node (B) at (-5.127,.19) {$(f((gh)k))p$}
	  edge [->, double] node [l, above left] {$(1_{f}  a)1_{p}$} (A);
          \node (C) at (-5.127,-2.581) {$((f(gh))k)p$}
	  edge [->, double] node [l, left] {$a  1_{p}$} (B);
          \node (D) at (-4.306,-3.532) {$(((fg)h)k)p$}
	  edge [->, double] node [l, left] {$(a  1_{k})  1_{p}$} (C);
          \node (E) at (3.212,-4.9) {$((fg)h)(kp)$}
            edge [<-, double] node [l, below] {$a$} (D);
          \node (F) at (4.306,-3.396) {$(fg)(h(kp))$}
            edge [<-, double] node [l, below right] {$a$} (E);
          \node (G) at (4.306,3.532) {$f(g(h(kp)))$}
            edge [<-, double] node [l, right] {$a$} (F);
          \node (H) at (2.872,5.07) {$f(g((hk)p))$}
            edge [->, double] node [l, above right] {$1_{f}(1_{g}a)$} (G);
          \node (I) at (-.135,5.617) {$f((g(hk))p)$}
            edge [->, double] node [l, above] {$1_{f}a$} (H)
            edge [<-, double] node [l, above left] {$a$} (A);
          \node (M) at (-2.391,-.901) {$((fg)(hk))p$}
            edge [<-, double] node [l, right] {$a1_{p}$} (D)
            edge [->, double] node [l, right] {$a1_{p}$} (A);
          \node (N) at (2.872,-1.858) {$(fg((hk)p))$}
            edge [<-, double] node [l, above] {$a$} (M)
            edge [->, double] node [l, left] {$a$} (H)
            edge [->, double] node [l, left] {$(1_{f}1_{g})a$} (F);
	    \node at (-4,-.5) {$\Rrightarrow \pi \cdot 1_{_{1_{p}}}$};
          \node at (0,-3) {\tikz\node [rotate=-90] {$\Rrightarrow$};};
          \node at (0.5,-3) {$\pi$};
          \node at (0,2) {\tikz\node [rotate=-45] {$\Rrightarrow$};};
          \node at (0.5,2) {$\pi$};
          \node at (3.5,1) {$\cong$};
        \end{scope}
      \end{tikzpicture}
      \[ = \]
      \begin{tikzpicture}[line join=round]
        \filldraw[white,fill=green,fill opacity=0.1](-2.872,1.858)--(-.135,5.617)--(-2.391,3.949)--(-5.127,.19)--cycle;
        \filldraw[white,fill=red,fill opacity=0.1](3.212,-4.9)--(4.306,-3.396)--(4.306,3.532)--(2.391,.901)--(2.391,-3.949)--cycle;
        \filldraw[white,fill=green,fill opacity=0.1](-4.306,-3.532)--(3.212,-4.9)--(2.391,-3.949)--(-5.127,-2.581)--cycle;
        \filldraw[white,fill=red,fill opacity=0.1](-2.872,1.858)--(2.391,.901)--(4.306,3.532)--(2.872,5.07)--(-.135,5.617)--cycle;
        \filldraw[white,fill=red,fill opacity=0.1](-5.127,-2.581)--(-5.127,.19)--(-2.872,1.858)--(2.391,.901)--(2.391,-3.949)--cycle;
        \begin{scope}[font=\fontsize{8}{8}\selectfont]
          \node (A) at (-2.391,3.949) {$(f(g(h k)))p$};
          \node (B) at (-5.127,.19) {$(f((gh)k))p$}
            edge [->, double] node [l, above left] {$(1_{f}a)1_{p}$} (A);
          \node (C) at (-5.127,-2.581) {$((f(gh))k)p$}
            edge [->, double] node [l, left] {$a1_{p}$} (B);
          \node (D) at (-4.306,-3.532) {$(((fg)h)k)p$}
            edge [->, double] node [l, left] {$(a1_{k})1_{p}$} (C);
          \node (E) at (3.212,-4.9) {$((fg)h)(kp)$}
            edge [<-, double] node [l, below] {$a$} (D);
          \node (F) at (4.306,-3.396) {$(fg)(h(kp))$}
            edge [<-, double] node [l, below right] {$a$} (E);
          \node (G) at (4.306,3.532) {$f(g(h(kp)))$}
            edge [<-, double] node [l, right] {$a$} (F);
          \node (H) at (2.872,5.07) {$f(g((hk)p))$}
            edge [->, double] node [l, above right] {$1_{f}(1_{g}a)$} (G);
          \node (I) at (-.135,5.617) {$f((g(hk))p)$}
            edge [->, double] node [l, above] {$1_{f}a$} (H)
            edge [<-, double] node [l, above left] {$a$} (A);
          \node (J) at (-2.872,1.858) {$f(((gh)k)p)$}
            edge [->, double] node [l, below right] {$1_{f}(a1_{p})$} (I)
            edge [<-, double] node [l, below right] {$a$} (B);
          \node (K) at (2.391,-3.949) {$(f(gh))(kp)$}
            edge [<-, double] node [l, left] {$a(1_{k}1_{p})$} (E)
            edge [<-, double] node [l, above] {$a$} (C);
          \node (L) at (2.391,.901) {$f((gh)(kp))$}
            edge [<-, double] node [l, left] {$a$} (K)
            edge [<-, double] node [l, above] {$1_{f}a$} (J)
            edge [->, double] node [l, above left] {$1_{f}a$} (G);
          \node at (-1,-1) {\tikz\node [rotate=-45] {$\Rrightarrow$};};
	  \node at (-.5,-1) {$\pi$};
          \node at (1,3) {\tikz\node [rotate=-45] {$\Rrightarrow$};};
	  \node at (1.7,3) {$1_{_{1_{f}}} \cdot \pi$};
          \node at (3,-1.5) {\tikz\node [rotate=-45] {$\Rrightarrow$};};
          \node at (3.5,-1.5) {$\pi$};
          \node at (-1,-3.7) {$\cong$};
          \node at (-2.5,3) {$\cong$};
        \end{scope}
      \end{tikzpicture}
    \end{center}
\end{figure}
\clearpage

for any `composable quintet of morphisms':
\[ (f,g,h,k,p) \in (T_1 \times_{T_0} T_1 \times_{T_0} T_1 \times_{T_0} T_1 \times_{T_0} T_1)_A \]

A \define{3-supergroup} is a super tricategory with one object (more precisely,
the one-point supermanifold) and all morphisms, 2-morphisms and 3-morphisms
weakly invertible. Given a normalized $H$-valued 4-cocycle $\pi$ on $G$, we can
construct a 3-supergroup $\Brane_\pi(G,H)$ in the same way we constructed the
Lie 3-group $\Brane_\pi(G,H)$ when $G$ and $H$ were Lie groups, but deleting
every reference to elements of $G$ or $H$:
\begin{itemize}
	\item The supermanifold of objects is the one-point supermanifold, $1$.

	\item The supermanifold of morphisms is the supergroup $G$. 
		Composition at a 0-cell is given by multiplication in the group:
		\[ \cdot \maps G \times G \to G. \]
		The source and target maps are the unique maps to $1$. The
		identity-assigning map is the identity-assigning map for $G$:
		\[ \id \maps 1 \to G. \]

	\item The supermanifold of 2-morphisms is again $G$. The source, target
		and identity-assigning maps are all the identity on $G$.
		Composition at a 1-cell is the identity on $G$, while
		composition at a 0-cell is again multiplication in $G$.
		This encodes the idea that all 2-morphisms are trivial.

	\item The supermanifold of 3-morphisms is $G \times H$. The source and
		target maps are projection onto $G$. The identity-assigning map
		is the inclusion:
		\[ G \to G \times H \]
		that takes $A$-points $g \in G_A$ to $(g,0) \in G_A \times
		H_A$, for all $A$.

	\item There are three kinds of composition of 3-morphisms: composition at a
		2-cell and at a 3-cell are both given by addition on $H$:
		\[ 1 \times + \maps G \times H \times H \to G \times H. \]
		While composition at a 0-cell is just given by multiplication on
		the semidirect product:
		\[ \cdot \maps (G \ltimes H) \times (G \ltimes H) \to G \ltimes H. \]

        \item The associator and left and right unitors are trivial.

	\item The triangulators are trivial.

	\item The \define{2-associator} or \define{pentagonator} is given by the
		4-cocycle $\pi \maps G^4 \to H$, where the source (and target)
		is understood to come from multiplication on $G$.
        \item The inverse of 3-morphisms is just given by negation in $H$. The map
          \[ \inv_3 \maps G \times H \to G \times H \]
          sends the $A$-point $(g,h)$ to $(g,-h)$.
          
          The inverses for 2-morphisms are trivial, because 2-morphisms are
          trivial.
          
          Inverses for 1-morphisms are just inverses in $G$. The map
          \[ \inv_1 \maps G \to G \]
          is just the usual inverse map on $G$. This is made into a biadjoint
          biequivalence with $\epsilon$ and $\eta$ trivial, and $\Phi$ and $\Theta$
          chosen to satisfy the swallowtail identity. There are many possible
          choices; here is a convenient one in terms of $A$-points:
          \[ \Phi(g) = -\pi(g,g^{-1}, g, g^{-1}), \quad \Theta(g) = 0 . \]
          \end{itemize}
A \define{slim 3-supergroup} is one of this form. It remains to check that it
is, indeed, a 3-supergroup.

\begin{prop} \label{prop:3supergroup} $\Brane_\pi(G,H)$ is a 3-supergroup: a super
  tricategory with one object and all morphisms, 2-morphisms and 3-morphisms
  weakly invertible.
\end{prop}
\begin{proof}
	This proof is a duplicate of Proposition
	\ref{prop:Lie3group}, but with $A$-points instead of elements.
\end{proof}

\section{Integrating nilpotent Lie \emph{n}-superalgebras} \label{sec:integrating}

Any mathematician worth her salt knows that we can easily construct Lie algebras as
the infinitesimal versions of Lie groups, and that a more challenging inverse
construction exists: we can `integrate' Lie algebras to get Lie groups. In fact, the
same is true of supergroups and Lie superalgebras, and indeed for $n$-supergroups and
Lie $n$-superalgebras for all $n$!

In the following, we recall our solution to this integration problem for slim,
nilpotent Lie $n$-superalgebras, which appeared in our previous paper
\cite{Huerta:susy3}. As we saw in Section \ref{sec:Lie-n-superalgebras}, slim Lie
$n$-algebras are built from $(n+1)$-cocycles in Lie algebra cohomology.
Remember, $p$-cochains on the Lie algebra $\g$ are linear maps:
\[ C^p(\g,\h) = \left\{ \omega \maps \Lambda^p \g \to \h \right\} , \]
where $\h$ is a representation of $\g$, though we shall restrict ourselves to
the trivial representation $\h = \R$ in this section.

On the other hand, in Section \ref{sec:lie3groups}, we saw that slim Lie 3-groups are
built from 4-cocycles in Lie group cohomology. Remember, $p$-cochains on $G$ are
smooth maps:
\[ C^p(G,H) = \left\{ f \maps G^p \to H \right\} , \]
where $H$ is an abelian group on which $G$ acts by automorphism, though we
shall restrict ourselves to $H = \R$ with trivial action in this section.

This parallel suggests a naive scheme to integrate Lie 3-algebras. Given
a slim Lie 3-superalgebra $\brane_\omega(\g,\h)$, we seek a slim 
3-supergroup $\Brane_\pi(G,H)$ where:
\begin{itemize}
	\item $G$ is a Lie supergroup with Lie superalgebra $\g$; i.e.\, it is a Lie
		supergroup integrating $\g$,
	\item $H$ is a Lie supergroup with Lie superalgebra $\h$; i.e.\, it is a Lie
		supergroup integrating $\h$,
	\item $\pi$ is a Lie supergroup $4$-cocycle on $G$ that, in some suitable sense,
		integrates the Lie superalgebra $4$-cocycle $\omega$ on $\g$.
\end{itemize}
In this section, we describe an elegant, geometric procedure to integrate Lie
superalgebra cocycles to obtain supergroup cocycles, which works when the Lie
superalgebra in question is nilpotent.

Of course, this falls far short of integrating a general Lie $n$-superalgebra to an
$n$-supergroup, which has been done by others. Building on the earlier work of
Getzler \cite{Getzler} on integrating nilpotent Lie $n$-algebras, Henriques
\cite{Henriques} has shown that any Lie $n$-algebra can be integrated to a `Lie
$n$-group', which Henriques defines as a sort of smooth Kan complex in the category
of Banach manifolds.  More recently, Schreiber \cite{Schreiber} has generalized this
integration procedure to a setting much more general than that of Banach manifolds,
including both supermanifolds and manifolds with infinitesimals. For both Henriques
and Schreiber, the definition of Lie $n$-group is weaker than the one we sketched in
Section \ref{sec:lie3groups}---it weakens the notion of multiplication so that the
product of two group `elements' is only defined up to equivalence. This level of
generality seems essential for the construction to work for \emph{every} Lie
$n$-algebra.

However, for \emph{some} Lie $n$-algebras, we can integrate them using the more
naive idea of Lie $n$-group we prefer in this paper: a smooth $n$-category
with one object in which every $k$-morphism is weakly invertible, for all $1
\leq k \leq n$. We shall see that, for some slim Lie $n$-algebras, we can
integrate the defining Lie algebra $(n+1)$-cocycle to obtain a Lie group
$(n+1)$-cocycle. In other words, for certain Lie groups $G$ with Lie algebra
$\g$, there is a cochain map:
\[ \smallint \maps C^\bullet(\g) \to C^\bullet(G) . \]
which is a chain homotopy inverse to differentiation.

When is this possible? We can always differentiate Lie group cochains to obtain
Lie algebra cochains, but if we can also integrate Lie algebra cochains to obtain
Lie group cochains, the cohomology of the Lie group and its Lie algebra
will coincide:
\[ H^\bullet(\g) \iso H^\bullet(G) . \]
By a theorem of van Est \cite{vanEst}, this happens when all the homology
groups of $G$, as a topological space, vanish.

Thus, we should look to Lie groups with vanishing homology for our examples.
How bad can things be when the Lie group is not homologically trivial? To get a
sense for this, recall that any semisimple Lie group $G$ is diffeomorphic to
the product of its maximal compact subgroup $K$ and a contractible space $C$:
\[ G \approx K \times C . \]
When $K$ is a point, $G$ is contractible, and certainly has vanishing homology.
At the other extreme, when $C$ is a point, $G$ is compact. And indeed, in this
case there is no hope of obtaining a nontrivial cochain map from Lie algebra
cochains to Lie group cochains:
\[ \smallint \maps C^\bullet(\g) \to C^\bullet(G) \]
because \emph{every smooth cochain on a compact group is trivial}.

Nonetheless, there is a large class of Lie $n$-algebras for which our Lie
$n$-groups \emph{are} general enough. In particular, when $G$ is an
`exponential' Lie group, the story is completely different. A Lie group or Lie
algebra is called \define{exponential} if the exponential map 
\[ \exp \maps \g \to G \]
is a diffeomorphism.  For instance, all simply-connected nilpotent Lie groups
are exponential, though the reverse is not true. Certainly, all exponential Lie
groups have vanishing homology, because $\g$ is contractible. We caution the
reader that some authors use the term `exponential' merely to indicate that
$\exp$ is surjective.

When $G$ is an exponential Lie group with Lie algebra $\g$, we can use a
geometric technique developed by Houard \cite{Houard} to construct a cochain
map:
\[ \smallint \maps C^\bullet(\g) \to C^\bullet(G). \]
The basic idea behind this construction is simple, a natural outgrowth of a
familiar concept from the cohomology of Lie algebras. Because a Lie algebra
$p$-cochain is a linear map:
\[ \omega \maps \Lambda^p \g \to \R, \]
using left translation, we can view $\omega$ as defining a $p$-form on the
Lie group $G$. So, we can integrate this $p$-form over $p$-simplices in $G$.
Thus we can define a smooth function:
\[ \smallint \omega \maps G^p \to \R, \]
by viewing the integral of $\omega$ as a function of the vertices of a
$p$-simplex:
\[ \smallint \omega(g_1, g_2, \dots, g_p) = \int_{[1, g_1, g_1 g_2, \dots, g_1 g_2 \cdots g_p]} \omega . \]
For the right-hand side to truly be a function of the $p$-tuple $(g_1, g_2,
\dots, g_p)$, we will need a standard way to `fill out' the $p$-simplex $[1,
g_1, g_1 g_2, \dots, g_1 g_2 \cdots g_p]$, based only on its vertices. It is
here that the fact that $G$ is exponential is key: in an exponential group, we
can use the exponential map to define a unique path from the identity $1$ to
any group element. We think of this path as giving a 1-simplex, $[1,g]$, and we
can extend this idea to higher dimensional $p$-simplices. 

Therefore, when $G$ is exponential, we can construct $\smallint$. Using this
cochain map, it is possible to integrate the slim Lie $n$-algebra
$\brane_\omega(\g)$ to the slim Lie $n$-group $\Brane_{\smallint \omega}(G)$. 

\begin{defn} \label{def:simplices}
Let $\Delta^p$ denote $\{(x_0, \dots, x_p) \in \R^{p+1} : \sum x_i = 1, x_i
\geq 0 \}$,  the standard $p$-simplex in $\R^{p+1}$.  Given a collection of
smooth maps
\[ \varphi_p \maps \Delta^p \times G^{p+1} \to G \]
for each $p \geq 0$, we say this collection defines a \define{left-invariant notion
of simplices} in $G$ if it satisfies:
\begin{enumerate}
	\item \textbf{The vertex property.} For any $(p+1)$-tuple, the
		restriction 
		\[ \varphi_p \maps \Delta^p \times \{(g_0, \dots, g_p)\} \to G \] 
		sends the vertices of $\Delta^p$ to $g_0, \dots, g_p$, in that
		order. We denote this restriction by 
		\[ [g_0, \dots, g_p]. \] 
		We call this map a \define{\boldmath{$p$}-simplex}, and regard
		it as a map from $\Delta^p$ to $G$.
	      
	\item \textbf{Left-invariance.} For any $p$-simplex $[g_0, \dots, g_p]$ and any $g \in G$, we
		have:
		\[ g [g_0, \dots, g_p] = [g g_0, \dots, g g_p ]. \]

	\item \textbf{The face property.} For any $p$-simplex
		\[ [g_0, \dots, g_p] \maps \Delta^p \to G \]
		the restriction to a face of $\Delta^p$ is a $(p-1)$-simplex.
\end{enumerate}

\end{defn}

As we noted above, every exponential Lie group can be equipped with a left-invariant
notion of simplices \cite{Huerta:susy3}. On any such group, we have the following
result:
\begin{prop} \label{prop:integral}
	Let $G$ be a Lie group equipped with a left-invariant notion of
	simplices, and let $\g$ be its Lie algebra. Then there is a cochain map
	from the Lie algebra cochain complex to the Lie group cochain complex
	\[ \smallint \maps C^{\bullet}(\g) \to C^{\bullet}(G) \]
	given by integration---that is, if $\omega$ is a left-invariant
	$p$-form on $G$, and $S$ is a $p$-simplex in $G$, then define:
	\[ (\smallint \omega)(S) = \int_S \omega. \]
\end{prop}
\begin{proof}
	See \cite[Prop. 13]{Huerta:susy3}. 
\end{proof}

\begin{prop}
	Let $G$ be a Lie group with Lie algebra $\g$. Then there is a cochain map
	from the Lie group cochain complex to the Lie algebra cochain complex:
	\[ D \maps C^{\bullet}(G) \to C^{\bullet}(\g) \]
	given by differentiation---that is, if $F$ is a homogeneous
	$p$-cochain on $G$, and $X_1, \dots, X_p \in \g$, then we can define:
	\[ DF(X_1, \dots, X_p) = \frac{1}{p!} \sum_{\sigma \in S_p} \mathrm{sgn}(\sigma) X_{\sigma(1)}^1 \dots X_{\sigma(p)}^p F(1, g_1, g_1 g_2, \dots, g_1 g_2 \dots g_p), \]
	where by $X_i^j$ we indicate that the operator $X_i$ differentiates
	only the $j$th variable, $g_j.$
\end{prop}
\begin{proof}
	See Houard \cite{Houard}, p.\ 224, Lemma 1.
\end{proof}

Having now defined cochain maps
\[ \smallint \maps C^{\bullet}(\g) \to C^{\bullet}(G) \]
and 
\[ D \maps C^{\bullet}(G) \to C^{\bullet}(\g), \]
the obvious next question is whether or not this defines a homotopy equivalence
of cochain complexes. Indeed, as proved by Houard, they do:
\begin{thm}
	Let $G$ be a Lie group equipped with a left-invariant notion of
	simplices, and $\g$ its Lie algebra. The cochain map
	\[ D \smallint \maps C^{\bullet}(\g) \to C^{\bullet}(\g), \]
	is the identity, whereas the cochain map
	\[ \smallint D \maps C^{\bullet}(G) \to C^{\bullet}(G) \]
	is cochain-homotopic to the identity. Therefore the Lie algebra cochain
	complex $C^\bullet(\g)$ and the Lie group cochain complex
	$C^\bullet(G)$ are homotopy equivalent and thus have isomorphic
	cohomology.
\end{thm}
\begin{proof}
	See Houard \cite{Houard}, p.\ 234, Proposition 2.
\end{proof}

We now generalize the above results from groups to supergroups. As above, our concern
will be for \define{exponetial supergroups}, where the exponential map
\[ \exp \maps \g \to G \]
from Lie superalgebra $\g$ to supergroup $G$ is a diffeomorphism.  Our first concern,
however, is to translate Lie superalgebra cocycles into Lie algebra cocycles. Note
that for any Lie superalgebra $\g$, we have Lie algebra on the $A_0$-module
$\g_A$. We call this an $A_0$-Lie algebra, because it is a Lie algebra over $A_0$.

\begin{prop}
	Let $\g$ be a Lie superalgebra, and let $\g_A$ be the $A_0$-Lie algebra
	of its $A$-points. Then there is a cochain map:
	\[ C^\bullet_0(\g) \to C^\bullet(\g_A) \]
	given by taking the even $p$-cochain $\omega$
	\[ \omega \maps \Lambda^p \g \to \R \]
	to the induced $A_0$-linear map $\omega_A$:
	\[ \omega_A \maps \Lambda^p \g_A \to A_0, \]
	where $\Lambda^p \g_A$ denotes the $p$th exterior power of $\g_A$ as
	an $A_0$-module.
\end{prop}
\begin{proof}
  See \cite[Prop. 23]{Huerta:susy3}.
\end{proof}
This proposition says that from any Lie superalgebra cocycle on $\n$ we obtain a Lie
algebra cocycle on $\n_A$, albeit now valued in $A_0$. Since $N_A$ is an exponential
Lie group with Lie algebra $\n_A$, we can integrate $\omega_A$ to a group cocycle,
$\smallint \omega_A$, on $N_A$.

As before, we need a notion of simplices in $N$. Since $N$ is a supermanifold,
the vertices of a simplex should not be points of $N$, but rather $A$-points
for arbitrary Grassmann algebras $A$. This means that for any $(p+1)$-tuple of
$A$-points, we want to get a $p$-simplex:
\[ [n_0, n_1, \dots, n_p] \maps \Delta^p \to N_A, \]
where, once again, $\Delta^p$ is the standard $p$-simplex in $\R^{p+1}$, and
this map is required to be smooth. But this only defines a $p$-simplex in
$N_A$. To really get our hands on a $p$-simplex in $N$, we need it to depend
functorially on the choice of Grassmann algebra $A$ we use to probe $N$. So if $f
\maps A \to B$ is a homomorphism between Grassmann algebras and $N_f \maps N_A \to
N_B$ is the induced map between $A$-points and $B$-points, we require:
\[ N_f \circ [n_0, n_1, \dots, n_p] = [ N_f(n_0), N_f(n_1), \dots, N_f(n_p) ] \]
Thus given a collection of maps:
\[ (\varphi_p)_A \maps \Delta^p \times (N_A)^{p+1} \to N_A \]
for all $A$ and $p \geq 0$, we say this collection defines a
\define{left-invariant notion of simplices} in $N$ if 
\begin{itemize}

	\item each $(\varphi_p)_A$ is smooth, and for each $x \in \Delta^p$,
		the restriction: 
		\[ (\varphi_p)_A \maps \left\{x\right\} \times N_A^{p+1} \to N_A \] 
		is $A_0$-smooth;

	\item it defines a left-invariant notion of simplices in $N_A$ for each
		$A$, as in Definition \ref{def:simplices}; 

	\item the following diagram commutes for all homomorphisms $f \maps A
		\to B$:
		\[ \xymatrix{
		\Delta^p \times N_A^{p+1} \ar[r]^>>>>>{(\varphi_p)_A} \ar[d]_{1 \times N_f^{p+1}}  & N_A \ar[d]^{N_f} \\
		\Delta^p \times N_B^{p+1} \ar[r]_>>>>>{(\varphi_p)_B} & N_B 
		}
		\]

\end{itemize}
In fact, every exponential supergroup can be equipped with a left-invariant notion of
simplices \cite{Huerta:susy3}. We can use this left-invariant notion of simplices to
define a cochain map $\smallint \maps C^\bullet(\n) \to C^\bullet(N)$:
\begin{prop} \label{prop:superintegrating}
	Let $\n$ be a nilpotent Lie superalgebra, and let $N$ be the
	exponential supergroup which integrates $\n$. There is a cochain map:
	\[ \smallint \maps C_0^\bullet(\n) \to C^\bullet(N) \]
	which sends the even Lie superalgebra $p$-cochain $\omega$ to the
	supergroup $p$-cochain $\smallint \omega$, given on $A$-points by:
	\[ (\smallint \omega)_A(n_1, \dots, n_p) = \int_{[1, n_1, n_1 n_2, \dots, n_1 n_2 \dots n_p ] } \omega_A \]
	for $n_1, \dots, n_p \in N_A$.
\end{prop}
\begin{proof}
	See \cite[Prop. 24]{Huerta:susy3}.
\end{proof}

\section{The super-2-brane Lie 3-supergroup} \label{sec:finale}

We are now ready to unveil the Lie 3-super\-group which integrates our favorite
Lie 3-super\-algebra, $\twobrane(n+2,1)$.  Remember, this is the Lie
3-superalgebra which occurs only in the dimensions for which the classical
2-brane makes sense. It is \emph{not} nilpotent, simply because the
Poincar\'e superalgebra $\siso(n+2,1)$ that forms degree 0 of
$\twobrane(n+2,1)$ is not nilpotent. Nonetheless, we are equipped to integrate
this Lie 3-superalgebra using only the tools we have built to perform this task
for nilpotent Lie $n$-superalgebras.

The road to this result has been a long one, and there is yet some ground to
cover before we are finished. So, let us take stock of our progress before we
move ahead:

\begin{itemize}
	\item In spacetime dimensions $n+3 = 4$, 5, 7 and 11, we used division
		algebras to construct a 4-cocycle $\beta$ on the
		supertranslation algebra:
		\[ \T = \V \oplus \S \]
		which is nonzero only when it eats two vectors and two spinors:
		\[ \beta(\A,\B,\Psi,\Phi) = \langle \Psi, (\A\B - \B\A) \Phi \rangle . \]

	\item Because $\beta$ is invariant under the action of $\so(n+2,1)$,
		it can be extended to a 3-cocycle on the Poincar\'e
		superalgebra:
		\[ \siso(n+2,1) = \so(n+2,1) \ltimes \T. \]
		The extension is just defined to vanish outside of $\T$, and we
		call it $\beta$ as well.

	\item Therefore, in spacetime dimensions $n+3$, we get a Lie
		3-superalgebra $\twobrane(n+2,1)$ by extending
		$\siso(n+2,1)$ by the 4-cocycle $\beta$.
\end{itemize}

In the last section, we built the technology necessary to integrate Lie
superalgebra cocycles to supergroup cocycles, \emph{provided} the Lie
superalgebra in question is nilpotent. This allows us to integrate nilpotent
Lie $n$-superalgebras to $n$-supergroups. But $\twobrane(n+2,1)$ is not
nilpotent, so we cannot use this directly here.

However, the cocycle $\beta$ is supported on a nilpotent subalgebra: the
supertranslation algebra, $\T$, for the appropriate dimension.  This saves the
day: we can integrate $\beta$ as a cocycle on $\T$.  This gives us a cocycle
$\smallint \beta$ on the supertranslation supergroup, $T$, for the appropriate
dimension. We will then be able to extend this cocycle to the Poincar\'e
supergroup, thanks to its invariance under Lorentz transformations.

The following proposition helps us to accomplish this, but takes its most
beautiful form when we work with `homogeneous supergroup cochains', which we
have not actually defined. Rest assured---they are exactly what you expect. If
$G$ is a supergroup that acts on the abelian supergroup $M$ by automorphism, a
\define{homogeneous $M$-valued $p$-cochain} on $G$ is a smooth map:
\[ F \maps G^{p+1} \to M \]
such that, for any Grassmann algebra $A$ and $A$-points $g, g_0, \dots, g_p \in
G_A$:
\[ F_A(gg_0, g g_1, \dots, g g_p) = g F_A(g_1, \dots, g_p) . \]
We can define the supergroup cohomology of $G$ using homogeneous or
inhomogeneous cochains, just as was the case with Lie group cohomology.

\begin{prop} Let $G$ and $H$ be Lie supergroups such that $G$ acts on $H$, and
	let $M$ be an abelian supergroup on which $G \ltimes H$ acts by
	automorphism.  Given a homogeneous $M$-valued $p$-cochain $F$ on $H$:
	\[ F \maps H^{p+1} \to M, \]
       	we can extend it to a map of supermanifolds:
	\[ \tilde{F} \maps (G \ltimes H)^{p+1} \to M \]
	by pulling back along the projection $(G \ltimes H)^{p+1} \to H^{p+1}$.
	In terms of $A$-points 
	\[ (g_0,h_0), \ldots, (g_p,h_p) \in G_A \ltimes H_A , \] 
	this means $\tilde{F}$ is defined by:
	\[ \tilde{F}_A((g_0,h_0), \ldots, (g_p,h_p)) = F_A(h_0, \ldots, h_p), \]
	Then $\tilde{F}$ is a homogeneous $p$-cochain on $G \ltimes H$ if and only if $F$ is
	$G$-equivariant, and in this case $d\tilde{F} = \widetilde{dF}$.
\end{prop}
\begin{proof}
	See \cite[Prop. 26]{Huerta:susy3}.
\end{proof}

\noindent 
Now, at long last, we are ready to integrate $\beta$. In the following
proposition, $T$ denotes the \define{supertranslation group}, the exponential
supergroup of the supertranslation algebra $\T$.

\begin{prop}
	In dimensions 4, 5, 7 and 11, the Lie supergroup 4-cocycle $\smallint
	\beta$ on the supertranslation group $T$ is invariant under the action
	of $\Spin(n+2,1)$.
\end{prop}

\noindent This is an immediate consequence of the following:

\begin{prop} 
	Let $H$ be a nilpotent Lie supergroup with Lie superalgebra $\h$.
	Assume $H$ is equipped with its standard left-invariant notion of 
	simplices, and let $G$ be a Lie supergroup that acts on $H$ by
	automorphism. If $\omega \in C^p(\h)$ is an even Lie superalgebra
	$p$-cochain which is invariant under the induced action of $G$ on $\h$,
	then $\smallint \omega \in C^p(H)$ is a Lie supergroup $p$-cochain
	which is invariant under the action of $G$ on $H$.
\end{prop}

\begin{proof}
	See \cite[Prop. 28]{Huerta:susy3}.
\end{proof}

It thus follows that in dimensions 4, 5, 7 and 11, the cocycle $\smallint
\beta$ can be extended to the Poincar\'e supergroup:
\[ \SISO(n+2,1) = \Spin(n+2,1) \ltimes T . \] 
By a slight abuse of notation, we continue to denote this extension by
$\smallint \beta$. As an immediate consequence, we have:

\begin{thm} \label{thm:2branegroup}
	In dimensions 4, 5, 7 and 11, there exists a slim Lie 3-supergroup
	formed by extending the Poincar\'e supergroup $\SISO(n+2,1)$ by the
	4-cocycle $\smallint \beta$, which we call the \define{2-brane Lie
	3-supergroup}, \define{\boldmath{$\Twobrane(n+2,1)$}}.
\end{thm}

\section*{Acknowledgements}

We thank Jim Dolan, Ezra Getzler, Nick Gurski, Hisham Sati and Christoph Wockel for
helpful conversations. We especially thank Urs Schreiber for freely sharing his
insight into higher gauge theory. We thank Roger Picken for his guidance and a
careful reading of this manuscript. We thank Michael Stay for allowing us to use his
beautiful TikZ diagrams in our definition of tricategory, more of which can be found
in his paper \cite{Stay}. Finally, we thank John Baez for his untiring help as he
supervised the thesis on which this paper is based. This work was partially supported
by FQXi grant RFP2-08-04, by the Geometry and Mathematical Physics Project
(EXCL/MAT-GEO/0222/2012), and by a postdoctoral grant from the Portuguese science
foundation, SFRH/BPD/92915/2013.

\end{document}